\documentclass[a4paper]{article}
\usepackage{amsmath,amsthm,amssymb}
\usepackage{graphicx,capt-of} 
\usepackage[textwidth=3cm,textsize=small]{todonotes}

\usepackage{caption}
\numberwithin{equation}{section}

\makeatletter

\let\c@table\c@figure
\makeatother

\numberwithin{figure}{section}
\numberwithin{table}{section}

\newcommand{\refpart}[1]{{\it (#1)}}  
\newcommand{\s}{\;\,}
\newcommand{\ff}{\sf}

\newtheorem{theorem}{Theorem}[section]
\newtheorem{lemma}[theorem]{Lemma}

\newtheorem{definition}[theorem]{Definition} 
 
\newtheorem{remark}[theorem]{Remark} 

\newcommand{\ZZ}{\mathbb{Z}}
\newcommand{\QQ}{\mathbb{Q}}

\title{\bf Conway's game ``Life" perturbed}

\author{Raimundas Vidunas, Arnas Vaicekauskas\\
\em Vilnius Univeristy, Institute of Applide Mathematics}

\newcommand{\expp}[1]{\mbox{\bf e}^{#1}}
\newcommand{\epsi}{\varepsilon}
\newcommand{\lifefont}{\sf}

\begin{document}

\date{\empty}

\maketitle

\begin{abstract}
A stochastic modification of Conway's cellular automaton ``Life'' is introduced here. 
Any cell could be perturbed spontaneously to the opposite (dead or alive) state at any iteration with a very low probability. 
This probability is assumed to be so low that perturbations affect most sensibly large patterns only in single cells
after they settle into stable, oscillating or moving configurations. This defines a Markov process on the set of stabilised patterns,
with unboundedly growing or overly large patterns represented by a general unspecified state of ``being huge''.
This stochastic model should approximate emergence of complexity
and live processes yet more interestingly than the original Conway's game.
This paper illustrates 
the proposed Markovian dynamics on the infinite ``Life'' grid with a limited set of most frequent patterns.
Concrete results are presented for this new game on small square toruses, of size up to $10\times 10$ cells. 
\end{abstract}
 
\section{Introduction}

John Conway's famed game ``Life'' is by far the best known cellular automaton.
It was first published in Martin Gardner's Mathematical Games column in the October 1970 issue of Scientific American \cite{sciam70}.
Its simple rules often create unforeseeably complicated evolution
of live and dead cells on a rectangular grid where each cell has eight neighbours.
In each generation a dead cell becomes alive if it has precisely three alive neighbouring cells, 
and a living cell remains alive if and only if it has two or three alive neighbours.

As typical for discrete deterministic dynamical systems, most small-scale initial configurations of live cells ``Life'' evolve
to a stable pattern (called {\em still-life} 
\cite[{\sf Still\_life}]{Wiki}) or a periodically repeating pattern (an {\em oscillator} \cite[{\sf Oscillator}]{Wiki}).
Some of the most common stabilised patterns are shown in Figure \ref{fig1}. 
More intriguing are moving patterns --- called {\em spaceships} \cite[{\sf Spaceship}]{Wiki}, including Glider in Figure \ref{fig1} --- 
and infinitely growing patterns: {\em glider guns, puffer trains, rakes, breeders}. 
The enticing appeal of Conway's cellular automaton lies in the complexity of behaviour that arises from its simple rules.
Often the evolution in ``Life" is unpredictable and chaotic for many iterations \cite{Adamatzky10}, \cite{preiter98}, \cite{socritic94}. 
Conway's game can model universal computation on Turing's machine \cite{rendell02}, \cite[Ch.~26]{Adamatzky10}.
In particular, ``Life'' can simulate itself on a larger scale \cite[{\sf OTCA\_metapixel}]{Wiki}. 
Avid research of Conway's game continues. It was proved recently that oscillators of any period exist \cite{periodic23}.
Since 2010 several spaceships travelling in oblique directions
--- including following the move of chess knight --- were found  \cite[{\sf Oblique\_spaceship}]{Wiki}.
Stochastic \cite{stochlife23} and continuous \cite{lenia20} versions of ``Life'' have been proposed.

\begin{figure}[tpb]
\[
\begin{picture}(382,200)(0,-10)
\multiput(-3,103)(10,0){5}{\line(0,1){54}}  
\multiput(-10,110)(0,10){5}{\line(1,0){54}} 
\multiput(58,102)(10,0){6}{\line(0,1){56}} 
\multiput(55,105)(0,10){6}{\line(1,0){56}} 
\multiput(129,102)(10,0){6}{\line(0,1){56}} 
\multiput(126,105)(0,10){6}{\line(1,0){56}} 
\multiput(197,102)(10,0){6}{\line(0,1){56}} 
\multiput(194,105)(0,10){6}{\line(1,0){56}} 
\multiput(265,102)(10,0){7}{\line(0,1){56}} 
\multiput(262,105)(0,10){6}{\line(1,0){66}} 
\multiput(343,102)(10,0){6}{\line(0,1){56}} 
\multiput(340,105)(0,10){6}{\line(1,0){56}} 
\multiput(4,2)(10,0){7}{\line(0,1){66}}  
\multiput(1,5)(0,10){7}{\line(1,0){66}} 
\multiput(86,2)(10,0){7}{\line(0,1){66}}  
\multiput(83,5)(0,10){7}{\line(1,0){66}} 
\multiput(168,8)(10,0){7}{\line(0,1){54}}  
\multiput(165,15)(0,10){5}{\line(1,0){66}} 
\multiput(336,3)(10,0){6}{\line(0,1){64}}  
\multiput(329,10)(0,10){6}{\line(1,0){64}} 
\multiput(250,2)(10,0){7}{\line(0,1){66}}  
\multiput(247,5)(0,10){7}{\line(1,0){66}} 
 \thicklines
\put(12,125){\circle{8}} \put(12,135){\circle{8}} 
\put(22,135){\circle{8}} \put(22,125){\circle{8}} 
\put(5,90){\lifefont Block}  
\put(73,120){\circle{8}} \put(83,120){\circle{8}} 
\put(73,130){\circle{8}} \put(83,140){\circle{8}} 
\put(93,130){\circle{8}} \put(74,90){\lifefont Boat}  
\put(144,120){\circle{8}} \put(154,120){\circle{8}} 
\put(144,130){\circle{8}} \put(154,140){\circle{8}} 
\put(164,130){\circle{8}} \put(164,140){\circle{8}} 
\put(145,90){\lifefont Ship}  \put(279,90){\lifefont Beehive} 
 \put(351,90){\lifefont Blinker}
\multiput(358,130)(10,0){3}{\circle{8}} 
\put(212,130){\circle{8}} \put(222,120){\circle{8}} 
\put(232,130){\circle{8}} \put(222,140){\circle{8}} 
\put(290,120){\circle{8}} \put(300,120){\circle{8}} 
\put(280,130){\circle{8}} \put(290,140){\circle{8}} 
\put(310,130){\circle{8}} \put(300,140){\circle{8}} 
\put(213,90){\lifefont Tub}   \put(104,-10){\lifefont Pond}
\put(19,30){\circle{8}}  \put(19,40){\circle{8}}  
\put(29,20){\circle{8}}  \put(29,50){\circle{8}} 
\put(39,20){\circle{8}}  \put(39,40){\circle{8}} 
\put(49,30){\circle{8}}  \put(24,-10){\lifefont Loaf}  
\put(101,30){\circle{8}}  \put(101,40){\circle{8}} 
\put(111,20){\circle{8}}  \put(111,50){\circle{8}} 
\put(121,20){\circle{8}}  \put(121,50){\circle{8}} 
\put(131,30){\circle{8}}  \put(131,40){\circle{8}} 
\put(183,30){\circle{8}} \put(183,40){\circle{8}} 
\put(193,30){\circle{8}} \put(203,40){\circle{8}} 
\put(213,30){\circle{8}} \put(213,40){\circle{8}} 
\put(185,-5){\lifefont Snake}
\put(265,20){\circle{8}} \put(265,30){\circle{8}} 
\put(275,20){\circle{8}} \put(275,40){\circle{8}} 
\put(285,40){\circle{8}} \put(295,40){\circle{8}} 
\put(295,50){\circle{8}} \put(269,-10){\lifefont Eater}
\multiput(351,25)(10,0){3}{\circle{8}} \put(348,-10){\lifefont Glider}
\put(371,35){\circle{8}}  \put(361,45){\circle{8}}  
\end{picture}
\]
\caption{Some of the most common stabilised patterns in ``Life", including Blinker (an oscillator of period 2) and Glider (a moving spaceship).} 
\label{fig1}
\end{figure}
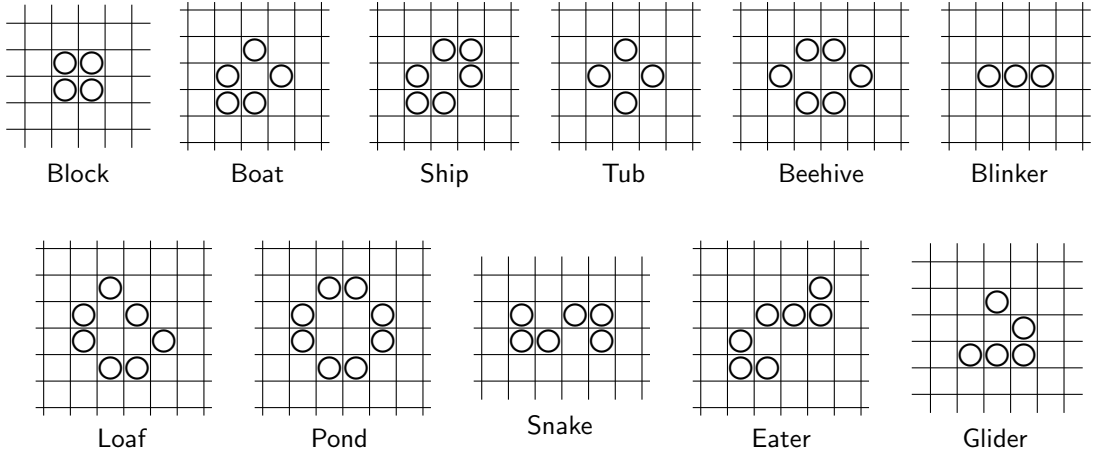

This paper proposes an unorthodox stochastic modification of Conway's game ``Life". 
Each cell is still either dead or alive at each moment,
but it could be perturbed spontaneously to the opposite state at any iteration with a very low probability $\epsi$. 
This probability is assumed to be so low that it occurs actually only when most patterns are stabilised
into a still life or an oscillating or moving configuration --- and just one cell is perturbed. 
Infinitely growing patterns would be perturbed in the process of their growth,
but it could be assumed that they would not would never return to a delimited configuration.
Just as the empty space state would revert back to itself by the assumed regime of perturbations,
vastly grown patterns would remain in the unspecific state of {\em Being Huge}.
We may include any unbounded ``soup'' of live cells and (more or less) familiar patterns into this state.

The assumed regime of perturbations defines a Markov process on the set of stabilised patterns
--- of limited size, in practice --- with the empty space and Being Huge as two absorbing states.
We choose to identify the patterns up to the translation and dihedral $D_4$ symmetries of the rectangular grid of cells.
To keep the size of the Markov chain manageable, we may fix a finite set of $N_0$ patterns 
(including the empty space) as not Being Huge.

We choose a continuous rather than discrete model of the Markov process, where the number of cells
whose perturbation leads to any other stabilised pattern (or an absorbing state) translates directly 
to the decay rate to that other stabilised pattern or state. 
The Markov process is represented by a linear matrix differential equation 
\begin{equation} \label{eq:diffs}
\frac{d}{dt}\vec{y}(t) = M\,\vec{y}(t),
\end{equation}
where $\vec{y}(t)$ is the evolving probability distribution vector over the possible patterns and states,
and $M$ is the transition matrix of decay rates. Its diagonal entries are negative numbers that balance
the sum of decay rates to the other states, so that the column sums are zero.
The size of $M$ is $(N_0+1)\times(N_0+1)$. %
Most often $M$ has distinct eigenvalues. Then \cite[\S 3.2.2]{Matrix85} the general solution is given by
\begin{equation} \label{eq:m0gensol}
\vec{y}(t)=\sum_{k=0}^{N_0} C_k\,\vec{y}_k\,\expp{\lambda_k t},
\end{equation}
where $\lambda_k$ are the eigenvalues, $y_k$ are the corresponding eigenvectors,
and $C_k$ are free constants. (We may assign the number $k=0$ to one of the absorbing states.)
The two absorbing states give two elementary eigenvectors with the eigenvalue $0$.
Other eigenvalues would have negative real part, 
but there might be other independent eigenvectors with the eigenvalue $0$.
That would mean that there is an {\em immortal} set of non-empty patterns that decay only to each other.
In particular, they do not die out to empty space, nor blow up to Being Huge.
A fascinating instance would be a single non-empty pattern that does not decay to anything else.
Existence of immortal patterns or sets is the foremost challenge of the proposed modification of Conway's ``Life''.

We will not go far with perturbing 
``Life'' on the infinite grid. Section \ref{sec:infinite} investigates the outlined model 
with perturbations of merely $N_0=10$ patterns 
as a principled demonstration.The perturbed game can be analysed more thoroughly
after making the grid space finite. Most customarily, a finite rectangular region is wrapped to a {\em torus}
in the standard topological way. Section \ref{sec:small} investigates the perturbations of ``Life'' 
on the $5\times 5$, $6\times 6$ and $7\times 7$ toruses, 
while Section \ref{sec:large} examines the $8\times 8$, $9\times 9$ and $10\times 10$ toruses. 
We count the patterns obtainable by consecutive perturbations from the simple Block pattern,
and compute the leading eigenvalues and eigenvectors of the Markovian processes.
As the number of patterns on toruses is finite, there is no Being Huge state. 
The empty state is still an absorbing state, but no other immortal patterns or sets were found.
Nevertheless, the non-zero complex eigenvalue closest to $0$ is significant.
Its eigenvector defines the asymptotic probability distribution of non-empty patterns
while the irreversible decay to the empty state did not take place. 

Section \ref{sec:compute} described our algorithmic methods and encountered computational issues.
The appendix section provides supplementary information about the perturbed ``Life'' on the considered toruses,
mostly in graphical form or in tables. 

Section \ref{sec:interpret} offers philosophical considerations for interpreting the stochastic 
model and arising of complexity in general. We suggest that stochastic modelling is more relevant to 
organic processes than assumed deterministic laws, as the character of organic and complex phenomena 
is defined by statistical tendencies and resilience rather than particular underlying dynamics.
We argue that emergence of particular forms of complexity 
from promising dynamical rules must be a G\"odelian or (practically) uncomputable problem in general, 
meaning that it is systematically undecidable without stumbling upon affirmative instances 
or fairly exhaustive (infinite or combinatorially exploding) search. We compare the leading eigenvectors 
(still dominated by the zero eigenvalue of the empty space) of our models 
to Friston's {\em Free Energy Principle} \cite{Friston10}, \cite{SchrodingerFE}
that seeks to explain cognitive and organic processes in terms of minimising experiential surprise.


\section{Perturbed ``Life" on the infinite space}
\label{sec:infinite}

Analysis on an infinite board will always be limited, because complexity of stabilisation of perturbed patterns 
is out of control even after perturbation of the most simple and frequent patterns.
But we venture two simple models, not so much to offer substantial results for the main version of Conway's ``Life'',
but to introduce our routine and key observations. 


To set a grounded ambience, let us assume that the very small possibility is $\epsi=10^{-120}$ briefly. 
If the size of a cell is of the Planck length $l_P\approx 1.6\cdot 10^{-35}$ meters
along the sides, and a generation lasts $t_P\approx 5.4\cdot 10^{-44}$ seconds,  
there would be a perturbation every $l_P^2t_P/\epsi\approx 1.4\cdot 10^7$ seconds (about 160 days) per square meter.
The diameter of Milky way is $d_{MW}\approx 8.27\cdot 10^{20}$ meters.
Hence a perturbation would happen in Milky Way every $4l_P^2/(\pi d_{MW}^2)\approx 477\cdot 10^{6}$ 
generations of ``Life''. On the other hand, we will encounter probabilities of the order $10^{-131}$ 
for occurrence of some patterns in our models; see part \refpart{c}, the last column in summary Table \ref{tb:summary}.

\subsection{A starting model}
\label{sec:start}

The first eight patterns in Figure \ref{fig1} and Glider are the only local patterns of the game ``Life''
that occur with the relative frequency better than 1 in 100; 
see \cite[\ff List\s{}of\s{}common\s{}still\s{}lifes, Frequency\s{}class]{Wiki}.
These most frequent configurations comprise 99\% of connected objects in {\em ashes} 
of random patterns \cite[\ff Soup\#Ash, Natural]{Wiki}.
Let us consider the minimal perturbations of these nine 
patterns. Stabilized outcomes of these perturbations are graphically depicted in Figure \ref{fg:inflife}.
The resulting stabilized patterns are shown 
in the corresponding cells by single letter labels. 
The transitions between the same nine patterns are depicted additionally by arrows and transition frequencies. 
The other transitions are explained in the figure caption; they mostly refer to Appendix Figure \ref{fg:pipatterns}.

\begin{figure}
\centering
\begin{picture}(272,430)(50,-5)
\input{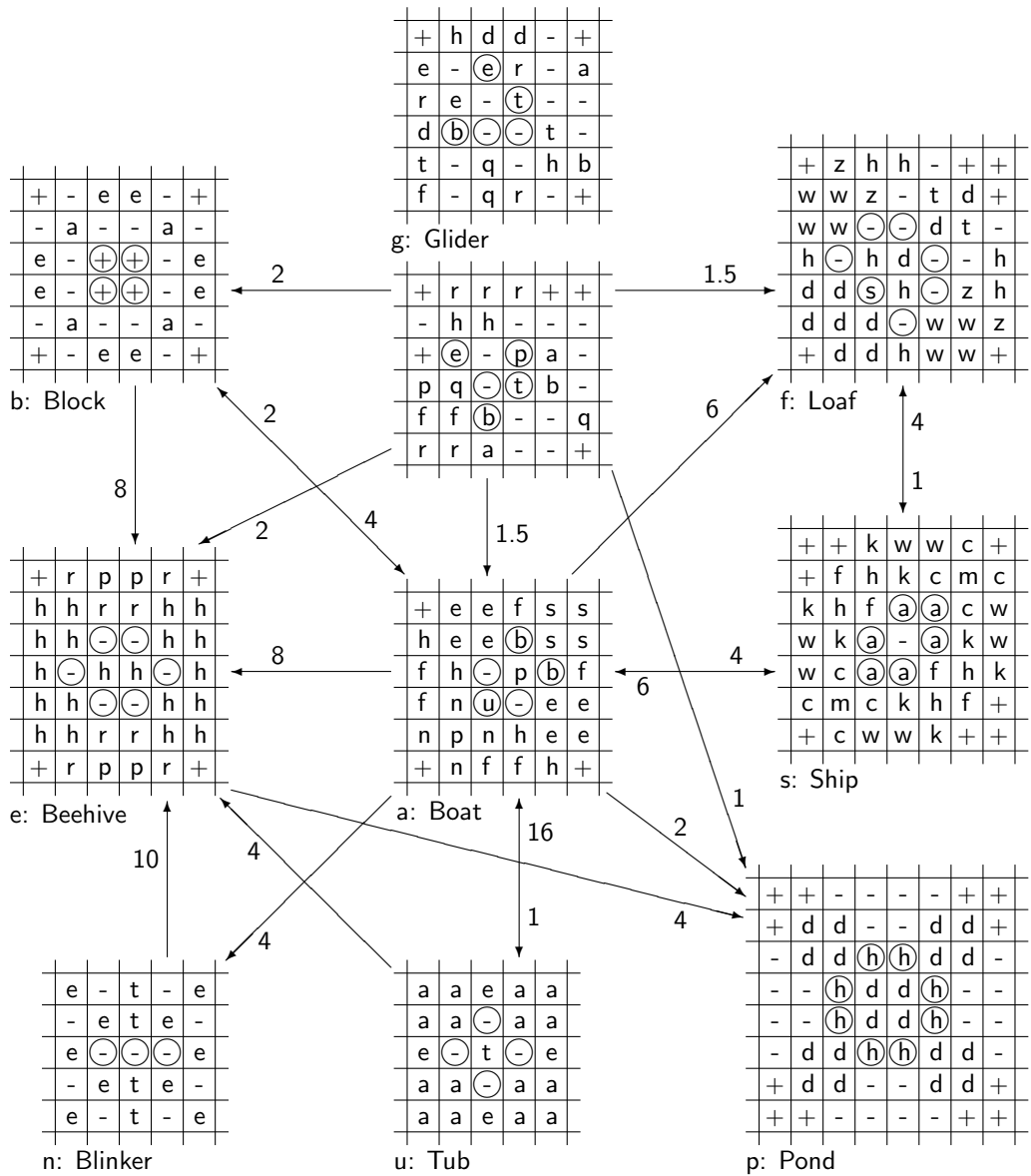}
\end{picture}
\caption{Transitions between the most common stabilized patterns. 
The transitions {\sf +, --} are to itself or the empty space, respectively.
The transitions {\sf c, d, h, k, m, q, t, z} result in the patterns of Appendix Figure \ref{fg:pipatterns}. 
Also, {\sf r} denotes the evolution sequence of R-pentamino of Figure \ref{fig1}\refpart{c}; 
{\sf w} denotes the transition to a block and a moving glider (of the Wing sequence \cite[\sf Wing]{Wiki});
and {\sf x} denotes the transition to the empty space.}
\label{fg:inflife}
\end{figure}

\begin{figure} \centering
\begin{picture}(306,404)(0,-24)
\put(48,78){\circle*{5}}  \put(48,84){\circle*{5}}
\put(48,90){\circle*{5}}  \put(60,66){\circle*{5}}
\put(66,66){\circle*{5}}  \put(72,66){\circle*{5}}
\put(78,0){\circle*{5}}  \put(78,6){\circle*{5}}
\put(78,150){\circle*{5}}  \put(78,156){\circle*{5}}
\put(84,-6){\circle*{5}}  \put(84,12){\circle*{5}}
\put(84,150){\circle*{5}}  \put(84,156){\circle*{5}}
\put(90,-6){\circle*{5}}  \put(90,12){\circle*{5}}
\put(96,0){\circle*{5}}  \put(96,6){\circle*{5}}
\put(102,138){\circle*{5}}  \put(102,144){\circle*{5}}
\put(108,138){\circle*{5}}  \put(108,144){\circle*{5}}
\put(150,0){\circle*{5}}  \put(150,6){\circle*{5}}
\put(156,0){\circle*{5}}  \put(156,6){\circle*{5}}
\put(174,168){\circle*{5}}  \put(180,168){\circle*{5}}
\put(186,168){\circle*{5}}  \put(204,0){\circle*{5}}
\put(204,6){\circle*{5}}  \put(210,0){\circle*{5}}
\put(210,6){\circle*{5}}  \put(252,138){\circle*{5}}
\put(252,144){\circle*{5}}  \put(258,138){\circle*{5}}
\put(258,144){\circle*{5}}  \put(264,0){\circle*{5}}
\put(264,6){\circle*{5}}  \put(270,-6){\circle*{5}}
\put(270,12){\circle*{5}}  \put(276,-6){\circle*{5}}
\put(276,12){\circle*{5}}  \put(276,150){\circle*{5}}
\put(276,156){\circle*{5}}  \put(282,0){\circle*{5}}
\put(282,6){\circle*{5}}  \put(282,150){\circle*{5}}
\put(282,156){\circle*{5}}  \put(288,66){\circle*{5}}
\put(294,66){\circle*{5}}  \put(300,66){\circle*{5}}
\put(312,78){\circle*{5}}  \put(312,84){\circle*{5}}
\put(312,90){\circle*{5}}  \multiput(45,-12)(30,0){10}{\line(0,1){186}}
\multiput(42,-9)(0,30){7}{\line(1,0){276}}
\put(43,-21){\lifefont q: $\pi$-heptamino ash}
\put(-23,51){\circle*{5}}  \put(-17,45){\circle*{5}}
\put(-17,57){\circle*{5}}  \put(-11,45){\circle*{5}}
\put(-11,57){\circle*{5}}  \put(-5,51){\circle*{5}}
\put(13,15){\circle*{5}}  \put(13,21){\circle*{5}}
\put(19,9){\circle*{5}}  \put(19,27){\circle*{5}}
\put(25,15){\circle*{5}}  \put(25,21){\circle*{5}}
\multiput(-26,3)(30,0){3}{\line(0,1){66}}
\multiput(-29,6)(0,30){3}{\line(1,0){66}}
\put(-30,-6){\lifefont d:\,teardrop bees}
\put(-23,128){\circle*{5}}  \put(-23,134){\circle*{5}}
\put(-17,122){\circle*{5}}  \put(-17,134){\circle*{5}}
\put(-11,122){\circle*{5}}  \put(-11,128){\circle*{5}}
\put(-5,110){\circle*{5}}  \put(-5,116){\circle*{5}}
\put(1,104){\circle*{5}}  \put(1,116){\circle*{5}}
\put(1,152){\circle*{5}}  \put(1,158){\circle*{5}}
\put(7,104){\circle*{5}}  \put(7,110){\circle*{5}}
\put(7,146){\circle*{5}}  \put(7,158){\circle*{5}}
\put(13,146){\circle*{5}}  \put(13,152){\circle*{5}}
\put(19,134){\circle*{5}}  \put(19,140){\circle*{5}}
\put(25,128){\circle*{5}}  \put(25,140){\circle*{5}}
\put(31,128){\circle*{5}}  \put(31,134){\circle*{5}}
\multiput(-26,98)(30,0){3}{\line(0,1){66}}
\multiput(-29,101)(0,30){3}{\line(1,0){66}}
\put(-29,89){\lifefont m: Fleet}
\put(-11,247){\circle*{5}}  \put(-11,253){\circle*{5}}
\put(-5,223){\circle*{5}}  \put(-5,229){\circle*{5}}
\put(-5,247){\circle*{5}}  \put(-5,253){\circle*{5}}
\put(1,223){\circle*{5}}  \put(1,229){\circle*{5}}
\put(109,223){\circle*{5}}  \put(109,229){\circle*{5}}
\put(115,199){\circle*{5}}  \put(115,205){\circle*{5}}
\put(115,223){\circle*{5}}  \put(115,229){\circle*{5}}
\put(121,199){\circle*{5}}  \put(121,205){\circle*{5}}
\multiput(-20,193)(30,0){6}{\line(0,1){66}}
\multiput(-23,196)(0,30){3}{\line(1,0){156}}
\put(-23,184){\lifefont k: Blockade}
\put(163,205){\circle*{5}}  \put(163,211){\circle*{5}}
\put(169,205){\circle*{5}}  \put(169,211){\circle*{5}}
\put(211,247){\circle*{5}}  \put(217,241){\circle*{5}}
\put(217,253){\circle*{5}}  \put(223,241){\circle*{5}}
\put(223,253){\circle*{5}}  \put(229,247){\circle*{5}}
\put(241,247){\circle*{5}}  \put(247,241){\circle*{5}}
\put(247,253){\circle*{5}}  \put(253,241){\circle*{5}}
\put(253,253){\circle*{5}}  \put(259,247){\circle*{5}}
\put(301,205){\circle*{5}}  \put(301,211){\circle*{5}}
\put(307,205){\circle*{5}}  \put(307,211){\circle*{5}}
\multiput(160,193)(30,0){6}{\line(0,1){66}}
\multiput(157,196)(0,30){3}{\line(1,0){156}}
\put(157,184){\lifefont z: bees under guard}
\put(205,368){\circle*{5}}  \put(211,368){\circle*{5}}
\put(205,362){\circle*{5}}  \put(211,362){\circle*{5}}
\put(205,332){\circle*{5}}  \put(211,332){\circle*{5}}
\put(205,326){\circle*{5}}  \put(211,326){\circle*{5}}
\put(205,302){\circle*{5}}  \put(211,302){\circle*{5}}
\put(205,296){\circle*{5}}  \put(211,296){\circle*{5}}
\put(301,296){\circle*{5}}  \put(301,290){\circle*{5}}
\put(301,284){\circle*{5}}  \multiput(196,278)(30,0){5}{\line(0,1){96}}
\multiput(193,281)(0,30){4}{\line(1,0){126}}
\put(193,269){\lifefont c: century ash}
\put(19,290){\circle*{5}}  \put(13,296){\circle*{5}}
\put(25,296){\circle*{5}}  \put(13,302){\circle*{5}}
\put(25,302){\circle*{5}}  \put(19,308){\circle*{5}}
\put(-11,320){\circle*{5}}  \put(-5,320){\circle*{5}}
\put(43,320){\circle*{5}}  \put(49,320){\circle*{5}}
\put(-17,326){\circle*{5}}  \put(1,326){\circle*{5}}
\put(37,326){\circle*{5}}  \put(55,326){\circle*{5}}
\put(-11,332){\circle*{5}}  \put(-5,332){\circle*{5}}
\put(43,332){\circle*{5}}  \put(49,332){\circle*{5}}
\put(19,344){\circle*{5}}  \put(13,350){\circle*{5}}
\put(25,350){\circle*{5}}  \put(13,356){\circle*{5}}
\put(25,356){\circle*{5}}  \put(19,362){\circle*{5}}
\multiput(-26,278)(30,0){4}{\line(0,1){96}}
\multiput(-29,281)(0,30){4}{\line(1,0){96}}
\put(-29,269){\lifefont h: Honey farm}
\put(121,305){\circle*{5}}  \put(127,305){\circle*{5}}
\put(133,305){\circle*{5}}  \put(109,317){\circle*{5}}
\put(145,317){\circle*{5}}  \put(109,323){\circle*{5}}
\put(145,323){\circle*{5}}  \put(109,329){\circle*{5}}
\put(145,329){\circle*{5}}  \put(121,341){\circle*{5}}
\put(127,341){\circle*{5}}  \put(133,341){\circle*{5}}
\multiput(100,293)(30,0){3}{\line(0,1){66}}
\multiput(97,296)(0,30){3}{\line(1,0){66}}
\put(97,284){\lifefont t: Traffic light}
\end{picture}
\caption{More expansive patterns resulting from perturbation of the basic patterns in Figure \ref{fg:inflife}. 
The displayed grids are of size $5\times5$ in cells.
We choose to capitalize their accepted names \cite{Wiki},
while the patterns without direct analogues in Conway's game are named provisionally and without capitalization.}
\label{fg:pipatterns}
\vspace{22pt} \small
\!\begin{tabular}{@{}ccrrrrr@{}}
\hline
id& \multicolumn{1}{c}{Name}  & \multicolumn{5}{l@{}}{Eigenvectors\,\dotfill} \\
&& \multicolumn{1}{r}{the 1st} & $\lambda\!=\!-25$ & \multicolumn{1}{r}{the 3rd} 
& $\lambda\!=\!-32$ & $\lambda\!=\!-38$
\\ \hline
\ff g & Glider & 0 & 0 & 0 & $90560$ & 0 \\  
\ff b & Block & 5.8042242497 & 0 & 47.6574170153 & $-36656$ & 0 \\ 
\ff a & Boat & 15.1179236341 & 0 & 17.5583327381 & $-17248$ & 0 \\ 
\ff u & Tub & 6.8429092056 & 0 & $-7.7583815523$ & $2464$ & 0 \\ 
\ff s & Ship & 3.6066196521 & 0 & 5.7056779406 & $-7648$ & 0 \\ 
\ff f & Loaf &  5.2022637278 & 0 & 8.1447467127 & $160$ & 0 \\ 
\ff n & Blinker & 27.3716368225 & 8775 & $-31.0335262092$ & $9856$ & 0 \\ 
\ff e & Beehive & 30.8012669084 & 6750 & $16.7979532896$ & $-23616$ & 133 \\ 
\ff p & Pond & 5.2531557998 & 1000 & 4.1358723037 & $-1920$ & 38 \\ 
\hline ... & bang & $-64.1048731311$ & $-10084$ & $-35.0762994498$ & $-929$ & $-126$ \\ 
\ff x & \!empty space\! & $-35.8951268689$ & $-6441$ & $-26.1317927887$ & $-15023$ & $-45$ \\  
\hline
\end{tabular} 
\captionof{table}{Eigenvectors of $M_0$.} 
\label{tb:infm0}
\end{figure}

Let us consider the continuous time dynamics of the transitions between the nine patterns.
Let $\vec{y}$ denote the 9-dimensional vector, representing the probabilities of the nine configurations
in the labelled order {\ff g}, {\ff b}, {\ff a}, {\ff u}, {\ff s}, {\ff f}, {\ff n}, {\ff e}, {\ff p}. 
Let us consider the continuous time model (\ref{eq:diffs}) with $M=M_0$ being
the transition matrix between the nine patterns:
\begin{equation}
M_0 = \left( \begin{array}{ccccccccc}
-32 & 0 & 0 & 0 & 0 & 0 & 0 & 0 & 0  \\
2 & -28 & 2 & 0 & 0 & 0 & 0 & 0 & 0  \\
3/2 & 4 & -33 & 16 & 6 & 0 & 0 & 0 & 0  \\
0 & 0 & 1 & -25 & 0 & 0 & 0 & 0 & 0  \\
0 & 0 & 4 & 0 & -41 & 1 & 0 & 0 & 0 \\
3/2 & 0 & 6 & 0 & 4 & -43 & 0 & 0 & 0 \\
0 & 0 & 4 & 0 & 0 & 0 & -25 & 0 & 0  \\
2 & 8 & 8 & 4 & 0 & 0 & 10 & -38 & 0 \\
1 & 0 & 2 & 0 & 0 & 0 & 0 & 4 & -52 
\end{array} \right),
\end{equation}
following Figure \ref{fg:inflife}. 
The diagonal numbers represent the total decay rates away from each configuration. The column sums are negative, 
representing the fact that decay to other patterns or the empty space are always possible.
Perturbations of Glider are assumed to happen with the same probability in any phase,
hence its decay rates are arithmetic averages over the two actually different phases for each decay product.
The same averaging principle will apply to any oscillator or moving pattern.
Blinker is essentially the same pattern in both oscillation phases, hence averaging its decay rates is a nominal triviality.

The differential system (\ref{eq:diffs}) with the constant matrix $M=M_0$ is straightforward to solve.
The solution is similar to (\ref{eq:m0gensol}), as multiple eigenvalues are not encountered.
A peripheral distinction is that the absorbing states (of the empty space and Being Huge) are not included,
and all eigenvalues are negative as $M_0$ is strictly diagonally dominant \cite[Corollary 6.2.27\refpart{a}]{Matrix85}. 
The eigenvalues are, approximately,
\begin{align} \label{eq:m0evs}
& -22.7907169042,\ -25,\ -27.2631437523,\ -32,\ -32.9850657973,  \quad\nonumber \\
& -38,\  -43.1843156322,\ -43.7767579139,\ -52. 
\end{align}
The non-integer eigenvalues are algebraic numbers of degree 5. They can be expressed as
$\lambda=\xi-34$, where $\xi$ satisfies
\[
\xi^5 - 176\xi^3 - 2\xi^2 + 6963\xi - 6882=0.
\]
The integer eigenvalues reflect presence of irreversible transitions in Figure \ref{fg:inflife},
which result in invariant subspaces of $M_0$ or $M_0^T$. In particular:
\begin{itemize}
\item The transitions to Pond are permanently irreversible in this model. 
This gives the eigenvalue $-52$ (i.e., the self-decay rate of Pond), 
and an invariant 1-dimensional subspace generated 
by the corresponding eigenvector $(0,\ldots,0,1)^T$.
\item Transitions to Beehive are irreversible as well, as it decays (within Figure \ref{fg:inflife}) 
only to the absorbing Pond. This gives the eigenvalue $-38$, 
equal to the self-decay rate of Beehive.
\item Transitions to Blinker are irreversible as well, as it decays only to the just considered Beehive.
This gives the eigenvalue $-25$, 
equal to the self-decay rate of Blinker.
\item Nothing in Figure \ref{fg:inflife} decays to Glider. This gives the eigenvalue $-32$,
as the transposed matrix $M_0^T$ has the corresponding eigenvector $(1,0,\ldots,0)$.
\end{itemize}
Presence of 
irreversible transitions means the transition matrix $M_0$ has a block-lower-triangular structure:
there are four trivial blocks (of size $1\times 1$) giving the integer eigenvalues,
and a block of size $5\times 5$ in the middle.

Some of the eigenvectors are given in Table \ref{tb:infm0}; ignore the last two rows for now.
The most important eigenvector is the one corresponding to the largest eigenvalue $-22.7907\ldots$ in (\ref{eq:m0evs}).
It determines the term in (\ref{eq:m0gensol}) with the lowest decay rate, 
and gives the asymptotic probability distribution among the nine patterns 
while the decay to other patterns or the empty space is delayed. 
The dominant vector displayed in the third column 
shows directly the percentages.


\begin{table} \small
\centering
\begin{tabular}{@{}ccrrrr@{}}
\hline
id& \multicolumn{1}{c}{Name}  & \multicolumn{4}{l@{}}{Eigenvectors\,\dotfill} \\
&& \multicolumn{1}{r}{the 1st} & $\lambda\!=\!-25$ & \multicolumn{1}{r}{the 3rd} 
& $\lambda\!=\!-52$ 
\\ \hline
\ff g & Glider & 27.3962182674 & 15336000 & 24.1715471046 & 0 \\ 
\ff b & Block & 16.8212683665 & 13896000 & 57.1172063342 & 0 \\ 
\ff a & Boat & 11.9054362841 & 0 & 8.8655336409 & 0 \\ 
\ff u & Tub & 2.2273920990 & $-4942125$ & $-13.4848830450$ & 0 \\
\ff s & Ship & 2.4767515229 & 81000 & 2.6879457153 & 0 \\
\ff f & Loaf &  5.2445444901 & 1296000 & 5.7778274448 & 0 \\ 
\ff n & Blinker & 8.9095683961 & $-6520950$ & $-53.9395321800$ & 0 \\ 
\ff e & Beehive & 20.8564233953 & 4374000 & $-1.3878679364$ & 0 \\ 
\ff p & Pond & 4.1623971784 & 1216000 & 1.3799397602 & $13$  \\ 
... & bang & $-64.4991670278$ & $-13083043$ & $-10.5110561062$ & $-7$ \\ 
\ff x & empty space & $-52.7079978383$ & $-15946962$ & $-26.6519028443$ & $-6$ \\  
\hline
\end{tabular}
\caption{Eigenvectors of $M_2$.}  \label{tb:infm1}
\end{table}

\subsection{Enhanced modeling}

A proper Markov model is obtained by amending
the transitions to empty space, and 
the transitions to the other configurations (as too complex to analyze further).
The placeholder state for complex patterns could be labelled as the bang state,
which is a modest name within our small model in comparison to the Being Huge state 
described in the introduction, 
or the cosmological Big Bang.
This update leads to augmenting the matrix $M_0$ by:
\begin{itemize} 
\item two rows counting these two kinds of transitions for each of the nine patterns of Figure \ref{fg:inflife},
namely 
\begin{equation}
 \left( \begin{array}{ccccccccc}
11 & 0 & 4 & 1 & 30 & 32 & 4 & 28 & 28 \\
13 & 16 & 2 & 4 & 1 & 10 & 11 & 6 & 24 
\end{array} \right).
\end{equation}
\item two zero columns, as the empty space and the bang state would be two {\em absorbing states}; 
presumably there would be no further ``noticeable" decay.
\end{itemize}
The column sums of the augmented matrix would equal 0, as fitting for a continuous time Markov process.
Then $\lambda=0$ would be a double eigenvalue, with the corresponding eigenvectors supported 
(i.e., having non-zero components) only for the two newly introduced states.

Apart from the new eigenspace for $\lambda=0$,  
the eigenvectors of the initial model are adjusted by two new components; see the bottom two rows in Table \ref{tb:infm0}.
The pairs of negative components in these rows sum up to $-100$ by the lemma below,
giving the asymptotic probability distributions between the bang and the empty space.
\begin{lemma} \label{th:evortho}
Components of these eigenvectors sum up to zero, because of the orthogonality to the eigenvector 
$(1,1,\ldots,1)$ of the transposed matrix . 
The corresponding eigenvector $\vec{w}_0$ of the transposed matrix has all its entries equal to 1,
as the column sums of $M_5$ are zero. As well-known, the eigenvectors other than $\vec{v}_0$ must be orthogonal.
\end{lemma}
\begin{proof}
It is well-known that \cite[Theorem 1.4.7\refpart{a}]{Matrix85} that 
this vector is orthogonal to $\vec{v}_1,\ldots,\vec{v}_9$.
Indeed, $0=\lambda_0\vec{w}_0\cdot \vec{v}_k=\vec{w}_0\cdot M_5\vec{v}_k =\lambda_k\,\vec{w}_0\cdot \vec{v}_k$.
\end{proof}

We immediately offer the following additional modifications to the stochastic model:
\begin{itemize}
\item As the {\ff w}-decay \cite[\ff Wing]{Wiki} 
generates a block and a glider that is moving away, we consider this decay as generation of two separate patterns. 
This is consistent without our assumption of the low probability of decays, so that the glider will move undeterminably far away 
until either it or the block would decay.
\item Similarly, the {\ff r}-decay \cite[\ff R-pentomino]{Wiki} generates an extensive ash, a glider in one direction, 
and two packs of (two or three) gliders in other diagonal directions. We consider this process as 
generating two separate patterns (a bang and a glider) as well, 
while ignoring the two packs of gliders conservatively. 
\end{itemize}
These modifications will make some of matrix columns sums positive, offering a potential of positive growth.
The transition matrix after all these modifications becomes
\begin{equation*}
M_2 = \left( \begin{array}{ccccccccccc}
\bf{-28} & 0 & 0 & 0 & \bf{8} & \bf{8} & 0 & \bf{8} & 0 & 0 & 0  \\
2 & -28 & 2 & 0 & \bf{8} & \bf{8} & 0 & 0 & 0 & 0 & 0  \\
3/2 & 4 & -33 & 16 & 6 & 0 & 0 & 0 & 0 & 0 & 0  \\
0 & 0 & 1 & -25 & 0 & 0 & 0 & 0 & 0 & 0 & 0  \\
0 & 0 & 4 & 0 & -41 & 1 & 0 & 0 & 0 & 0 & 0 \\
3/2 & 0 & 6 & 0 & 4 & -43 & 0 & 0 & 0 & 0 & 0 \\
0 & 0 & 4 & 0 & 0 & 0 & -25 & 0 & 0 & 0 & 0  \\
2 & 8 & 8 & 4 & 0 & 0 & 10 & -38 & 0 & 0 & 0 \\
1 & 0 & 2 & 0 & 0 & 0 & 0 & 4 & -52 & 0 & 0 \\
11 & 0 & 4 & 1 & \bf{22} & \bf{24} & 4 & 28 & 28 & 0 & 0 \\
13 & 16 & 2 & 4 & 1 & 10 & 11 & 6 & 24 & 0 & 0
\end{array} \right).
\end{equation*}
The {\ff w}- and {\ff r}-decays brought modifications to the first two rows (and the penultimate bang row). 
The eigenvalues are: {
\begin{align}  \label{eq:m1evs} 
& 0,\ 0,\ -19.6549885450,\ -25,\ -25.6574423828,\ -30.5755536586,  \\
& \!-37.2841176353\pm1.1753984861i,  -42.4590793638, 7792, -52.   \nonumber
\end{align}
It is instructive to compare this sequence to (\ref{eq:m0evs}).
The leading eigenvector for the largest non-zero eigenvalue is given in the third column of Table \ref{tb:infm1}.
Only Pond is now an intermediate absorbing state, giving the same eigenvalue $-52$. 
The other remaining integer eigenvalue $-25$ gives a puzzling eigenvector; see the fourth column. 
The other eigenvectors are algebraic numbers of degree 7.

The leadin eigenvector $\vec{v}_1$ of the eigenvalue $\lambda_1$ plays an analogous role to the ordering of webpages
in Google's famous PageRank algorithm (which basically models the links between webpages as Markov's process, and considers
the dominant eigenvector for the page ordering). It should be thus natural and informative to order the stable/cyclic patterns
by their numerical components in the eigenvector $\vec{v}_1$.

\section{``Life'' on small toruses}
\label{sec:small}

The standard method to avoid unbounded ``Life'' evolution on the infinite space is to take a rectangular $m\times n$ region of cells
and wrap it to a torus by the classical topological construction in Figure \ref{fig:torus}. There is no unbounded growth on a finite torus,
and moving patterns return to their initial position eventually. As a consequence, any stabilised configuration is either a still-life or an oscillator. 
For example, the Glider 
in Figure \ref{fig:torus} moves around diagonally with the period $8n$, fading into one corner and then emerging from the opposite corner. 
We consider only square toruses for the following related reasons:
\begin{itemize}
\item Gliders would travel around ergodically around a rectangular torus with (nearly) coprime $m,n$,
with disproportionally large period 4lcm$(m,n)$. Some other moving patterns would behave similarly.
\item As Gliders would move ergodically, they would hit stationary patterns at (nearly) any position. 
There would be no stable combinations of Gliders with stationary patterns, or the diversity of such combinations would be greatly reduced.
\end{itemize}
We identify stabilised patterns after decay perturbations up to the torus symmetries.
For a square $n\times n$ torus, these symmetries are generated by the horizontal and vertical shifts on the torus (permuting the columns or rows cyclically)
and the dihedral group $D_4$ of the symmetries of the square. The group of symmetries has $8n^2$ elements. 

\begin{figure} \centering
\begin{picture}(164,164)
\put(-2,-10){\includegraphics[height=180pt]{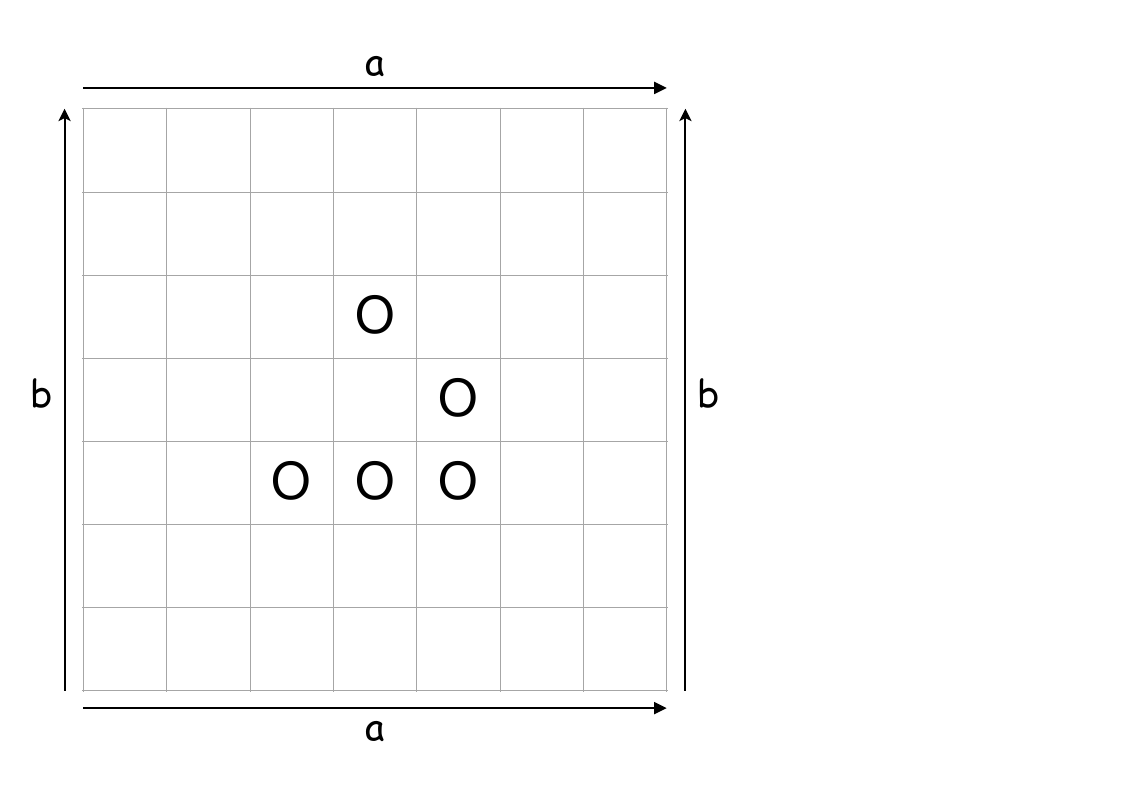}}
\end{picture} 
\caption{Wrapping a square region to a torus. The bottom and top edges are joined in a parallel way following the vector $\vec{a}$,
while the side edges are joined following the vector $\vec{b}$.  }
\label{fig:torus}
\end{figure}

Apart from Appendix Section \ref{sec:stable5}, we consider only the stabilised patterns that can be obtained consequently by decay perturbations
starting from the Block pattern (with 4 live cells) on a considered torus. Excluding the empty space pattern,
we refer to these patterns as {\em ground patterns}. 

Subsections \ref{sec:5x5}, \ref{sec:6x6} describe the ground patterns and transition matrices 
for the $n\times n$ toruses with $n\in\{5,6,7\}$. 
Let us first note down our main conventions and definitions. 

\subsection{Conventions and definitions}

As in Section \ref{sec:start}, we make the following choices:
\begin{itemize}
\item We ignore transitions to the empty space, hence our transition matrices are {\em pseudo-Markovian} rather than Markovian,
as in Section \ref{sec:start}.
\item We choose the left action of the transition matrix on probability vectors. This entails that column sums (rather than row sums)
of the transition matrix are less than or equal $0$.
\item We adopt the continuous time (rather than discrete time) model. The non-diagonal entries $P_{j,k}$ of the transition matrix
are the transition rates (that is, the number of perturbations, averaged over phases) from the ground state $\#k$ to the ground state $\#j$.
We have $P_{k,k}\le 0$, and the column sums are equal to the number of perturbations to the empty space multiplied by $-1$.
The matrix entries $p_{j,k}$ for the discrete time model would be the transition probabilities:
\begin{equation} \label{eq:probable}
p_{j,k}=\frac{P_{j,k}}{n^2} \quad \mbox{if $j\neq k$},  \qquad 
p_{k,k}=1+\frac{P_{k,k}}{n^2}.
\end{equation}
Let us refer to the transitions from a ground state $\#k$ to itself as {\em inert transitions}.
\end{itemize} 
From now on we choose the opposite ordering of ground patterns than in Section \ref{sec:infinite}:
from the most frequent (that is, the most {\em entropic} as products) to the least frequent. 
As a consequence, the transition matrices will be denser in the upper-triangular part (see Figure \ref{fig:matrices89}). 

Let $T_n$ denote the torus of size $n\times n$, and let $M_n$ (for $n>2$) denote the transition matrix 
for the perturbative (pseudo)-Markovian process on $T_n$ in the chosen (mostly Google's PageRank) ordering of ground patterns.

If there are ground patterns which cannot decay ultimately back to Block,
the oriented graph representing the transitions is not {\em strongly connected} \cite[pg.~400--403]{Matrix85},
the transition matrix is reducible, and there are invariant linear subspaces under the action of $M_n$. 
Then the ground patterns can be permuted so that $M_n$ would have a block-upper-triangular structure.
The irreducible invariant subspaces correspond to {\em strongly connected subgraphs};
we refer to the corresponding sets of ground patterns as {\em downstream blocks}.
Presence of these {\em downstream patterns} and blocks is well demonstrated in our starting model in Section \ref{sec:start},
and in the set of all stabilised patterns on $T_5$ and $T_4$ in Appendix Section \ref{sec:stable5}.


Given an eigenvalue $\lambda$, 
its eigenvectors $\vec{v}$ are determined by the vector equation
\begin{equation} \label{eq:basicla}
(M_{n}-\lambda I)\,\vec{v}=\vec{0}.
\end{equation}
of basic linear algebra \cite[\S 1.4]{Matrix85}. 
The equivalent linear system of equations for the components $v_k$ is overdetermined, as non-zero solutions must exist.
We refer to these linear equations as {\em eigenvector equations}. 
For example, the eigenvalue component for a pattern \#$k$ with a single predecessor \#$j$ is determined by the two-term equation
\begin{equation} \label{eq:neg2}
(P_{k,k}-\lambda^{})\,v^{}_k+P_{k,j}\,v^{}_j=0.
\end{equation}
Note that $P_{k,k}=-n^2(1-p_{k,k})$ by (\ref{eq:probable}), and $\lambda$ is negative in the considered cases.
We will not encounter multiple eigenvalues. 
The largest real eigenvalue of $M_n$ must be a real number, as its eigenvector describes 
the asymptotic probabilities for the ground patterns (under the condition that the empty space is not reached). 
We call them {\em the leading eigenvalue} and {\em the leading eigenvector},
recognising that the asymptotically dominant eigenvector for the full Markov process 
would represent the absorbing state of empty space.


\subsection{The 5x5 torus} 
\label{sec:5x5}

There are 9 ground patterns on $T_5$. 
They are labeled and listed in the first three columns of Table \ref{tb:5x5},
along with an additional pattern \#10.
The Markov process between them is depicted graphically in Figure \ref{fg:lifet5}. 
Unsurprisingly, we recognize several well known stable or cyclic patterns 
of the infinite version of Conway's ``Life" transplanted onto the torus. 
The cells of the displayed patterns are labeled to indicate the emerging pattern 
after the perturbation at that cell. The inert transitions 
are labeled by +.
\begin{figure}
\begin{center}
\begin{picture}(272,500)(50,-3)
\input{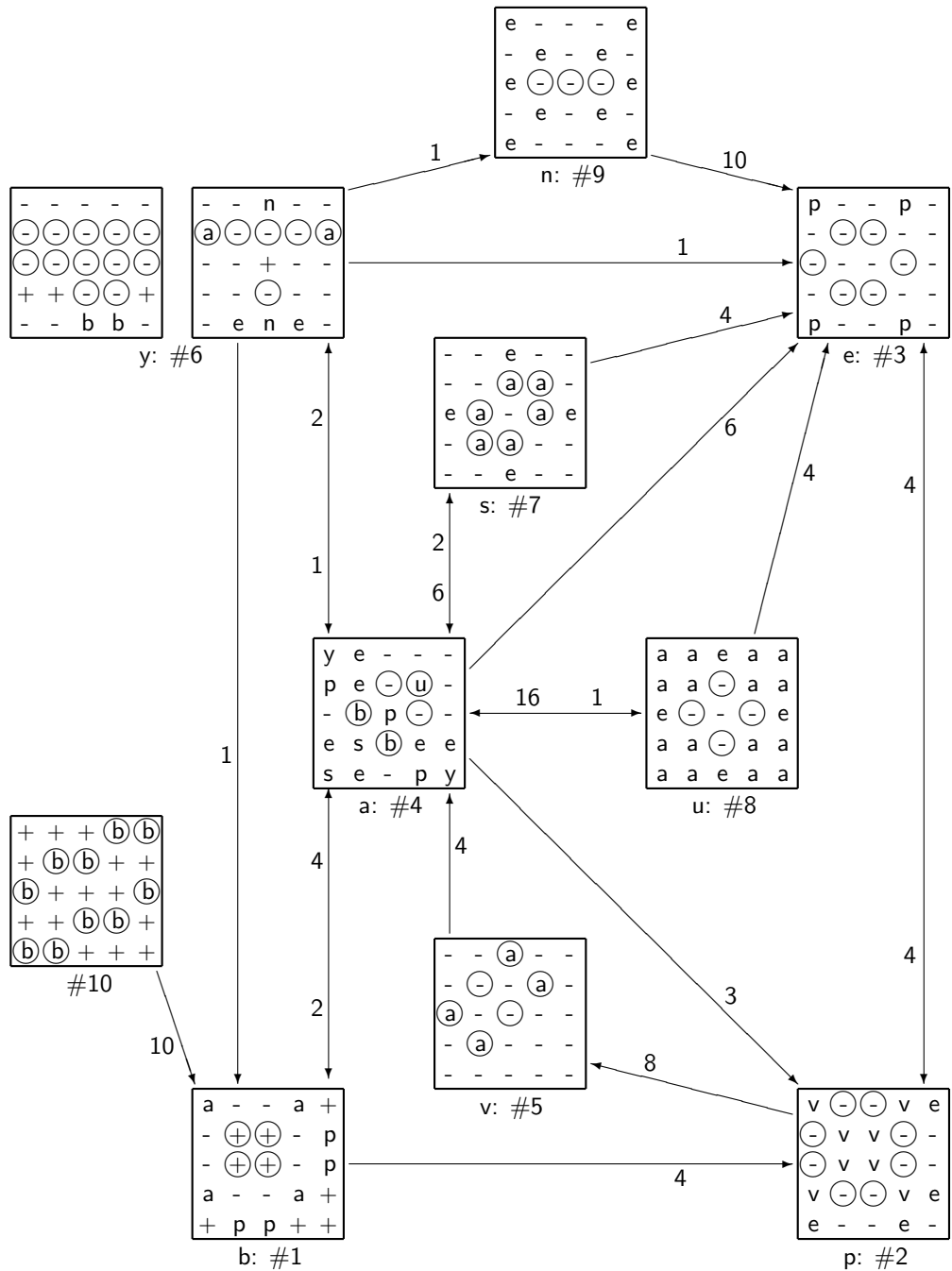}
\end{picture}
\end{center}
\caption{The Markov process between ground patterns on the $5\times 5$ torus. 
The transitions {\sf +, --} are to itself or the empty space, respectively.
The pattern \#10 does not belong to the set of ground patterns,
but has an exceptionally low decay rate $-10$.}
\label{fg:lifet5}
\end{figure}

Including the empty space pattern \#0 (but without the extra pattern \#10), 
the transition matrix of the Markov process is
\begin{equation} \label{eq:m5}
M_5 = \left( \begin{array}{cccccccccc}
0 & 8 & 13 & 21 & 9 & 21 & 19 & 15 & 5 & 15 \\
0 & -16 & 0 & 0 & 2 & 0 & 1 & 0 & 0 & 0 \\
0 & 4 & -25 & 4 & 3 & 0 & 0 & 0 & 0 & 0 \\
0 & 0 & 4 & -25 & 6 & 0 & 1 & 4 & 4 & 10 \\
0 & 4 & 0 & 0 & -25 & 4 & 1 & 6 & 16 & 0 \\
0 & 0 & 8 & 0 & 0 & -25 & 0 & 0 & 0 & 0 \\
0 & 0 & 0 & 0 & 2 & 0 & -23 & 0 & 0 & 0 \\
0 & 0 & 0 & 0 & 2 & 0 & 0 & -25 & 0 & 0 \\
0 & 0 & 0 & 0 & 1 & 0 & 0 & 0 & -25 & 0 \\
0 & 0 & 0 & 0 & 0 & 0 & 1 & 0 & 0 & -25
\end{array} \right).
\end{equation}
The matrix is diagonizable, and the eigenvalues $\lambda_k$ (for $k\in\{0,1,\ldots,9\}$) are
\begin{align}
& 0, -14.1315655736,\ -18.5040532266,\ -23.8982690980\pm 2.68268839792i, \quad\nonumber \\
& -24.5739220123,\ -25,\ -25,\ -28.0793972866,\ -30.9145237049. 
\end{align}
\begin{table} \small
\centering
\begin{tabular}{@{}rccrrrrrr@{}}
\hline
\# & \!Label\! &  \multicolumn{1}{c}{Name} & 
\hspace{-8pt}Cycle & \multicolumn{5}{l@{}}{Eigenvectors\,\dotfill} \\
&&&& \multicolumn{1}{c}{the 1st} & \multicolumn{1}{c}{the 2nd} 
& \multicolumn{2}{c}{$\lambda=-25$} & $\lambda=-10$
\\ \hline
0 & \sf -- & empty space & 
1 & \multicolumn{1}{l}{$-100.$} & $\!\!-85.2526462458$ & $-5$ & $-2$ & $\!\!-110861336$ \\ 
1 & \sf b & Block & 
1 & 20.6547809997 & $\!\!-14.7473537542$ & 0 & 0 & 37388785 \\ 
2 & \sf p & Pond & 
1 & 19.4485069846 & 17.3121325304 & 0 & 0 & 15673200 \\ 
3 & \sf e & Beehive & 
1 & 19.1833837155 & 31.5336529738 & 0 & 0 & 10811840 \\ 
4 & \sf a & Boat & 
1 & 17.3407213313 & 15.1044981909 & 0 & 0 & 14098500 \\ 
5 & \sf v & Barge & 
1 & 14.3155904312 & 21.3205349544 & 5 & 5 & 8359040 \\ 
6 & \sf y & \!\!pyramide\,5/5/2\!\!\! & 
4 & 3.9106612278 & 6.7191623710 & 0 & 0 & 2169000 \\ 
7 & \sf s & Ship & 
1 & 3.1910246961 & 4.6504377939 & 2 & $\!-6$ & 1879800 \\ 
8 & \sf u & Tub & 
1 & 1.5955123480 & 2.3252188969 & $\!-2$ & 1 & 939900 \\ 
9 & \sf n & Blinker & 
2 & 0.3598182658 & 1.0343622886 & 0 & 2 & 144600 \\ 
10 & & slanted bands & 
1 & --- & --- & --- & --- & 19396671 
\\ \hline
\end{tabular}
\caption{The ground patterns and some Markov process eigenvectors 
on the $5\times 5$ torus. As in Appendix Figure \ref{fg:pipatterns}, accepted names of familiar patterns
are capitalized, while patterns without direct analogues in Conway's  
``Life'' on the the infinite plane are named provisionally.} 
\label{tb:5x5}
\end{table}
The non-integer eigenvalues 
are algebraic numbers of degree 7.
They can be expressed as $\lambda=\xi-23$, where $\xi$ satisfies
\begin{equation}
\xi^7+3\xi^6-84\xi^5-326\xi^4+520\xi^3+4016\xi^2+17120\xi+20192=0.
\end{equation}
The stochastic evolution of the probability vector $\vec{v}$ is given by a similar linear expression as in (\ref{eq:m0gensol}):
\begin{equation}
\vec{v}=\sum_{k=0}^N c_k\expp{-\lambda_kt}\vec{v}_k,
\end{equation}
where $\vec{v}_k$ are the corresponding eigenvectors, 
the coefficients $c_k$ are determined by the initial distribution, and $N=9$.
The eigenvalue $\lambda_0=0$ dominates as \mbox{$t\to\infty$.} 
It represents the attraction and stability of the empty space; its eigenvector is non-zero only at the \#0 component.
Contribution of the other eigenvectors decreases exponentially, 
but the eigenvector $\vec{v}_1$ of the next eigenvalue $\lambda_1\approx-14.13$ 
dominates the asymptotic probability distribution under the condition that the \#0 state has not been reached.

Table \ref{tb:5x5} orders the ground patterns \#1 to \#9 by the leading eigenvector, which is given in the fourth column.
As the \#0 entry is normalized to $-100$, the other components are positive, 
and give the asymptotic percentages for the ground patterns by Lemma \ref{th:evortho}.
The eigenvectors can be expressed symbolically with some neatness from the observation
that the bottom five rows of 
$M_5-\lambda I$ have only two non-zero entries each. 
Starting from the eigenvector equations implied for these rows, the following expression 
can be deduced:
\begin{align} \label{eq:ev5}
\left( -v_0,\frac{8\,(\xi+1)}{\xi-7},(\xi+2)\,\eta,\frac{(\xi+2)^2\,\eta}4-3\xi-\frac{8\,(\xi+1)}{\xi-7},4\xi,8\eta,8,\frac{8\xi}{\xi+2},\frac{4\xi}{\xi+2},
\frac{8}{\xi+2} \right),
\end{align}
where 
\begin{align*}
\eta = \frac{\xi^2+2\xi-2}{8}-\frac{7\xi}{2\,(\xi+2)}-\frac{\xi+1}{\xi-7}, \qquad 
v_0  = \frac{(\xi^2+8\xi+44)\,\eta}{4}+\frac{\xi^2+22\xi+24}{\xi+2},
\end{align*}
as computed eventually from the most dense rows.

The whole set of stabilised patterns on the $T_5$ tours is described in Section \ref{sec:stable5}.

\subsection{The $6\times6$ and $7\times7$ toruses}
\label{sec:6x6}

The ground patterns on $T_6$ and $T_7$ 
are listed in Table \ref{tb:6x6} and Appendix Table \ref{tb:7x7a}.
Most of them are depicted in Figures \ref{fg:6x6} and \ref{fg:7x7}. 
In particular, moving pyramid patterns with a full row (or column) of living cells appear occasionally.
They either move cyclically through the whole orthogonal dimension, or overturn cyclically (like \#6 on $T_5$). 

\begin{figure}
\begin{center}
\begin{picture}(440,198)
\put(0,-101){\includegraphics[width=440pt]{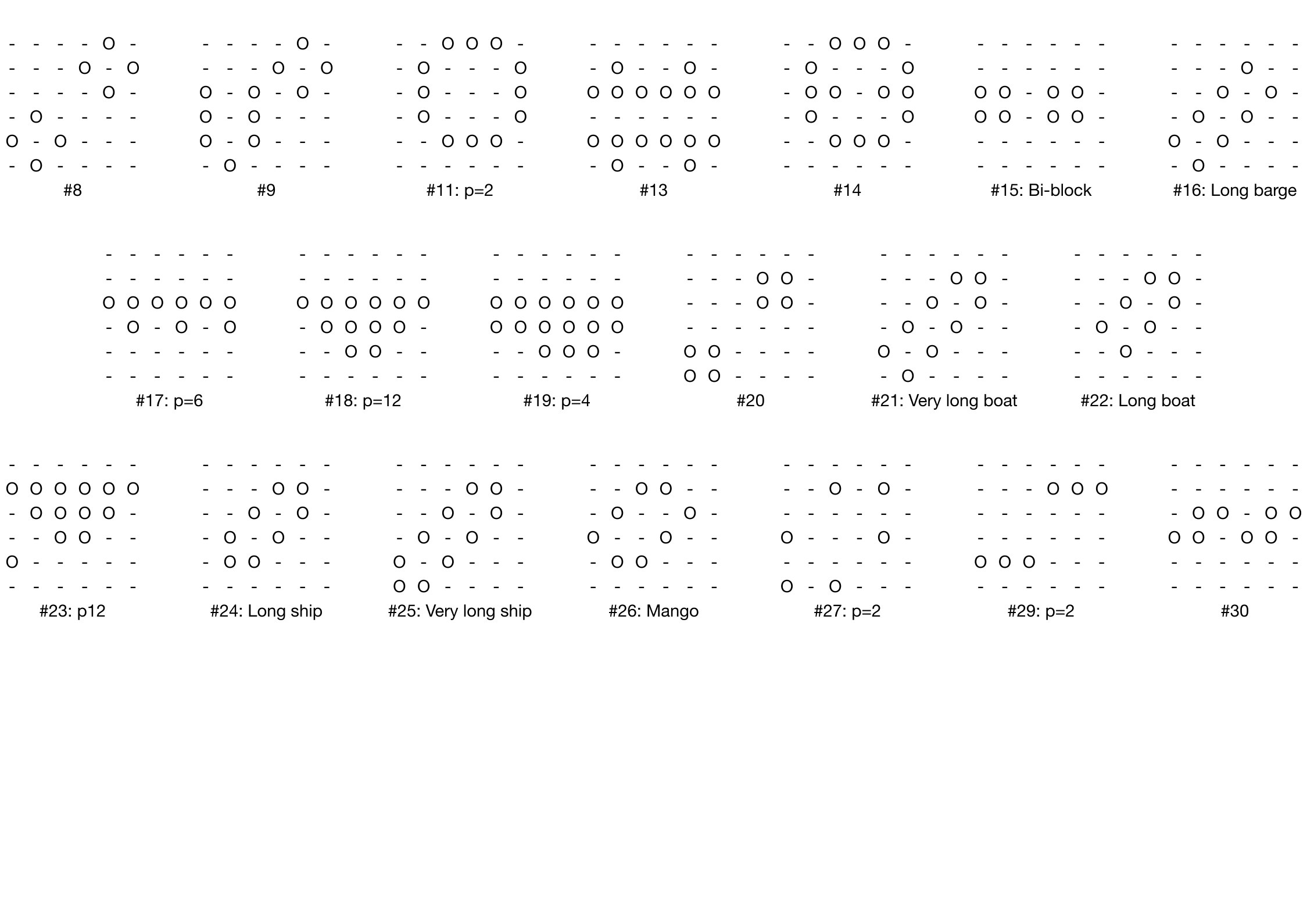}}
\end{picture}
\end{center}
\caption{Novel ground patterns on the $6\times 6$ torus.} 
\label{fg:6x6}
\vspace{25pt}
\small \centering 
\begin{tabular}{@{}rcrlccl@{}}
\hline
\# & \multicolumn{1}{c}{Name} & \!\!\!\!Cycle & The leading  & \multicolumn{3}{@{}l@{}}{\!Transition rates to\dotfill} \\
 & & & eigenvector & \!itself & \!\!\#0 & the others  (\#:\,rate)
\\ \hline
1 & Block & 1 & 55829.\,2876928 & 8 & 16 & 2:\,8;\, 3:\,4 \\ 
2 & Beehive & 1 & 33374.\,5099614 & 2 & 18 & 1:\,16 \\ 
3 & Boat & 1 & 18958.\,4865190 & 3 & 6 & 1:6; 2:8; 4:6; 5:1; 6:4; 7:2 \\ 
4 & Loaf & 1 & 6740.\,56822004 & 1 & 23 & 1,3:\,2;\, 6:\,1; 7:\,2; 8:\,3; 10:2 \\ 
5 & Tub & 1 & 4363.\,85255859 & 11 & 4 & 2:\,4;\, 3:\,16;\, 11:\,1 \\ 
6 & Ship & 1 & 4219.\,43666015 & 1 & 21 & 1:\,4;\, 3:\,6;\, 4:\,4 \\ 
7 & Pond & 1 & 2589.\,15754090 & 0 & 16 & 1:\,8;\, 8:\,12 \\ 
8 & chess tubs & 1 & 2532.\,56743661 & 0 & 20 & 5,9:\,8 \\ 
9 & tub+3+1+3 & 1 & 1034.\,90235122 & 1 & 28 & 5:\,2; 12:\,1; 13:\,2; 15,17:\,1 \\ 
10 & Glider & 24 & 815.\,148252458 & 4 & 19 & 1:\,4;\, 2:\,2;\, 3,4:\,1.5;  7,8:\,1;\, 11:\,2 \\  
11 & blinker hoop & 2 & 491.\,022937595 & 4 & 25.5 & 5:\,0.5;\, 7,14,16:\,2 \\ 
12 & Blinker & 2 & 179.\,522423355 & 11 & 11 & 2:\,10;\, 11:\,4 \\ 
13 & double L-band & 1 & 100.\,587053737 & 0 & 28 & 18:\,8 \\ 
14 & narrow hoop & 1 & 55.\,8702876463 & 3 & 30 & 11:\,3 \\ 
15 & Bi-block & 1 & 52.\,0322389837 & 0 & 20 & 2:\,4;\, 11:\,12 \\ 
16 & Long barge & 1 & 51.\,2300206006 & 0 & 13 & 5:\,2;\, 12:\,8;\, 20:\,2; 21:\,6;\, 22:\,4;\, 25:\,1 \\ 
17 & pyramide\,6+(1,1,1) & 6 & 50.\,2935268684 & 0 & 19.5 & 2:\,3;\, 11:\,9;\, 12:\,1.5;\, 19:\,3 \\ 
18 & pyramide\,6/4/2 & 12 & 41.\,1544950478 & 0.5 & 19.5 & 1:0.5;\;2:2;\;3,4,7:1;\,11:3;\,12,19:2;\,23:2.5;\,29,30:0.5 \\ 
19 & pyramide\,6/6/3 & 4 & 17.\,8167817849 & 6 & 20.5 & 2:\,2.5;\, 4,11:\,2;\, 12:\,3 \\ 
20 & chess blocks & 1 & 17.\,1047609258 & 12 & 0 & 2:\,16;\, 5:\,8 \\ 
21 & Very long boat & 1 & 15.\,8889547850 & 0 & 19 & 1:\,3;\, 3:\,1;\, 10,12,15:\,2; 16:\,4;\, 24:\,2;\, 25:\,1  \\ 
22 & Long boat & 1 & 11.\,9590003498 & 0 & 18 & 1:\,2;\, 5:\,1;\, 6,7:\,2; 24,26:\,2;\, 28:\,3 \\ 
23 & pyramide\,6/4/2+1 & 12 & 5.\,39313847507 & 1.5 & 17 & 1:1;\;2:2;\;4,5,7:0.5;\;11:2;\;12:3;\;18:4;\,19:2.5;\,29:1.5 \\ 
24 & Long ship & 1 & 4.\,28556011147 & 2 & 10 & 2:\,8;\, 3,16:\,2;\, 20:\,4; 22:\,6;\, 27:\,2 \\ 
25 & Very long ship & 1 & 3.\,26180531714 & 0 & 10 & 6,8:\,2;\, 11,19:\,4;\, 20:\,5;\, 21:\,6;\, 27:\,3 \\ 
26 & Mango & 1 & 2.\,70011501622 & 0 & 18 & 1:\,2;\, 2:\,4;\, 3,5:\,2;\, 7,20:\,4 \\ 
27 & isolated sparks & 2 & 2.\,14014299147 & 12 & 24 \\ 
28 & Barge & 1 & 1.\,93123340814 & 2 & 0 & 3:4; 4:16; 7:2; 22:8; 26:4 \\ 
29 & parallel blinkers & 2 & 1.\,54312177838 & 2 & 24 & 12:\,10 \\ 
30 & snake band & 1 & 1.\,0 & 0 & 24 & 2:\,8;\, 15:\,4 
\\ \hline
\end{tabular}
\captionof{table}{The ground patterns, the leading eigenvalue, and the transition data for the $6\times 6$ torus.} 
\label{tb:6x6}
\end{figure}

\begin{figure}
\centering
\begin{picture}(440,254)
\put(0,-48){\includegraphics[width=440pt]{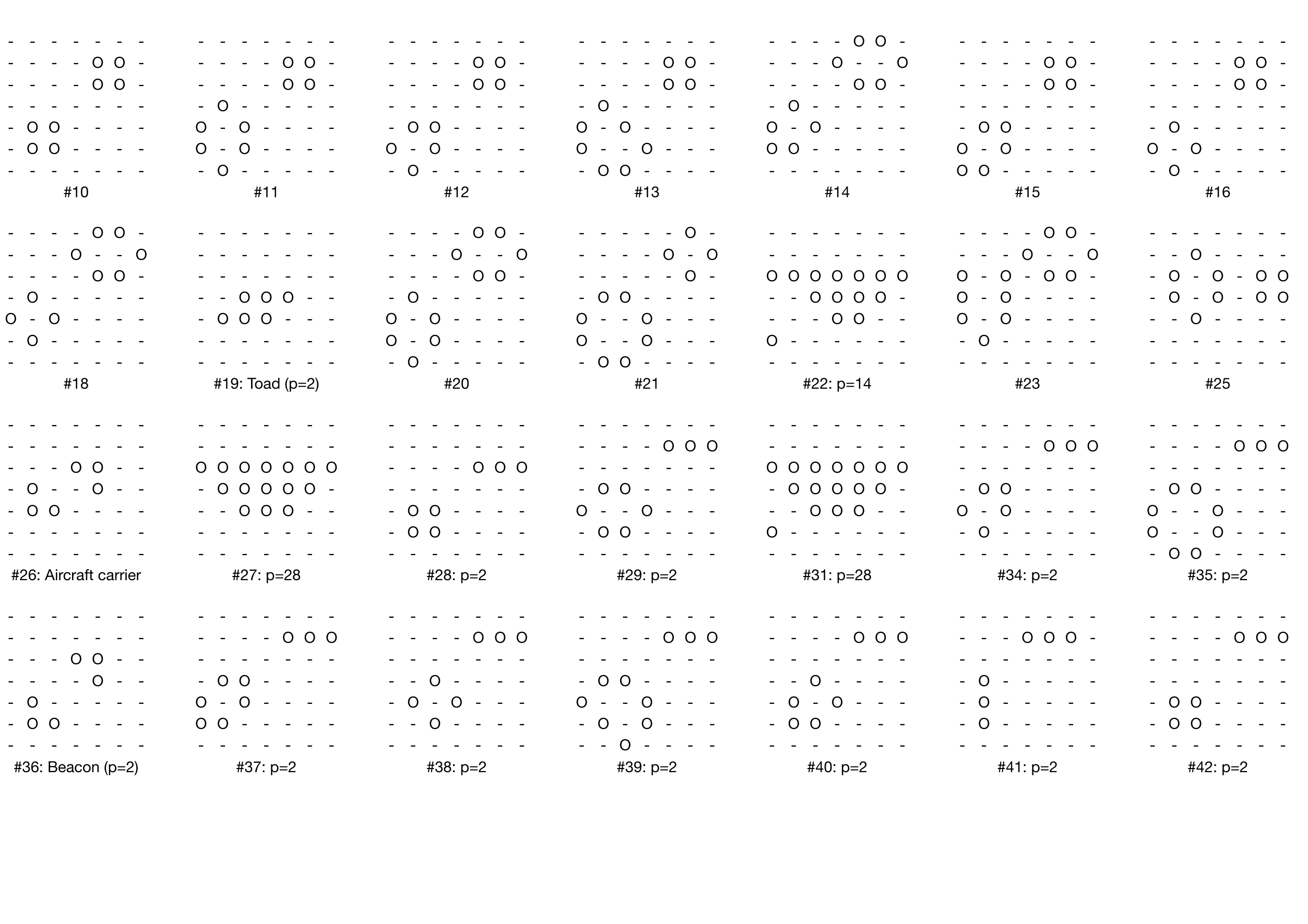}}
\end{picture}
\caption{Novel ground patterns on the $7\times 7$ torus.} 
\label{fg:7x7}
\end{figure}

The leading eigenvalues are:
\begin{align*}
\lambda_6=-15.4227524761 && \mbox{of algebraic degree 28, for $T_6$;} \\
\lambda_7=-22.6111083077 && \mbox{of algebraic degree 36, for $T_7$.}
\end{align*}
Tables \ref{tb:6x6}  and \ref{tb:7x7a} list the corresponding leading eigenvectors in the fourth column.
The sparse transition matrix $M_6$ of the pseudo-Markovian process on $T_6$ 
has the size $30\times30$. It is described by the 5th and 7th columns of Table \ref{tb:6x6}. 
The diagonal entries are equal to $\varrho_i-36$, 
where $\varrho_i$ is the self-transition frequency in the 5th column.
The transition frequencies from each pattern sum up to 36, and their appear in the columns of $M_6$. 
The curious pattern \#27 decays only to itself and the empty space \#0. 
The live cells of this oscillating pattern (identifiable as homotopically looping Barberpole \cite[{\sf Barberpole}]{Wiki}) 
are isolated: they never have a live neighbour!
The most significant eigenvalues following $\lambda_6$ are: 
\begin{equation*}
-22.2782002055, -23.9600500434\pm 0.8554013360i, -24, -28.2463940448, \ldots.
\end{equation*}
The eigenvalue for the integer eigenvector $-24$ is supported only on the pattern \#27, 
which is irreversibly an intermediate pattern towards the decay to \#0. 
There is other integer eigenvalue $-36$. 


The transition matrix $M_7$ of the pseudo-Markovian process on $T_7$ 
has the size $42\times42$. 
It is described by Appendix Table \ref{tb:7x7a}. 
The diagonal entries are equal to $\varrho_i-49$, where $\varrho_i$ is the self-transition frequency in the 5th column.
The patterns \#1, \#3, \#6, \#8 decay only to each other and to \#0. 
Consequently, these four patterns comprise a block-diagonal component of $M_7$
that represents an irreversible intermediate stage of the eventual decay to \#0. 
The transition matrix on these four patterns is 
\begin{equation}
\left( \begin{array}{cccc}
-38 & 10 & 0 & 6 \\
0 & -25 & 0 & 4 \\
4 & 0 & -44 & 0 \\
0 & 0 & 8 & -38 
\end{array} \right),
\end{equation}
and the dominant eigenvalue in this block is  $\lambda^*_7\approx -24.6102656350$.
The dominant eigenvector (in percentage) of this downstream block is 
\begin{equation*}
(38.5420268338, 48.7564410315, 7.95101698832,  4.75051514635).
\end{equation*}
The whole transition matrix $M_7$ has two integer eigenvalues: $-41$ and $-49$. 
The most significant eigenvalues 
following $\lambda_7,\lambda^*_7$ are:
\begin{equation*}
-26.8283684934,\ -31.2767375457,\ -31.9999811357,\ -34.0916602249,\ \ldots.
\end{equation*}
The eigenvalues for the $n\times n$ torus tend to concentrate around the point $\lambda=n^2$,
as evident in Appendix Figure \ref{fg:egvs}. This makes iterative computation of the most negative eigenvalue 
(and particularly of the corresponding eigenvector) very slow and problematic.
The eigenvalue closest to $0$ seems to be fairly isolated for larger toruses,
makes its iterative computation more feasible after an appropriate diagonal shift
--- although we will remark that it is much more effective to iterate the transposed matrix.

\section{Larger toruses}
\label{sec:large}

We computed the ground patterns and the transition matrices $M_{n}$ 
for $T_8$, $T_9$ and $T_{10}$ 
as well. The summary statistics is presented in Table \ref{tb:summary}.
The number $N$ of ground patterns is starting to grow tremendously.
The quadratic exponential growth 
\begin{equation}
N=O\Big(B^{\,n^2}\Big)
\end{equation} 
is expectable, but its rate $B$ does not appear settled; see the 3rd row. 
The number of ground patterns for $T_{11}$ 
could be $\approx 1.15^{121}$, which is about 22 million.

\begin{table} \small 
\begin{tabular}{@{}lrrrrrr@{}}
\hline
Torus size ($n\times n$) & $5\times 5$ & $6\times6$ & $7\times 7$ & $8\times8$ & $9\times9$ & $10\times10$
\\ \hline 
Ground patterns ($N$) & 9 & 30 & 42 & 305 & 7362 & 513875  \\ 
Downstream blocks & --- & \{\#27\}\! & \!\{\#1,3,6,8\}\! & --- & --- & \{\#14342\}\! \\
$B=\exp((\ln N)/n^2)$ & 1.091867 & 1.099085 & 1.079264 & 1.093496 & 1.116197 & 1.140535  \\
1st eigenvalue 
& $-14.13156$ & $-15.42275$ & $-22.61111$ & $-9.789896$ & $-7.021839$ & $-8.721470$ \\
2nd eigenvalue 
& $-18.50405$ & $-22.27820$ & $-24.61026$ & $-26.01347$ & $-26.267$ & $-25.222$  \\
Last eigenvalue 
& $-30.91452$ & $-40.20555$ & $-51.20720$ & $-69.18476$ & $-96.963$ &  $-121.18$ \\
\hline 
\refpart{a} Entropy rate (\ref{eq:entrate}) & 1.430507 & 1.886287 & 1.637994 & 1.927611 & 1.650791 & 1.555746 \\
Median entropy rate & 1.335431 & 1.928736 & 2.310640 & 3.074516 & 3.663965 & 4.217012 \\[1pt]
\refpart{b} Shannon entropy \!of $\vec{v}^{\,(1)}\!$ & 2.692098 & 2.374400 & 2.986337 & 2.719185 & 2.586408 & 3.067091 \\
Top 10 percentage & 100\hspace{21pt} & 99.15879 & 99.59111 & 97.25507 & 99.74079 & 95.55507 \\[1pt]
\refpart{c} $\log_{10} \, v^{(1)}_{1}\big/\,v^{(1)}_{N}$ & 1.758937 & 4.746862 & 11.01298 & 15.53864 & 65.01111 & 131.0476 \\[2pt]
Median $v^{(1)}_{k}\big/v^{(1)}_{k+1}$ 
& 1.218418 & 1.313861 & 1.403064 & 1.058780 & 1.008688 & 1.000150 \\[1pt] 
\refpart{d} Boltzmann entropy rate & $-0.129045$ & $-0.108289$ & 0.721732 & $-2.643578$ & $-1.680071$ & $-2.376045$ \\
Median B. entropy rate & 0.333032 & 13.51961 & 68.72683 & 113.7597 & 899.4877 & 1280.550 \\[1pt]
\refpart{e} Decays to \#0, mean & 14\hspace{26pt} & 17.4\hspace{19pt} & 20.17857 & 19.83326 & 19.71595 & 21.72508 \\
Inert decays, mean & 1.22222 & 2.9\hspace{19pt} &  7.46429 & 8.58942 & 8.79680 & 7.21781\\[1pt]
\refpart{f} Entropic decays, mean & 6.88889 & 11.86667 & 18.65476 & 30.94083 & 48.10735 & 65.20174 \\
Negentropic decays, mean & 2.88889 & 3.83333 & 2.70238 & 4.63649 & 4.37990 & 5.85537 \\
Entropic decay products, mean\! & 1.44444 & 2.93333 & 4.97619 & 9.67869 & 17.38549 & 28.31303 \\
Negentropic products, mean & 0.88889 & 1.26667 & 1.11905 & 2.36066 & 2.41361 & 3.25232 \\
\refpart{g} Negentropy reciprocated,\,\%\! 
& 62.5\hspace{19pt} & 36.84210 & 40.42553 & 47.5\hspace{19pt} & 76.50402 & 73.11398 \\
Negentropy compensated,\;\% & 37.5\hspace{19pt} & 26.31579 & 36.17021 & 34.58333 & 49.94654 & 46.82840 \\
Negentropy balanced,\;\% & 12.5\hspace{19pt} & 0\hspace{26pt} & 8.51064 & 8.33333 & 11.13737 & 7.24622 \\
\refpart{h} With entropic decays only & 5 
& 12 
& 17 
& 81 & 1197 & 46186 \\
Max negentropic \#-rank jump & 4 
& 12  
& 7 
& 117 
& 868  
& 59614 \\ 
Max negentropic decay & 10.86843 
& 41.15450  
& 44.13044  
& 3313.867  
& 147.9563  
& 2853.7694 \\ 
More negentropic than (\ref{eq:neg1}) & 0 & 2 & 2 & 78 & 82 & 5338 \\ 
\refpart{i} 
With 1 decay predecessor\!\! & 5 & 7 & 11 & 33 & 214 & 3625 \\ 
With 2 decay predecessors & 1 & 8 & 6 & 36 & 462 & 13807 \\ 
With 3 decay predecessors & 1 & 3 & 9 & 38 & 768 & 30234 \\ 
\refpart{j}  Mean live cell number & 5.66667 & 8.25\hspace{14pt} & 8.35714 & 11.48728 & 15.13386 & 18.86846 \\ 
Expected live cell number & 5.87708 & 5.14679 & 5.46696 & 5.32580 & 4.84413 & 5.21633 \\
 \hline
Still-lives & 7 & 21 & 25 & 242 & 4957 & 307062 \\
Oscillators of period $p=2$ & \{\#9\}\! & 4 & 13 & 46 & 2389 & 206738 \\
Oscillators of period $p=3$ & --- & --- & --- & --- & \{\#2830\}\! & 6 \\
Oscillators of period $p=4$ & \{\#6\}\! &\{\#19\}\! & --- & \{\#114\}\! & \{\#267\}\! & 5 \\
Oscillators of period $p=5$ & --- & --- & --- & --- & \{\#498\}\! & \{\#43990\}\! \\
Oscillators of period $p=6$ & --- & \{\#17\} & --- & \{\#11,28\}\! & --- & --- \\
Oscillators of period $p=8$ & --- & --- & --- & 7 & --- & \{\#12673\}\! \\
Oscillators of period $p=9$ & --- & --- & --- & \{\#59\}\! & --- & --- \\
Oscillators of period $p=10$ & --- & --- & --- & \{\#117\}\! & --- & 6 \\
Oscillators of period $p=2n$ & --- & \{\#18,23\}\! & \{\#22\}\! & 2 & 5 & 8 \\
Oscillators of period $p=4n$ & --- & \{\#10\}\! & 3 & \{\#14\} & 8 & 40 \\
Other oscillators, $p$ and $\{\#\}$  & --- & --- & --- & 48 \{\#10\},\! & --- & 12,\,14,\,18,\\
&&&&132 \{\#9\} && 26,\,38,\,60, \\ 
&&&& && 100,\,220\, \\ \hline
\end{tabular}
\caption{The summary table. The middle part is discussed in the commentary \refpart{a}\,--\refpart{f},
and is illustrated in Figure \ref{fig:entropy} on the example of the $9\times 9$ torus.} 
\label{tb:summary}
\end{table}

The 
intriguing hallmark 
for the larger toruses is the sudden jump of the leading eigenvalue $\lambda^{(1)}$ of the pseudo-Markovian process 
to values $\lambda^{(1)}>-10$; see the 5th row. 
The second eigenvalue does not shift significantly at all; see the subsequent row. 

The middle part of Table \ref{tb:summary} describes chiefly the leading eigenvector,
the induced asymptotic probability distribution for the ground patterns,
and associated entropic measures. This is explained in 
Subsection \ref{sec:entropic}. 

Figure \ref{fig:matrices89} displays the distribution of non-zero entries in the transition matrices 
$M_8$ and $M_9$, 
overstressing their presence somewhat.
Pattern generation on $T_8$ 
is dominated by high period oscillators at \#9 and \#10 (see Appendix Figure \ref{fg:8x8}), 
while there are no oscillators with many distinct phases on $T_9$ 
(see Figure \ref{fg:9x9}). The shape of the transition matrices 
is commented in parts \refpart{b},\refpart{f},\refpart{g} of Subsection \ref{sec:entropic},
and in Appendix Section \ref{sec:regimes}.

The lower part of Table \ref{tb:summary} counts separately the number of still-lives and various oscillators. 
If there is one or two oscillators of a given period, the \#-rankings of those oscillating patterns is displayed.
The ground patterns are presented concisely in Section \ref{sec:ground}, with references to Appendix.

Subsection \ref{sec:equal} presents the technical issue that the leading eigenvalue components of some ground patterns 
on $T_9$ and $T_{10}$ 
are equal. Therefore Google's PageRank algorithm is not sufficient  to order the ground patterns on these toruses. 
Subsection \ref{sec:acriterion} formulates an additional criterion for ordering the ground patterns.

\subsection{Ground patterns}
\label{sec:ground}

Appendix Table \ref{tb:patterns} displays \#-ranking of some recognisable 
``Life'' patterns 
on various toruses, together with top 10 asymptotic percentages 
on each torus. Evidently, the asymptotic probability distribution tends to be very uneven, 
and the same 8 simple patterns dominate the distribution.

\begin{figure}
\begin{center}
\begin{picture}(400,200)
\put(-14,-4){\includegraphics[width=200pt]{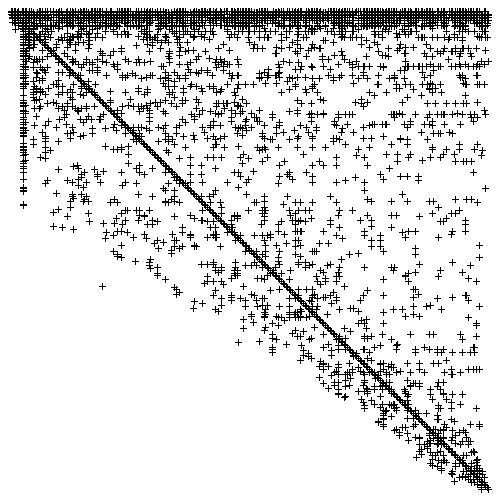}}
\put(216,-4){\includegraphics[width=200pt]{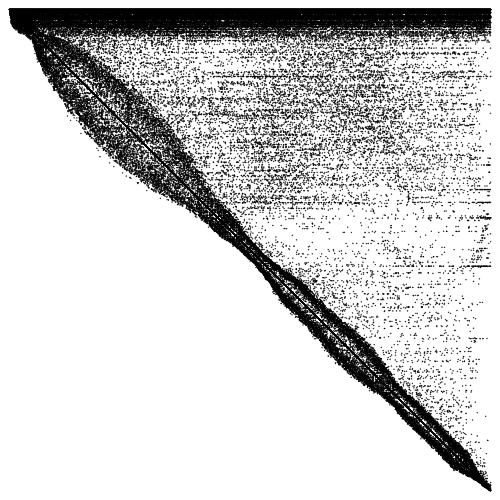}}
\put(-12,0){\small{}\refpart{a}}  \put(216,0){\small{}\refpart{b}} 
\end{picture}
\end{center}
\caption{\small{}\refpart{a} The distribution of non-zero entries of the transition matrix for the $8\times8$ torus.
\refpart{b} The distribution of non-zero entries of the transition matrix for the $9\times9$ torus.} 
\label{fig:matrices89}
\begin{center}
\begin{picture}(440,127)
\put(32,-136){\includegraphics[width=380pt]{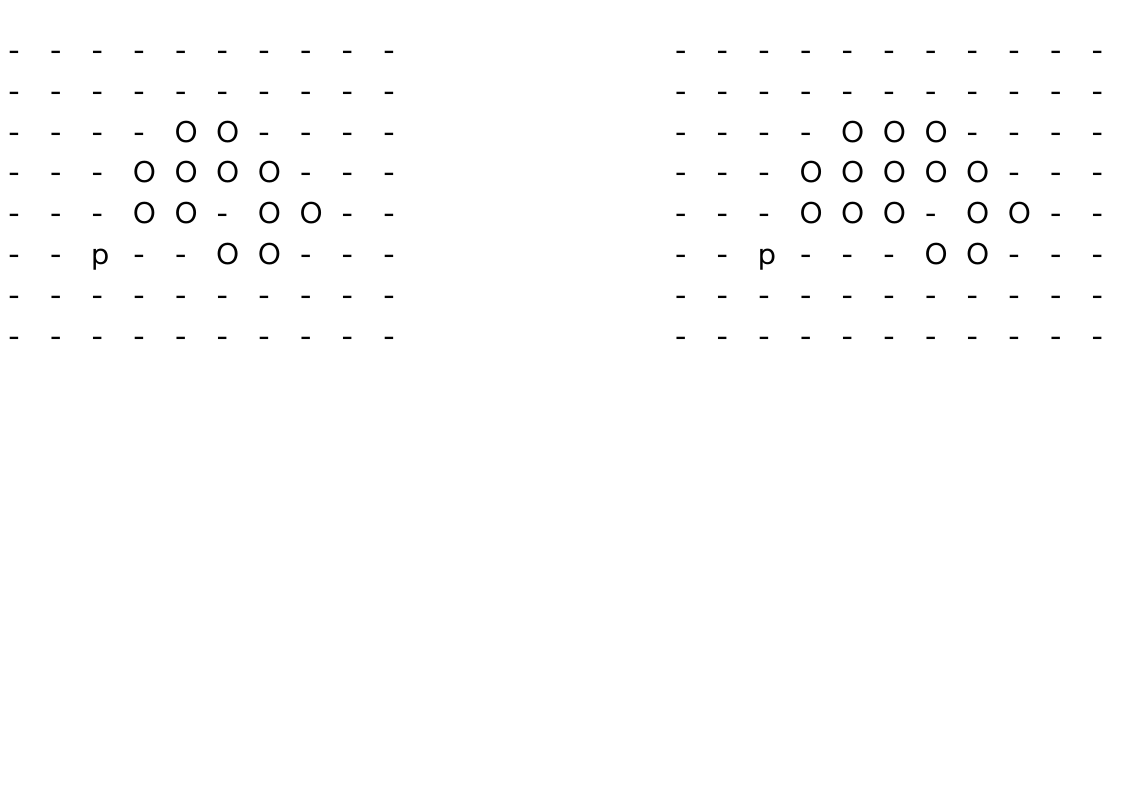}}
\put(14,0){\small{}\refpart{a}}  \put(240,0){\small{}\refpart{b}} 
\end{picture}
\end{center}
\caption{\small{} \refpart{a} Perturbation (marked by ``p'') of Light weight spaceship to Middle weight spaceship.
\refpart{b} Perturbation of Middle weight spaceship to Heavy weight spaceship.} 
\label{fig:spaceshipp}
\end{figure}

Most ground patterns on larger toruses consists of arrangements  of a few islands,
predominantly of the most familiar top 8 patterns in Table \ref{tb:patterns}. 
This is discussed in Appendix Section \ref{sec:islands}.
The well-known pattern Traffic light (see in Figure \ref{fg:pipatterns}) appears at \#10 
on $T_{10}$. 
More peculiar patterns on $T_{10}$ 
are displayed in Figure \ref{fg:10x10}.
The period 220 oscillator \#31900 moves diagonally the minimal distance every 22 generations. 
Pairs of gliders start to appear  on $T_9$ 
(\#1919, \#2370, \#3582, \#6627). 
There are 38 pairs of gliders on $T_{10}$, 
including:
\begin{itemize}
\item \#123607 with the period 20 rather than 40, as the two gliders switch positions in synchrony. 
See Appendix Figure \ref{fig:gliderpairs10}.
\item \#4641, \#14994, \#20134, \#61046, with two gliders moving in orthogonally intersecting directions.
\end{itemize}
The only ground pattern where a glider or spaceship combines with a stationary object is \#7297 on $T_{10}$; 
see Figure \ref{fig:gliderpairs10}.
Interestingly, there are (unique over two phases) perturbations of Light weight spaceship to Middle weight spaceship, 
and then to Heavy weight spaceship. These perturbations are indicated in Figure \ref{fig:spaceshipp}.

\subsection{Entropic measures}
\label{sec:entropic}

The middle part of Table \ref{tb:summary} displays a few entropy measures of the Markovian processes,
and reflects the asymptotic probability distribution for the ground patterns 
defined by the leading eigenvector 
\begin{equation}
\vec{v}^{\,(1)}=\big(\, v^{(1)}_{k} \,\big)_{k=1}^N.
\end{equation}
Figure \ref{fig:entropy} and Appendix Section \ref{sec:regimes}
supplement some parts of the table and these corresponding explanations:
\begin{figure}
\begin{center}
\begin{picture}(380,526)
\put(8,365){\includegraphics[width=360pt]{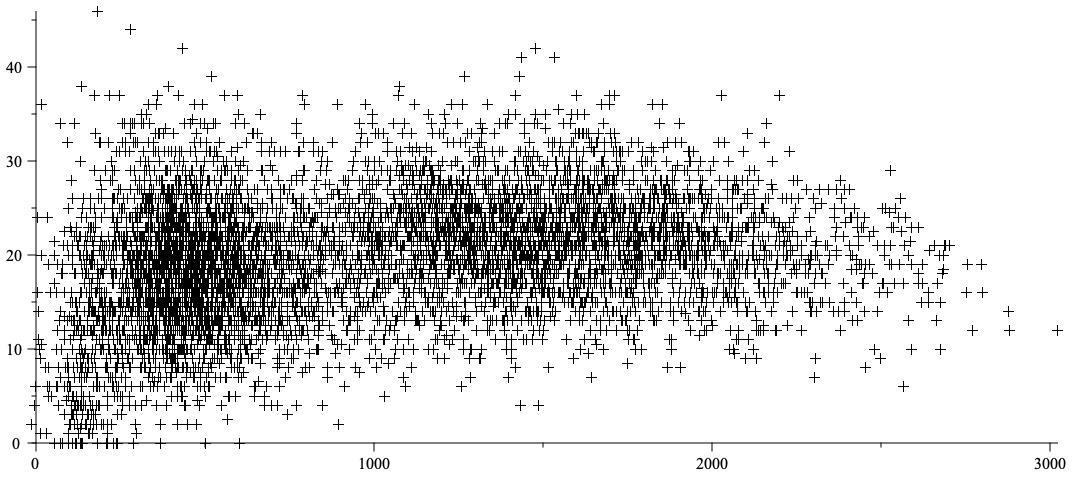}}
\put(8,185){\includegraphics[width=360pt]{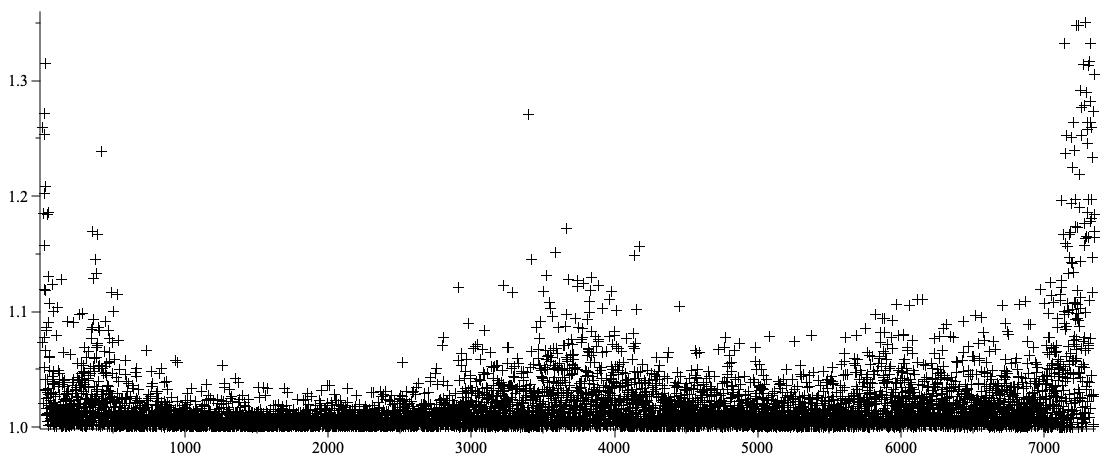}}
\put(8,-1){\includegraphics[width=360pt]{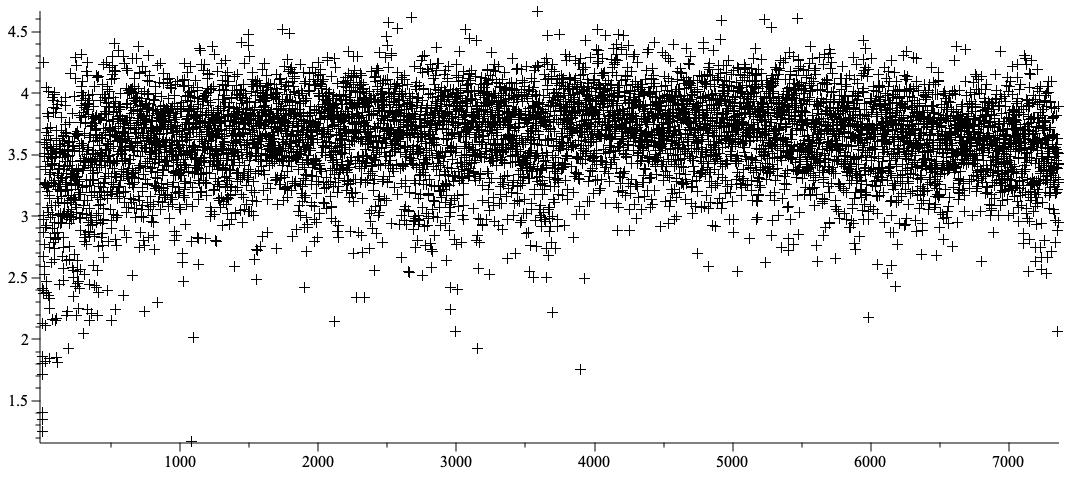}}
\put(0,365){\small{}\refpart{e}}  \put(0,183){\small{}\refpart{c}}  \put(0,-2){\small{}\refpart{a}}
\end{picture}
\end{center}
\caption{Entropic distributions on the $9\times 9$ torus. 
With reference to the commentary \refpart{a}\,--\refpart{f} to the middle part of Table \ref{tb:summary}: 
\refpart{a} Shannon entropy
of each distribution of the transition probabilities from all states $\#k\in\{1,2,\ldots,7362\}$, numbered horizontally.
\refpart{c} The relative jumps of consecutive components of the leading eigenvector. Out of range jumps appear near the top 
(up to $v^{(1)}_{12}/v^{(1)}_{11}\approx 2.52343$) or near the tail (up to $40.13927$) of the ranked sequence. 
\refpart{d} Entropy growth rate (\ref{eq:dbentropy}) from each state (horizontally) vs the number of transitions to the empty space \#0 (vertically).} 
\label{fig:entropy}.
\end{figure}
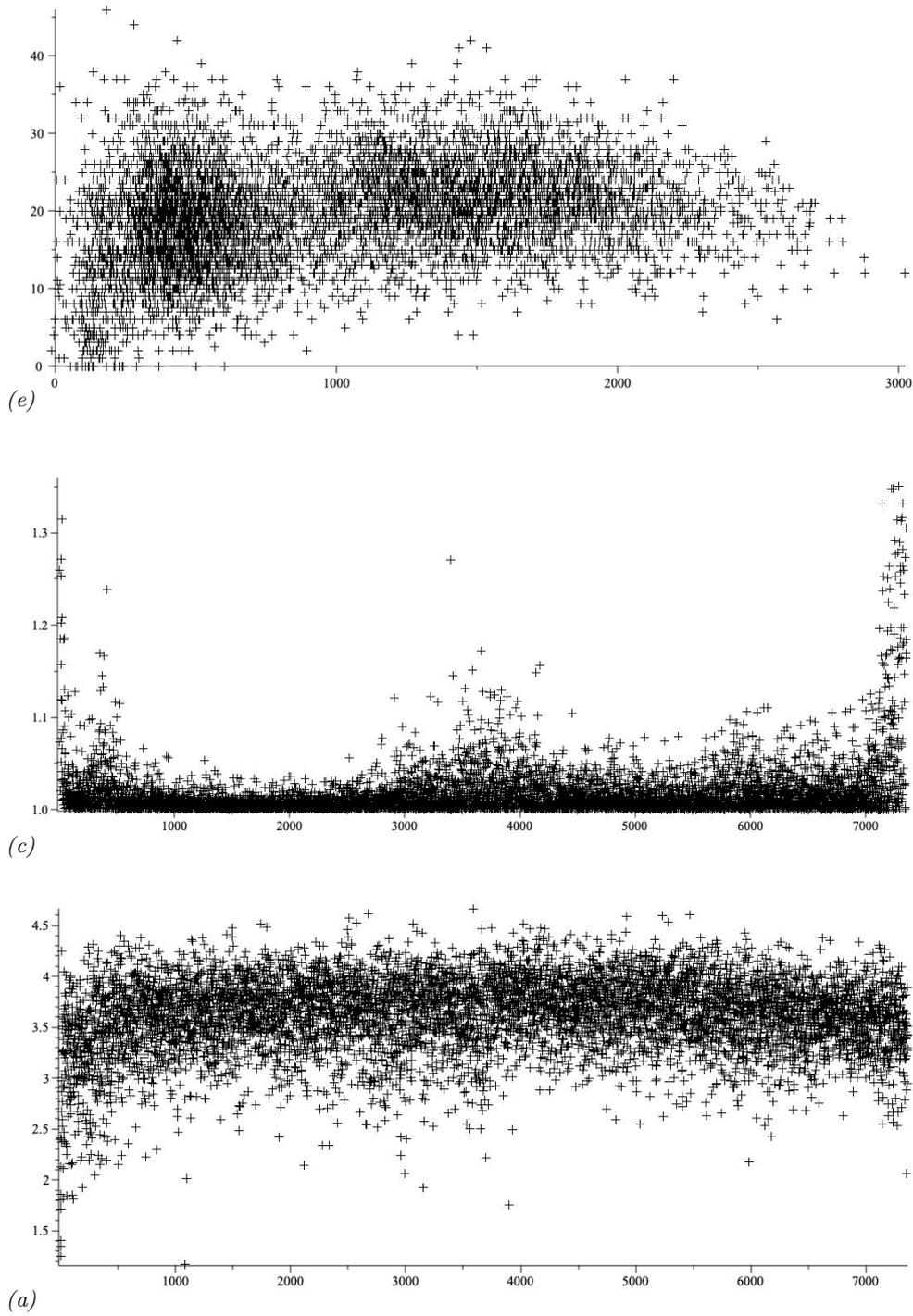
\begin{enumerate}
\item[\refpart{a}] Stochasticity of a Markovian process with $K$ states is often measured 
by the {\em entropy rate}
\begin{equation}
\sum_{k=1}^{K} u^{(1)}_{k} \; \sum_{j=1}^{K} -p_{j,k}\log_2 p_{j,k}.  
\end{equation}
Here the inner sum gives the Shannon entropy 
of the distribution of the transition probabilities from each state $k$, 
and the outer sum averages the Shannon entropies according to the asymptotically dominant distribution $\big(u^{(1)}_{k}\big)_{k=1}^K$.
As our considered Markov processes are dominated by the attractive state \#0 of empty space, 
the entropy rate is trivially zero for the examined toruses.
Nevertheless, the Shannon entropies 
from each state are informative; see Figure \ref{fig:entropy}\refpart{a}. 
We choose to average these Shannon entropies using the leading eigenvector $\vec{v}^{\,(1)}$ 
as the most relevant asymptotic probability distribution,  normalized to sum up to 1:
\begin{equation} \label{eq:entrate}
\sum_{k=1}^N v^{(1)}_{k} \; \sum_{j=0}^N -p_{j,k}\log_2 p_{j,k}.  
\end{equation}
Here $p_{j,k}$ are the transition probabilities (\ref{eq:probable}).
Note that the inner sum includes the transitions to the absorbing \#0 state, 
while the outer sum ignores that state (whose Shannon entropy is zero anyway).
As shown in the \refpart{a}-row 
of Table \ref{tb:summary}, this averaged entropy rate is suspiciously low: 
less than 2 bits. On the other hand, the median Shannon entropy (shown in the subsequent row) of the transitions 
from each state grows significantly for larger toruses, as can be expected.
The disparity of the entropy rate (\ref{eq:entrate}) points to very uneven distribution given 
by the dominant vector $\vec{v}^{\,(1)}$. 
\item[\refpart{b}] The third row in the middle part shows the Shannon entropy of the dominant eigenvector 
$\vec{v}^{\,(1)}$. 
It is indeed low --- at best, barely reaching 3 bits for $T_{10}$ 
--- and not clearly growing. This means that the asymptotic distributions (until the irreversible transition to \#0) 
are highly concentrated on the top 8\,--\,10 most frequent ground patterns, as $2^3=8$. 
This is evident in the percentage of the 10 most frequent patterns in the subsequent row,
and in the density of the top rows in Figure \ref{fig:matrices89}.
\item[\refpart{c}] Unevenness of the probability distribution $\vec{v}^{\,(1)}$ is strikingly reflected 
by the growing scale of magnitude of the probabilities $v^{(1)}_{k}$. As the next row of Table \ref{tb:summary} shows,
the probabilities differ maximally by the factors $\approx 10^{65}$ or $10^{131}$ for $T_9$ and $T_{10}$, 
respectively. We can consider the logarithm of each probability as the {\em Boltzmann entropy} of the corresponding 
ground pattern. This entropy is lesser (more negative) for rarer ground patterns.
This Boltzmannian entropy was considered in the context of the Ehrenfest urn model  \cite{Klein56},
which is a classical Markov process whose dominant eigenvector gives the binomial distribution. 
For comparison, the famous Boltzmann formula $S=k\log V$ expresses gas entropy $S$ in terms of the phase space 
volume $V$ of a macrostate, and the probability of the macrostate is assumed to be proportional to $V$. 

To appreciate differences 
or quotients 
between neighbouring eigenvector components, we consider 
the geometrically averaged relative difference
\begin{equation} \label{eq:jump}
\left( \frac{v^{(1)}_{1}}{v^{(1)}_{N}} \right)^{1/(N-1)}, 
\end{equation}
given in the subsequent row of the table. 
Figure \ref{fig:entropy}\refpart{c} displays the jumps $v^{(1)}_{k}/v^{(1)}_{k+1}$ of consecutive eigenvalue components for $T_9$. 
\item[\refpart{d}] Within the context of Boltzmann's H-theorem \cite{Moran60}, it is interesting to consider 
the changes in the Boltzmannian state entropy under the perturbative transitions between ``Life" patterns. 
The expected entropy increase rate from the $k$-th pattern is
\begin{equation} \label{eq:dbentropy}
\sum_{j=1}^{N} \, P_{j,k} \, \log_{10} \,\frac{v^{(1)}_{j}}{\,v^{(1)}_{k}\,}\,.
\end{equation}
Recall that $P_{j,k}$ is the transition rate from the state \#$k$ to the state \#$j$, as in (\ref{eq:probable}). 
The set of these entropy change rates (for $T_9$) 
is depicted along the horizontal axis of Figure \ref{fig:entropy}\refpart{d}.
The row ``Boltzmann entropy rate" of Table \ref{tb:summary} shows 
the $\vec{v}^{\,(1)\,}$-weighted average of these entropy change rates.
Even if there are just 2--3 patterns with (only slightly!) negative entropy change rate,  
the $\vec{v}^{\,(1)}$-average is negative for the most considered toruses!
This shows that the most frequent patterns are so dominant that the pseudo-Markovian 
process finds itself predominantly at the entropy maximum.
For comparison, this averaged entropy change rate is zero for the Ehrenfest urn model, 
as the entropy changes of any pair of opposite transitions cancel each other out.
As for the H-theorem itself, the expected ensemble entropy increases for Markov processes  \cite{Morimoto63}. 
in the evolution (\ref{eq:m0gensol}) that exponentially approaches the leading eigenvector.
\item[\refpart{e}] The vertical axis in Figure \ref{fig:entropy}\refpart{d} shows the number 
of decays to the empty space \#0 for each pattern on $T_9$. 
These irreversible decays represent entropy rise ``of another level", so to speak.
The $\vec{v}^{\,(1)\,}$-weighted average of the number of $\#0$-decays is nothing else 
but the leading eigenvalue multiplied by $-1$.
\item[\refpart{f}] It is sensible to distinguish entropic decays (to patterns with higher Boltzman entropy) 
and negentropic decays (to patterns with lower entropy). 
The first two rows of the \refpart{f}-part 
of Table \ref{tb:summary} 
compare the aggregate decay rates of entropic and negentropic transitions.
The latter two rows expose 
the overall density difference of non-zero entries in the transition matrix above or below the main diagonal. 
Figure \ref{fig:matrices89}\refpart{b} rather reveals narrow negentropic deviation from the main diagonal.
\item[\refpart{g}] The visual symmetry around the main diagonal in Figure \ref{fig:matrices89}\refpart{b}
indicates that negentropic transitions are often reciprocated by the reverse entropic transitions.
As the \refpart{g}-part 
of Table \ref{tb:summary} shows, the rate of the reverse transition is frequently not lower than 
the negentropic transition, and sometimes those pairs of transitions are balanced exactly.
\item[\refpart{h}] Continuing the theme of entropic vs negentropic transitions, we count the patterns which decay 
only entropically (or inertly), and the maximal negentropic transitions in terms of the \#-rank jump 
or the quotient of $v^{(1)}$-probabilities. 
A good benchmark for negentropic transitions is the quotient
\begin{equation} \label{eq:neg1}
n^2+\lambda^{(1)}
\end{equation}
that is characteristic for the transitions to the patterns which:
\begin{itemize}
\item have a single decay predecessor, with the transition rate 1;
\item have no inert decays. 
\end{itemize}
This follows from (\ref{eq:neg2}).
Appendix Section \ref{sec:regimes} and Figure \ref{fig:entropy2} investigate the negentropic part 
of the transition matrices in greater detail. 
\item[\refpart{i}] Here we count the patterns with just one, two or three predecessors. 
The corresponding eigenvector equations have merely a few terms; see (\ref{eq:neg2}) and 
Sections \ref{sec:ladirect}, \ref{sec:simplified10}.  
\item[\refpart{j}] Here we average the number of live cells in the patterns over the whole set of $N$ patterns 
(with the oscillators represented by the simple average over the phases) or with the $v^{(1)}$-weighting. 
The average density of live cells can be computed after division by $n^2$. 
The discrepancy of the two averages shows again unevenness of the $v^{(1)}$-distribution.
\end{enumerate}

\subsection{Equal components of eigenvectors}
\label{sec:equal}

Ordering the stabilized ``Life'' patterns according to the leading eigenvector
encounters a technical issue, namely possible equality of the eigenvector components.
As an example, consider Figure \ref{fg:parallel} with patterns on $T_9$. 
The delineated two sets of four patterns enjoy ``parrallel'' equalities 
\begin{equation} \label{eq:}
v_{4753}=v_{4754}, \quad v_{4968}=v_{4969}, \quad v_{5065}=v_{5066}, \quad v_{5280}=v_{5281}
\end{equation}
of the eigenvalue $(v_j)_{j=1}^N$ components, because:
\begin{itemize}
\item The transitions among them have the same corresponding frequencies (1, 2 or 4)
within each set of four patterns, as shown by the arrows in the delineated areas.
The arrows are placed at the cells to perturb for the relevant decays, 
and they point (roughly diagonally) toward the 
resulting pattern.
The relevant transitions are not affected by the phase of the Blinkers.
\item These patterns do not perturb into themselves, hence their ``self-decay'' rate equals $-81$.
\item The ``feeding'' patterns in the middle column decay to the delineated corresponding patterns
with the same frequency 1. 
\end{itemize}
For a linear algebra  calculation, let us copy the $v$-components of the delineated 
patterns into two 4-dimensional vectors: 
\begin{equation} \label{eq:equalvc9}
\vec{w}_j=\left( v_{4753+j}, v_{4968+j}, v_{5065+j}, v_{5280+j} \right)^T, \qquad j\in\{0,1\}.
\end{equation}
According to the indicated transitions,  both these vectors satisfy the following linear system:
\begin{align}
\left(\begin{array}{cccc} \lambda+81 & -1 & -4 & 0 \\  -2 & \lambda+81 & 0  & -2 \\  
-1 & 0 & \lambda+81 & -1 \\  0 & -1 & -4 & \lambda+81 \end{array}\right) \vec{w}_j =
\left(\begin{array}{c} v_{4437} \\ \!\!v_{4983}+v_{5165}\!\! \\ 0 \\ v_{5477}  \end{array}\right)\!.
\end{align}
The right-hand side reveals the common ``feeding'' 
patterns. The matrix is invertible when $\lambda\not\in\{-81,-81\pm2\sqrt3\}$.
Then $\vec{w}_0=\vec{w}_1$, particularly for the components of the leading eigenvector $\vec{v}^{\,(1)}$.
This leads to (\ref{eq:equalvc9}). 

In total, there are 49 pairs of patterns on $T_9$ 
with equal eigenvector components. An ordering within these pairs is defined in Section \ref{sec:acriterion}.
Let us refer to such a pair by the notation @$k$, where $k$ is the earlier \#-rank within the pair.
We group these pairs into {\em clusters} by extrapolating the presence of transitions among the patterns 
in these pairs into an equivalence relation. For example, Figure \ref{fg:parallel} displays  
the considered cluster of four pairs. 
We count 16 clusters of paired patterns; they are listed in Appendix Table \ref{tb:parallel}.

This parallelism of patterns becomes abundant 
on $T_{10}$. 
Remarkably,  there are 57 quartets (rather than pairs!) of patterns with the same eigenvector component.
As well, there are 3 triples and 3465 pairs of equal eigenvalue components. 
Let us refer to these pairs, triples or quartets as {\em bunches}.
We extend the notation @$k$ to all bunches, 
with $k$ equal the earliest participating 
\#-rank.  Except for 28 pairs and a triple, the patterns in each bunch are composed of the same islands.
Most frequent island configurations in the bunches are given in Appendix Table \ref{tb:parallel10}.
The exceptional bunches with differing islands in their patterns are these:
\begin{itemize}
\item The pairs @215467, @247360, @393345 consist of (pairwise) similar single islands 
that appear only once in the whole list of ground patterns. See the middle row of Appendix Figure \ref{fg:10x10}.
\item The pair @3440 and the triple @80541 include either still lives or period 2 oscillators.
These patterns consist of similar configurations of two islands. 
\item The 24 pairs
\begin{align} \label{eq:deltan}
& \hspace{-9pt} 
@295524+4325\!+\!13772\!+\!30193\!+\!4726\!+\!5328\!+\!38037\!+\!779\!+\!3335\!+\!13210\!+\!34827\!+\!5808,
 \qquad\nonumber \\ 
 & \hspace{-9pt} 
@451381+ 1760+437+821+2114+37060+2638+7392+1467+5949+401+59.
\end{align}
Here the summands after the first @-ids are equal to the gap differences of subsequent @-ranks 
of the involved patterns. 
\end{itemize}
The bunches are assembled into 232 clusters 
by the immediate transitions among themselves, 
as summarised in Appendix Table \ref{tb:parallel10}.
The size of a cluster is the number of bunches 
with the same eigenvector component. The largest cluster counts 634 pairs of patterns.

\begin{figure}[tbp]
\begin{center}
\begin{picture}(430,278)
\put(0,-17){\includegraphics[width=430pt]{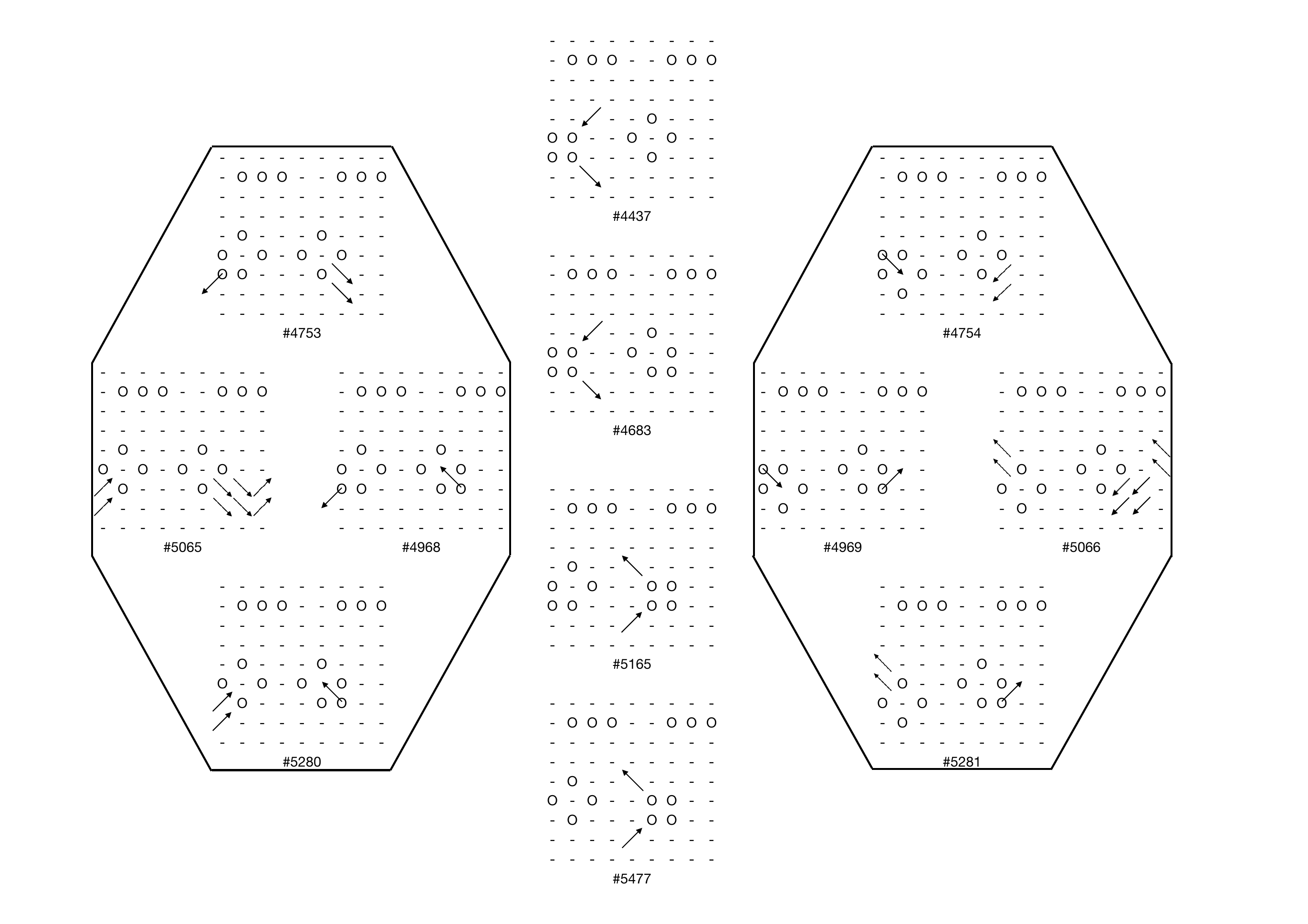}}
\end{picture}
\end{center}
\caption{The cluster of parallel pairs @4753, @4968, @5065, @5280 of patterns on the $9\times9$ torus,
together with their common predecessors \#4437, \#4683, \#5165, \#5477. 
The paired patterns have correspondingly equal eigenvector $\vec{v}^{\,(1)}$ components.
The arrows indicate the relevant transitions, as described in the text.}
\label{fg:parallel}
\end{figure}

Some equalities of eigenvalue components for $T_{10}$ 
follow less trivially than  by strict parallelism of the number of transitions into them. 
For shorthand, let $Eq(k)$ denote the equality $v_k=v_{k+1}$ of the components of the leading eigenvector $\vec{v}^{\,(1)}$,
and let $\Delta v_k$ denote the difference $v_k-v_{k+1}$. Then:
\begin{enumerate}
\item We have $Eq(299849)$ and $Eq(409229)$ from
\begin{equation}
\left(\begin{array}{cc} \lambda+90 & -1 \\  0 & \lambda+89   \end{array}\right) { v_{k} \choose v_{k+1} } =
{ \,v_{\ell} \choose \,v_{\ell} }.
\end{equation}
for $(k,\ell)\in\{(299849,276194),(409229,398202)\}$. 
By subtracting the two equations we get $(\lambda+90)\Delta v_k=0$.
These two pairs are in the list (\ref{eq:deltan}).
\item We have $Eq(511019)$, $Eq(511420)$, $Eq(511479)$ as
\begin{align*} \hspace{-7pt}
\left(\begin{array}{cccccc} \lambda+100 & -2 & -1 & 0 & 0 & 0 \\  -1 & \lambda+99 & 0  & -1 & 0 & 0 \\  
-2 & 0 & \lambda+100 & -1 & -1 & 0 \\  0 & -2 & -1 & \lambda+100 & 0 & -1 \\
0 & 0 & -1 & 0 & \lambda+100 & -1 \\  0 & 0 & 0 & -1 & -1 & \lambda+100 \end{array}\right) \!
\left(\begin{array}{c} v_{511420} \\ v_{511421} \\ v_{511020} \\ v_{511019} \\  v_{511479} \\ v_{511480} \end{array}\right)
\end{align*}
equals the vector $(u_1,u_1,u_2,u_2,u_3,u_3)^T$ with $u_1=v_{510849}+v_{511802}$, 
\mbox{$u_2=v_{510353}+v_{511429}$,} and $u_3=v_{510921}+v_{511834}$. 
It is instructive to consider the matrix sliced into $2\times2$ blocks.
By subtracting pairs of equations, we get a more transparent linear system for the differences:
\begin{align} \hspace{-6pt}
\left(\begin{array}{ccc} \lambda+101 & -1 & 0 \\  
-2 & \lambda+101 & -1 \\  0 & -1 & \lambda+101 \end{array}\right) \!
\left(\begin{array}{c} \Delta v_{511420}  \\ \!-\Delta v_{511019}\! \\  \Delta v_{511479} \end{array}\right)
= \left(\begin{array}{c} 0 \\ 0 \\ 0 \end{array}\right)\!.
\end{align}
\item Two similar eigenvector equations 
imply that $(\lambda+\frac{151}2)v_{j+413133}+3v_{418250-j}$
for \mbox{$j\in\{0,1\}$} evaluate to two sums of 16 pairwise identical terms and 35 equivalent terms, 
including the two terms $\frac32v_{j+420038}+v_{420039-j}$ in both sums
while $Eq(420038)$ holds. Likewise, the expressions
$4v_{j+413133}+(\lambda+80)v_{418250-j}$ for $j\in\{0,1\}$
evaluate to two sums of 7 pairwise identical terms and 16 pairwise equivalent terms. 
This implies $Eq(413133)$ and $Eq(418249)$. 
These two pairs merge two would be largest clusters of sizes 335 and 297 into the monster cluster of size 634.
\item We have $Eq(433952)$ and additionally $\frac12v_{440240}=Eq(441186)$ from
\begin{align*} 
& \textstyle (\lambda+96)v_{441186}=(\lambda+96)v_{441187}=\frac12v_{433953},
\quad (\lambda+96)v_{440240}=v_{433952}, \\
& (\lambda+96)v_{433953}-v_{441186}-v_{441187}=(\lambda+96)v_{433952}-v_{440240}=u_4,
\end{align*}
where $u_4=v_{j+426786}+v_{j+432015}+v_{j+433627}+2v_{j+438521}$ for either $j\in\{0,1\}$. 
This is a part of the third largest cluster of size 168. 
\item Most non-trivially, the following vector is zero:
\begin{align}
\big( & \Delta v_{389017}, \Delta v_{398238}, \Delta v_{405305}, \Delta v_{406492}, 
\Delta v_{415270}, \Delta v_{416253}, \Delta v_{423954}, \qquad\nonumber \\
&\; v_{396519}-v_{396521}-v_{415344}, \; \Delta v_{389286}-v_{423906} \big).
\end{align}
The reason is that this vector is sent to zero by the matrix
\begin{align}
\left( \begin{array}{ccccccccc} 
\!\lambda\!+\!91\! & -1 \\  -4 & \lambda+99 & -1 \\
& -2 & \!\!\lambda\!+\!95\! && -2 &&& -2 \\
&&& \!\!\lambda\!+\!100\!\! & -2 & -4 && -2 \\
&& -1 && \!\!\lambda\!+\!100\! & -6 & -16 \\
&&& -1 & -4 & \!\!\lambda\!+\!100\! \\ &&&& -1 && \!\lambda\!+\!100\! \\
&& -2 &&&&& \!\!\lambda\!+\!99\! & -8 \\ 
&&&&&&& -1 & \!\!\lambda\!+\!95\!
\end{array} \right) \;
\end{align}
according to the eigenvector equations and the easier 
equalities $Eq(372341)$, $Eq(372493)$, $Eq(414804)=Eq(414806)$.  
The matrix with $\lambda=\lambda_1$ is invertible as 
it is strictly diagonally dominant \cite[Corollary 6.2.27\refpart{a}]{Matrix85}. 
These patterns form a part of a cluster of size 22.
\end{enumerate}
The cases of very similar eigenvector equations giving nearly equal eigenvalue components are interesting as well;
see Appendix Sections \ref{sec:aequal} and \ref{sec:simplified10}.

\subsection{Additional criterion for ordering the patterns}
\label{sec:acriterion}

The considered equalities of eigenvector components complicate definition of a unique ranking order of the ground patterns.
We choose the following algorithmic definition for ordering patterns in the {\em bunches}, that is, 
in the pairs, triples or quartets of patterns with the same eigenvalue component.
\begin{definition}
If two patterns $A$, $B$ have equal components of the leading eigenvector: 
\begin{itemize}
\item[\refpart{C1}] Consider the patterns \underline{to which} $A$ and $B$ decay.
The patterns within any bunch 
are considered temporarily as having the same \#-rank,
and the decay rates to them from either $A$ or $B$ are added up. 
\item[\refpart{C2}] Choose the pattern (or the bunch) with the \underline{lowest} \#-rank
to which the (aggregate) decay rates from $A$ and $B$ differ.
\item[\refpart{C3}] The pattern $A$ will precede $B$ in the \#-ranking if the (aggregate) decay rate to the chosen pattern
(or bunch) from $B$ is greater than from $A$. By symmetry, $B$ will precede $A$ if the (aggregate) decay rate 
to the chosen pattern(s) from $A$ is greater.
\item[\refpart{C4}] If $A$ and $B$ decay to the same patterns and bunches with the same decay rates,
the ranking is still unresolved. This does not happen for ground patterns on $T_9$ and $T_{10}$. 
\end{itemize}
\end{definition}
\noindent Section \ref{sec:data} gives references to the final ranking on $T_9$ and $T_{10}$. 
Here is some practical commentary to the algorithmic criterion:
\begin{itemize} 
\item[\refpart{D1}] Since $A,B$ themselves form a pair (or a part of a larger bunch),
their self-decay rates are compared in the context of the same \#-rank. 
If $A,B$ decay to each other, those decay rates are added to the (negative) self-decay rates correspondingly.
For an instructive example, consider $Eq(409229)$ of the case \refpart{i}.
\item[\refpart{D2}] Most often, a pattern (or a bunch) will be chosen
to which only one of the patterns $A,B$ decays. 
Predominantly, we are comparing the decay products of $A,B$ of the lowest rank.
\item[\refpart{D3}] The rationale for the lesser rank of that $A$ or $B$ with the lower residual bottommost ranked decay
is to soften the most negentropic (that is, least entropic) residual decay.
As most patterns decay to the top 10 configurations frequently, 
comparison of the highest ranked decay products is not compelling.
\item[\refpart{D4}] Basically, the algorithmic procedure compares the columns (corresponding to $A$, $B$) 
of the transition matrix from the bottom. 
The rows of $A$, $B$ are the same after the bunch identification in \refpart{C1},
except in the rare cases \refpart{iv}, \refpart{v} of the previous section. 
The same entropic rationale would imply comparing the rows 
from the left, but that would not resolve any additional cases (for larger toruses) probably.
\end{itemize}

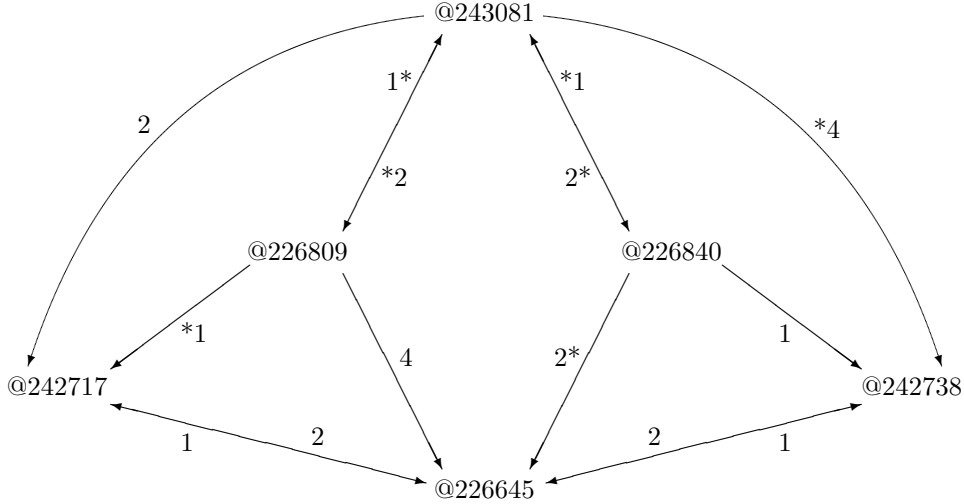
\begin{figure}
\begin{center}
\begin{picture}(350,200)(0,-5)
\put(0,39){@242717}  \put(160,0){@226645}  \put(320,39){@242738}
\put(90,90){@226809}   \put(230,90){@226840}  \put(160,180){@243081}
\put(126,85){\vector(1,-2){37}}  \put(233,85){\vector(-1,-2){37}} \put(147,50){4}  \put(205,50){2*}
\put(126,101){\vector(1,2){37}}  \put(233,101){\vector(-1,2){37}} \put(142,154){1*}  \put(207,154){*1}
\put(126,101){\vector(-1,-2){0}}  \put(233,101){\vector(1,-2){0}} \put(140,118){*2}  \put(209,118){2*}
\put(91,88){\vector(-4,-3){52}} \put(268,88){\vector(4,-3){52}} \put(65,59){*1}  \put(289,59){1}
\put(157,6){\vector(-4,1){118}}  \put(202,6){\vector(4,1){118}} \put(65,18){1}  \put(289,18){1}
\put(157,6){\vector(4,-1){0}}  \put(202,6){\vector(-4,-1){0}} \put(114,20){2}  \put(240,20){2}
\qbezier(156,182)(46,170)(8,50)  \qbezier(201,182)(312,170)(350,50)
\put(8,50){\vector(-1,-3){0}}  \put(350,50){\vector(1,-3){0}}  \put(49,138){2}  \put(302,136){*4}
\end{picture}
\end{center}
\caption{Some parallel transitions between 6 pairs from a cluster of 15 pairs. The transitions labelled by * transform 
the first pattern of a source pair to the second pattern of a pair of decay products  (in the established ordering), 
and the second source pattern is mapped to the first decay product. Parity of the number of these askew *-transitions along 
the triangles or the central quadrangle is preserved under change of ordering within the pairs. Unwelcome *-transitions
are unavoidable in the lower triangles and the central quadrangle.}
\label{fg:obstructed}
\end{figure}

The decay transitions between patterns within a cluster often link the patterns from bunches
neatly into parallel non-crossing orbits, like in the example of Figure \ref{fg:parallel}.
Then it is possible to coordinate the order of patterns within the bunches so that the decay transitions
would conserve the relative orderings in the bunches. 
But establishing the coordinated ordering in large clusters would be computationally heavy.
And this coordination is not always possible generally due to blending of 
the developing (expectedly parallel) orbits. 
Concrete obstructions are discerned in cyclical paths where the implied equivalence of involved patterns 
by decay transitions (ignoring their direction) does not allow separated parallel orbits within bunches. 
This is analogous to one-sidedness of the M\"obius strip \cite[Lect.~14]{fuchs07}. 
Figure \ref{fg:obstructed} demonstrates an example of this impossibility in a convoluted portion a cluster of 15 pairs,
based on parity invariance of askew transitions under attempts to change the ordering in some pairs.
Other instances of obstructions  are simpler:
\begin{itemize}
\item[\refpart{i}] Another triangle of pairs @256423, @262172, @277276 in the same cluster of size 15.
\item[\refpart{ii}] Two quadrangles 
@309315, @323510, @313771, @323878, and 
@309315, @331182, @321194, @331363 (with the common pair @309315) in the cluster of size 51.
\item[\refpart{iii}] Similar two quadrangles @275286, @289181, @279528, @289545, and 
@275286, @296786, @286798, @296959 in the cluster of size 75.
\item[\refpart{iv}] In the cluster of size 113, the path from @434340 to @448869 via the sequence
@177777, @407302, @121526 is not compatible with the shorter paths via @415386 or @421809 
for neat parallel separation.
\item[\refpart{v}] In the cluster of size 168, the two paths from @229500 to @261895 
via @243306 or @243962 are not compatible with the three paths
via @185136, @202375 or @207635. 
\item[\refpart{vi}] In the cluster of size 634, the path from @411418 to @430921 via @401749
is not compatible with the three paths via @413133, @419147 or @419170.
\end{itemize}

\section{Computations}
\label{sec:compute}

Encountered computational issues and employed algorithms are worth attention.
Section \ref{sec:arnas} here describes briefly our algorithmic methods to generate and recognise ground patterns.
Section \ref{sec:ladirect} presents possibilities of direct computation of the eigenvalues and eigenvectors.
Section \ref{sec:laiterate} describes our iterative methods for finding the leading eigenvalue and eigenvector.
Section \ref{sec:data} introduces the generated data that is available on Github \cite{github}. 

\subsection{Generating the ground patterns}
\label{sec:arnas}


The task of finding ground patterns can be computationally intensive especially for larger tori as their number $N$ grows as a quadratic exponential.
Certain optimisations are necessary for the calculation to run in reasonable amount of time. 
We encoded and stored patterns on the torus grid as 128-bit integers where each bit signifies the sate of an individual cell. 
Basic arithmetic operations (modulo the torus side length $n$) were used to relate the 2-dimensional coordinates $(x,y)$ to the 128-bit integer $r$:
\begin{align}
    \text{Coordinates}_{n}(i) & = { r \bmod n \choose \lfloor r / n \rfloor}, \\
    \text{Index}_n{x\choose y} & = x + y n.
\end{align}
This approach allows us to leverage bit-wise operations for hefty optimisation. 
For example, to calculate how many alive neighbours a certain cell has, 
we mask all neighbouring bits for that cell and take the pop-count which amounts to very few CPU instructions. 
These {\em neighbour masks} can be calculated once for each cell and reused indefinitely. 
For masks of cells on the edge of the torus grid, active bits wrap around the edges to ensure 
proper implementation of Conway's ``Life'' on the torus, as seen in Figure \ref{fig:masks}\refpart{a}. 

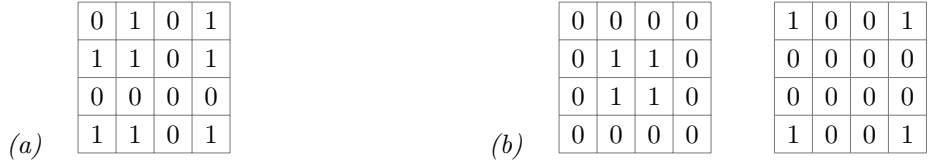
\begin{figure}
\centering \refpart{a} \quad 
\begin{tikzpicture}
\draw[step=0.5cm,color=gray] (-1,-1) grid (1,1);
\node at (-0.75,+0.75) {0};
\node at (-0.25,+0.75) {1};
\node at (+0.25,+0.75) {0};
\node at (+0.75,+0.75) {1};

\node at (-0.75,+0.25) {1};
\node at (-0.25,+0.25) {1};
\node at (+0.25,+0.25) {0};
\node at (+0.75,+0.25) {1};

\node at (-0.75,-0.25) {0};
\node at (-0.25,-0.25) {0};
\node at (+0.25,-0.25) {0};
\node at (+0.75,-0.25) {0};

\node at (-0.75,-0.75) {1};
\node at (-0.25,-0.75) {1};
\node at (+0.25,-0.75) {0};
\node at (+0.75,-0.75) {1};
\end{tikzpicture}
\hspace{90pt}  \refpart{b} \quad 
\begin{tikzpicture}
\draw[step=0.5cm,color=gray] (-1,-1) grid (1,1);
\node at (-0.75,+0.75) {0};
\node at (-0.25,+0.75) {0};
\node at (+0.25,+0.75) {0};
\node at (+0.75,+0.75) {0};

\node at (-0.75,+0.25) {0};
\node at (-0.25,+0.25) {1};
\node at (+0.25,+0.25) {1};
\node at (+0.75,+0.25) {0};

\node at (-0.75,-0.25) {0};
\node at (-0.25,-0.25) {1};
\node at (+0.25,-0.25) {1};
\node at (+0.75,-0.25) {0};

\node at (-0.75,-0.75) {0};
\node at (-0.25,-0.75) {0};
\node at (+0.25,-0.75) {0};
\node at (+0.75,-0.75) {0};
\end{tikzpicture}
\qquad 
\begin{tikzpicture}
\draw[step=0.5cm,color=gray] (-1,-1) grid (1,1);
\node at (-0.75,+0.75) {1};
\node at (-0.25,+0.75) {0};
\node at (+0.25,+0.75) {0};
\node at (+0.75,+0.75) {1};

\node at (-0.75,+0.25) {0};
\node at (-0.25,+0.25) {0};
\node at (+0.25,+0.25) {0};
\node at (+0.75,+0.25) {0};

\node at (-0.75,-0.25) {0};
\node at (-0.25,-0.25) {0};
\node at (+0.25,-0.25) {0};
\node at (+0.75,-0.25) {0};

\node at (-0.75,-0.75) {1};
\node at (-0.25,-0.75) {0};
\node at (+0.25,-0.75) {0};
\node at (+0.75,-0.75) {1};
\end{tikzpicture} \vspace{10pt}
\caption{\refpart{a} The neighbourhood bit-mask for a cell located at the upper left corner. 
The cells with the value 1 are the neighbours. \refpart{b} Two patterns on the torus representing Block still-life.}
\label{fig:masks}
\end{figure}


Any initial configuration on a finite torus settles inevitably into a still-life or an oscillator,  
as mentioned in Section \ref{sec:small}.  
We implemented an algorithm for oscillation detection by storing visited patterns in an unordered set,
and checking whether the current pattern was encountered in the previous steps past of the ongoing ``Life'' bout.


Different perturbations of stabilised patterns may evolve to the same still-life or oscillator. 
If two oscillators have an instance of equivalent (up to the torus symmetries) phase patterns in common, 
the oscillators are considered as the same cycling pattern. But equivalent patterns may be stored as different bit representations;
see Figure \ref{fig:masks}\refpart{b}. 
The torus symmetries are generated by the two-dimensional offset shifts on the torus and the dihedral group $D_4$, 
as mentioned in Section \ref{sec:small}. 
We compute naively all $8n^2$ possible encodings of a found still life or of all phases of a found oscillator
as 128-bit integers, and choose the minimal integer as a canonical representative of that still-live or oscillator. 
Needless this to say, this approach is very inefficient, but it is sufficiently fast for our purposes.

\subsection{Linear algebra: direct methods}
\label{sec:ladirect}

The characteristic polynomial for $M_8$ 
can still be found symbolically in about 35 minutes using {\sf Maple 2023}  on a 2 GHz MacBook Pro laptop. 
The distribution of eigenvectors is depicted in Appendix Figure \ref{fg:egvs}\refpart{b}.
Even a symbolic expression like (\ref{eq:ev5}) of the leading eigenvector might be feasible
thanks to the sparsity in most rows of the transition matrix.
Symbolic factorisation of many eigenvector components can be understood through Cramer's rule \cite[\S 0.8.3]{Matrix85}
and reducibility of sparser maximal minors after deletion of the most sparse row.

The transition matrix for our Markov models becomes more sparse as the size of the torus increases.
The sparsity is rather uniform across the columns; see the thick curve in Figure \ref{fg:colrowdensity}.
There are only 3 columns (\#992, \#7297, \#31900) with more than 70 non-zero entries in $M_{10}$; 
they have at most 153 non-zero entries. 
The distribution of non-zero entries is far less uniform across the rows, as evident in Figure \ref{fig:matrices89}\refpart{b}.
The top 10 rows of $M_{10}$ have the following number of non-zero entries (and density percentages), respectively:
\begin{align} \label{eq:toplength}
483144\ (94.0\%),\quad 413427\ (80.5\%),\quad 506128\ (98.5\%),\quad 509915\ (99.2\%),\quad 512225\ (99.7\%), \nonumber\\
463225\ (90.1\%),\quad 142000\ (27.6\%),\quad 449937\ (87.6\%),\quad 236770\ (46.1\%),\quad 491067\ (95.6\%).
\end{align}
They contain the eight densest rows. 
Nearly 58\% 
of rows have just 4--10 non-zero entries; see the thin curve in Figure \ref{fg:colrowdensity}.
\begin{figure}
\begin{center}
\begin{picture}(380,160)
\put(-4,-2){\includegraphics[width=352pt]{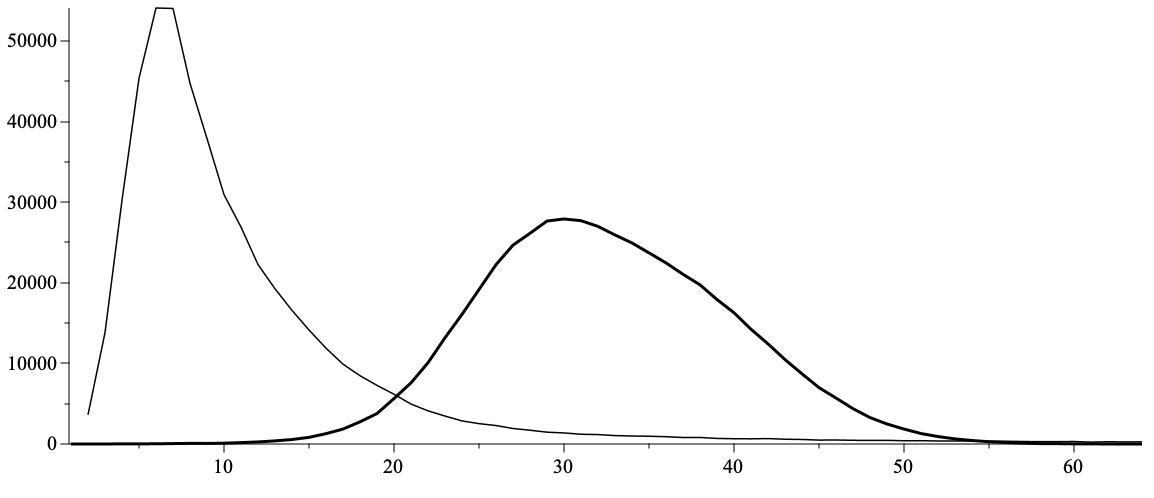}}
\end{picture}
\end{center}
\caption{The distributions of column density (the thick curve) 
and row density (the thin curve) of the transition matrix for the $10\times10$ torus.
The horizontal axis presents the number of non-zero entries in a column or a row;
the vertical axis shows the corresponding number of columns or rows.} 
\label{fg:colrowdensity}
\end{figure}
As most decay transitions follow the direction of increasing Boltzmannian entropy, 
the matrix is particularly sparse under the main diagonal.
The lower-triangular part of $M_{10}$ is up to 8.7 times dense than the upper-triangular part;
compare the bottom two rows in part \refpart{f} in Table \ref{tb:summary}.

There are many pairs of rows of $M_{10}$ that are equal in nearly all entries. 
Differences of the nearly equal rows give simpler equations for computing the eigenvector directly
once a pertinent eigenvalue is known.
The cases of (almost) equal eigenvalue components in Sections \ref{sec:equal} and \ref{sec:aequal} 
provide a bounty of these examples. The number of possible substantial simplifications increases
when the transition matrix is contracted by collapsing the bunches (with equal eigenvalue components)
into single representatives. Section \ref{sec:simplified10} describes the mass of these simplifications.

Of more certain usefulness are those linear relations between eigenvector components that have integer coefficients.
For shorthand, let us call them {\em integer relations}.
They can be used to eliminate some components without proliferating high-precision real numbers 
(or the symbolic eigenvalue $\lambda$) in entries of the adjusted matrix.
Reduced size of the transition matrix allows faster iterations in computation of the leading eigenvalue.
The 
number of integer relations must equal the deficit between the size of the matrix and the algebraic degree of the eigenvalue. 
The integer relations could be dense and involve large coefficients.
Compact integer relations reflect particular relationships between rows of $M_n$ 
and special decay patterns into the involved ``Life'' configurations; 
consider the examples \refpart{iv}, \refpart{v} in Section \ref{sec:equal}. 

Similarly there may be several compact integer relations between components of the leading eigenvector 
$\big(w^{(1)}_k\big)_{k=1}^{N}$ of the transposed matrix $M_{n}^{\,T}$. 
For shorthand, let us refer to this eigenvector as the {\em transposed eigenvector}.
Those compact relations suggest linear combinations of columns (that is, variable changes) toward a block-triangular structure. 
For example, the downstream pattern \#14342 does not decay to other ground patterns, 
yielding $w^{(1)}_{14342}=0$ immediately given that its self-decay rate $-60$ 
is not the leading eigenvalue. The dual implication is that $v^{(1)}_{14342}$ does not influence the other $v^{(1)}$-components,
hence it can be ignored temporarily and computed at the latest stage.

If $M_{n}$ has integer eigenvalues, then the corresponding transposed eigenvectors $\vec{w}$ give integer relations
between the components of $v^{(1)}$, or even between the components of any other eigenvector $v^{(j)}$ not corresponding to those integer eigenvalues.
This follows by the same argument as in Lemma \ref{th:evortho}. Vice versa, the eigenvectors $\vec{v}$ corresponding to integer eigenvalues
give integer relations for $w^{(1)}$ and most other eigenvectors of $M_n{\,T}$. 

We skip the upper indices ${}^{(1)}$ in the rest of this subsection for shorthand.
Here below are observations about integer relations for the considered $M_n$,
beside the bunch equalities of Section \ref{sec:equal}. 
\begin{itemize}
\item There are two independent relations for $M_{5}$ eigenvectors, 
namely $v_7=2v_8$ and $v_6=v_7+2v_9$.
The dual relations for the transposed eigenvector 
are:  $5w_5+2w_7=2w_8$, $8w_7=3w_8+2w_9$.
\item For $M_6$ 
we have two independent relations as well: $v_{13}=2v_{17}$, and
\begin{equation}
7062v_{20} = 1141v_{21} + 367v_{22} + 3941v_{24} + 5810v_{25} - 2354v_{26} + 32728v_{27} - 648v_{28}.
\end{equation}
Instead of the second relation, one could rather utilise the downstream block $w_{27}=0$.
\item There are six independent relations for $M_{7}$. 
The simplest relation is
\begin{align} 
v_{26} & = 2v_{28} - 4v_{38} - 2v_{40} - 4v_{42},
\end{align}
ant the next one can be written as 
\begin{equation*}
\textstyle\! \frac{260}9( v_{14}-2v_{16}+2v_{18}-2v_{20}-v_{21}) +v_{28} 
= v_{29} + 3v_{34} - \frac{529}{12}v_{35} + \frac{34}3v_{38} 
- \frac{80}{9}v_{39} + \frac{989}{27}v_{40} - \frac{82}{9}v_{41} - v_{42}. \nonumber
\end{equation*}
The others  
have coefficients with at least 27 digits. 
As noted before, there is a downstream block of size 4. 
Leaving this block out downsizes the transition matrix most practically.
\item There are no integer relations for $M_{8}$, 
as the characteristic polynomial is irreducible over $\QQ$, which implies the full algebraic degree $N=305$ of the eigenvalues.
The irreducibility follows from incompatible factorisations modulo 17 and 19.
\item A numerical investigation for $M_9$ 
found these simple relations:
$v_{4471}=2v_{4526}$, \mbox{$v_{4772}=2v_{4813}$,} 
$\;v_{2056}\!=\!v_{2059}+v_{2755}$, 
$\;v_{5864}\!=\!v_{5884}+v_{5892}$,
$\;v_{6020}\!=\!v_{6048}+v_{6057}$, 
\mbox{$v_{6176}=v_{6220}+v_{6227}$,}
$v_{6186}=v_{6221}+v_{6252}$, 
$v_{6303}=v_{6328}+v_{6361}$.
No simple relations of these forms were found for the transposed eigenvector components $w_j$. 
\item A numerical investigation for $M_{10}$ 
found the following 11 quotients equal to 2:
\begin{align*}
\!\frac{v_{285655}}{v_{289351}}, \frac{v_{309597}}{v_{313415}}, \frac{v_{375268}}{v_{377475}},
\frac{v_{412539}}{v_{413971}}, \frac{v_{449403}}{v_{450122}}, \frac{v_{453622}}{v_{454304}},
\frac{v_{454246}}{v_{454900}}, \frac{v_{458544}}{v_{459242}}, \frac{v_{504635}}{v_{504862}},
\frac{v_{506208}}{v_{506369}}, \frac{v_{510475}}{v_{510584}}.
\end{align*}
The following equalities  for the transposed eigenvector components $w_j$ 
were found: \mbox{$w_{204}=w_{815}$,} $w_{219}=w_{500}$, $w_{409}=w_{757}$, $w_{441}=w_{1728}$,
$w_{493}=w_{890}$, $w_{616}=w_{618}$.
\end{itemize}

\subsection{Linear algebra: iterative methods}
\label{sec:laiterate}

Direct iterative computation of the leading eigenvalue $\lambda^{(1)}$ and its eigenvector $\vec{v}^{\,(1)}$ 
is problematic for the largest $10\times10$ torus, 
because high precision is required due to the wide range in the $\log_{10}$-orders of magnitude of $\vec{v}^{\,(1)}$-components,
and the components of iterated vectors jump around various magnitudes for many hundreds of iterations.
Apparently, the components do not start to stabilise until their orders of magnitude are settled.
The general iterative solution
\begin{equation} \label{eq:itgensol}
(M_n+\sigma I)^k \, \vec{u}_0=\sum_{i=1}^N C_i\big(\lambda^{(i)}+\sigma\big)^k\,\vec{v}^{\,(i)}
\end{equation}
is analogous to (\ref{eq:m0gensol}); 
here $C_i$ are the coefficients in the expression of an initial vector $\vec{u}_0$ in terms of the eigenvectors $\vec{v}^{\,(i)}$.
The diagonal shift by $\sigma I$ should aim at these conditions: 
\begin{itemize}
\item The shifted leading eigenvalue $\lambda^{(1)}+\sigma$ becomes the largest in the real part.
\item $\sigma\in\ZZ$, so that multiplication by $M_n+\sigma I$ would use simpler arithmetic of multiplication by rational numbers, mostly integers.
\item The eigenvalues $\lambda^{(N)}$ and $\lambda^{(2)}$ with, respectively, the smallest and the next largest real part, appear to be real numbers
and the two candidates for the next largest eigenvalue (in the absolute value) after the shift by $\sigma$. 
We should aim at  $|\lambda^{(2)}+\sigma|\approx |\lambda^{(N)}+\sigma|$, so that the contrast with $\lambda^{(1)}+\sigma$ would be greatest,
giving faster convergence to the direction of $\vec{v}^{\,(1)}$. This requires low-precision estimates of $\lambda^{(2)}$ and $\lambda^{(N)}$;
see the last two rows in the upper part of Table \ref{tb:summary}.
\end{itemize}
We choose $\sigma=73$ for $M_{10}$ accordingly. 
An iteration then diminishes the summands in (\ref{eq:itgensol}) with $k>1$ by at least the factor
\begin{equation}
\frac{\max( |\lambda^{(2)}\!+\sigma|, |\lambda^{(N)}\!+\sigma|)}{\lambda^{(1)}+\sigma} \approx 0.74955
\end{equation}
per iteration, relative to the first term. This should increase the precision by one decimal digit every $8$ iterations,
since $\log_{10} 0.75 \approx 0.125$. 
One iteration takes about 130\,s with Maple 2023 on a Mac Book Pro laptop with 2 GHz processor.
But the coefficients $C_k$ (and components of the eigenvectors) appear to differ wildly by orders of magnitude for generic $\vec{v}_0$,
and various summands in (\ref{eq:itgensol}) likely become dominant through many iterations for most vector components.
This is dramatically exemplified at the end of this section. 
As a contributing detail, Figure \ref{fg:egvs} suggests that many eigenvalues clump close to $\lambda^{(N)}$.

On the other hand, the dominant eigenvector for the transposed matrix of a Markovian process consists of ones only.
Similarly, the dominant transposed eigenvector of 
$M_n^{\,T}$ have components of about the same magnitude. 
For example, the components of this vector for $M_{10}^{\,T}$ 
differ at most by the factor 8.3594 (excepting the zero component \#14342).
The transposed dominant eigenvector is straightforward to compute iteratively with the expected convergence rate. 
This computation gives the leading eigenvalue $\lambda^{(1)}$ to a requisite precision.
Then we can find $\vec{v}^{\,(1)}$ by solving (\ref{eq:basicla}). 
The dimension of this linear system can be reduced (while preserving the diagonally dominant structure) by these measures:
\begin{enumerate}
\item We can utilize known equalities of the $\vec{v}^{\,(1)}$-components, or their simple integer relations; 
see Sections \ref{sec:equal}, \ref{sec:ladirect} and Table \ref{tb:parallel10}.
\item We can throw out the downstream patterns (like \#14342) to which decays are irreversible, 
or by utilizing known simple integer relations for components of the transposed eigenvector $\vec{w}_j$.
\item We can throw out one of the densest rows (not from a downstream block) 
because of the defining linear dependence in (\ref{eq:basicla}). 
The corresponding column can be considered as the non-homogeneous part $\vec{y}$ of the working linear system.
\end{enumerate}
Then one can use either direct solving methods of Section \ref{sec:ladirect}, or an iterative method. 
Orthogonalisation methods do not apply easily, as the most significant eigenvectors do not deviate much
from the leading one by the usual metrics.

The obtained working matrix $\widetilde{M}$ is very sparse under the main diagonal and is close to an upper-triangular matrix, just as $M_{10}$.
We applied the Gauss-Seidel iteration method \cite[\S 11.2]{Golub12} by employing the non-zero entries under the diagonal
with components of the current (older) vector, and combining the entries above the diagonal with freshly (bottom-up) computed
components. Each iteration \mbox{$\vec{u}_k\mapsto \vec{u}_{k+1}$}  is solving 
\begin{equation}
(\widetilde{M}-L)\,\vec{u}_{k+1}=-L\vec{u}_k+\vec{y},
\end{equation}
where $L$ is the submatrix with non-zero entries under the diagonal, and $\widetilde{M}-L$ is the complementary upper-triangular matrix.
The relative discrepancy \mbox{$\max (\vec{u}_{k+1}-\vec{u}_k)/\vec{u}_k$} appears to decrease exponentially as $\approx 0.48355^k$, 
so we gain a binary bit of precision per each iteration. The used provisional ordering of ground patterns was already close to the final ordering.

As an aside, applying this working procedure to the small matrices $M_6,M_7$ (in the final ordering) showed that 
only a few choices in step \refpart{iii} lead to efficient convergence, up to $0.53923^k$ for $M_6$ and $0.37906^k$ for $M_7$.
Choosing the opposite decomposition of $U$ and $\widetilde{M}-U$, where $U$ is the denser part of $\widetilde{M}$ 
above the main diagonal, may give somewhat better convergence rate: 
up to $0.49535^k$ for $M_6$, and $0.34475^k$ for $M_7$.

It is tempting to check the obtained leading eigenvector by iterating the multiplication by $M_{10}+\sigma I$.
Surprisingly, the eigenvector gets gradually perturbed for about 130-140 iterations in the last 20--30 computed decimal digits,
and then gradually converges to the same vector with the presumed accuracy of 140-180 digits in about 200-400 iterations,
apparently depending on the additional accuracy of 20--50 digits of arithmetic operations.
This behaviour was observed several times when the arithmetics precision was increased from various levels
of computed accuracy. Apparently, the linear coefficients for random error terms (in terms of the eigenvectors) 
span similar orders of magnitude, and additional arithmetics accuracy of about 50 digits would be needed for this iterative computation.

\subsection{Produced data}
\label{sec:data}

The computed data is available at the Github repository \cite{github}. It includes the $\vec{v}^{\,(1)}$-ordered lists 
of ground patterns on the considered toruses, together with their island composition,
their components in the leading eigenvector, the decay data for the (sparse) transition matrices,
and additional data for searching and recognising the ground patterns. 
As well, the repository includes files that overview graphically the occurring islands in the ground patterns, 
list the bunches (pairs, triples and quartets) with equal eigenvalue components,
more detailed versions of Figures \ref{fig:matrices89},  \ref{fig:entropy}\refpart{e}, \ref{fig:entropy2}\refpart{a}\,--\refpart{b},
and more complete information regarding Tables \ref{tb:parallel10}, \ref{tb:simplify}.


\section{Interpretation}
\label{sec:interpret}

John Conway gave name ``Life" to his cellular automaton with a clear suggestion of modelling the manifest 
complexity of organic and evolutionary phenomena. The idea that organic complexity ought to arise 
from simple deterministic rules by computational evolution of discrete components 
was taken to a logical pinnacle 
by Stephen Wolfram \cite{Wolfram02}. On the other hand, 
emergence of complex life-like forms is frequently 
modelled by genetic algorithms \cite{genetic17} 
that imitate genetic variation, self-replication and 
selection of phenotypes directly. 
Study of self-replication of explicit programs within computational environments has been revived recently \cite{Arcas24}.

This article offers a stochastic perturbation of Conway's game ``Life'' towards integration with genetic algorithms. 
This minimal introduction of stochasticity to a deterministic cellular automaton appears to be a compelling 
theoretical step, as the dynamical analysis changes substantially, and novel resources for development of complexity come into view.
The presence of stochasticity represents fickleness of the environment, and a sense of adaptation arises.
In contrast to deterministic modelling, a stochastic component evokes statistical resilience of organic and complex phenomena.

Finding the following examples of configurations would amount to 
an outstanding 
prototype for modelling the persistence of organic life:
\begin{itemize}
\item[\refpart{E1}] A non-empty pattern (on an infinite plane grid or a torus) that is immune to any single decay mutation.
If we would allow only perturbations of empty cells to live cells, the pattern \#10 in Figure \ref{fg:lifet5} would already be an example.
\item[\refpart{E2}] An orbit of non-empty patterns that decay only to each other.
\end{itemize}
The leading eigenvalue $\lambda^{(1)}$ of the Markovian process of such an immortal orbit of patterns would equal $0$, 
just as the dominant eigenvalue of the empty space.
Meagre odds for these configurations to arise by the described decays 
are probably similar to 
the chemical odds of abiogenesis. Our perturbed variant of Conway's ``Life'' 
could provide impressive examples of how unlikely exceptional 
arrangements emerge nonetheless.

As an intermediate goal, we may seek orbits of ``Life" patterns whose Markovian processes 
have $\lambda^{(1)}$ 
ever closer to $0$. That would imply very rare decays to the empty space
or to simple overly familiar configurations, or to any other configurations outside the orbit.
A new record $\lambda^{(1)}$ could be estimated from an incomplete knowledge of the whole orbit,
just by assuming that some complex decays to go outside the orbit irreversibly.
The orbits with relatively good $\lambda^{(1)}$ could be leveraged by copying their patterns several times,
say, after doubling the size of the torus. 
As most decay perturbations change patterns locally a little, 
it is expectable that the orbits obtained from such enlarged configurations will often be slightly more stable. 
Similarly, unbounded repetitions and combinations of promising patterns 
may provide a productive ``soup'' for feeding appealing spread of configurations.

\subsection{Life as exceptional chemistry}

The considered ground configurations on small toruses are not convincing as model prototypes for living 
or mere complex artefacts. The leading eigenvalue $\lambda^{(1)}$ should be much closer to $0$ to represent persistence. 
At best, the ground patterns will play a role of background matter or ``quantum'' condensate
of the simplest somewhat stable building blocks. There must probably be a long chain of ingenious 
stabilising arrangements and coincidences before anything like lively, craving and recreating forms
could be recognised. Yet, any conspicuous novelty standing out of the abundance of established interactions
should be appreciated as a potentially pivotal next stage towards a splendid life form.
What is life if not a progression 
of exceptional 
biochemical and physical innovations at each level of organic complexity?

The questions of existence of completely or remarkably stable obits of configurations 
are comparable to Hilbert's tenth problem of solving Diophantine equations.
Determining solvability in general families of algebraic equations in integers 
is known to be undecidable algorithmically \cite{JonesDph80, Poonen08}. 
Consequently, absence of solutions to some Diophantine equations is a paradigmatic example 
of G\"odel's incompleteness. 
Enumeration of ``Life" patterns on an infinite plane which generate sufficiently stable 
decay orbits could analogously be undecidable by the Turing machine, 
and lead to a novel context for G\"odel's incompleteness.
Here ``sufficiently stable'' may mean \refpart{E1}, \refpart{E2}, or that $\lambda^{(1)}$ 
is greater than a predetermined  (negative) bound.
Quite similarly, this enumeration on larger toruses may be an NP-hard problem 
where heuristic considerations would improve the thorough search 
fairly marginally. Extending the analogy further from 
``Life'' to organic chemistry, we may surmise that freshly potent organic structures
emerge as surprise X-factors \cite{xfactor} rather than compactly calculable
or elegantly logical consequences of fundamental physical dynamics. 
They are mostly stumbled upon by evolution, observation, experience, or combinatorial computation.

Where to start searching for exceptional decay orbits of ``Life" patterns? 
Of the considered configurations, the most promising are those in the lower left corner of Figure \ref{fig:entropy}\refpart{d}.
They decay more rarely to the empty space \#0 (see the vertical axis)
and decay most negentropically (see the horizontal axis). 
We would like to have their decay products to stay in the lower left corner as well, but that is much to ask.
Before considering larger toruses, it is worthwhile to consider stable combinations of several top 8 patterns 
that are absent in the list of ground configurations.
For example, Figure \ref{fig:pairsEASN}\refpart{b} indicates a few missing Ship+Blinker pairs by the X-symbol.
Very likely, the missing pairs (or triples, larger sets) decay mostly to ground patterns or to the empty space.
But the configurations with least ``leakage" to ground patterns or the empty space 
would be the best candidates to analyse and leverage further. Conceivably, 
a limited orbit of particular non-ground localizations 
of basic patterns would be stable enough for playing the role of ``genetic'' information to start up
a really wild evolution of the perturbed ``Life'' genre. Other plausible category of stable configurations
for a potent ``soup'' and steady evolution under our perturbations is snaking labyrinths of live cells. 

It would be worthwhile to identify {\em bottleneck} patterns and perturbations 
that expand an orbit of ground or other configurations significantly. 
The orbits where the second eigenvalue $\lambda^{(2)}$ is close to $\lambda^{(1)}$ and $0$
would be interesting as well. The corresponding eigenvector would have significant influence for a long period of time,
similarly to the the apparent convergence instability 
described at the end of Section \ref{sec:laiterate}. 
If this second eigenvalue is real, the eigenvector will have both positive and negative entries likely;
the resulting antagonistic bipolarity may model trophic cascades in ecology \cite[Ch.~6]{Serengeti}.
Otherwise, the complex eigenvector could model cyclicality of organic processes. 

\subsection{A parallel to the Free Energy Principle}

The irreversible attraction to the empty space state under our perturbations is comparable with the 
second law of thermodynamics. In practical thermodynamic computations, maximisation of global entropy is often replaced
by minimisation of {\em free energy} that is defined for a bounded system under certain conditions \cite[\S 6.3]{Sekerka15}.
For example, Gibbs free energy is defined for a system under constant pressure and temperature.

The role of thermodynamics in actual life processes is not quite settled. 
Since mid-2000s, Karl Friston and colleagues developed a theoretical framework \cite{Friston10}, \cite{SchrodingerFE}
for explaining organic and cognitive processes in terms of minimisation 
of certain statistically defined free energy. 
The formulated {\em Free Energy Principle} characterises cognition, homeostasis and other vital function
as (primarily Bayesian) inference about the environment under which surprise, 
encountered uncertainty or representational errors are minimised.
In a sense, the brains and organisms mimic the second of law of thermodynamics in their activity
of maintaining homeostasis, satisfying drives, or matching inferences and actions with reality. 
The tendency of an organism to seize opportunities and satisfy own needs is not unlike 
the entropic drift toward disorder whenever barriers abate. 

Cognitive and organic activities are perhaps not particularly coupled to classical thermodynamic restrictions 
(if only ample energy flows are present), but are rather independent and even competing dynamical processes of convergence. 
The drives of alternative Free Energies can be seen as similar to subdominant vectors of our Markov processes 
of ``Life" perturbations. Eventually the second law of thermodynamics will prevail, just as the torus patterns will decay 
to empty space before long because $\lambda^{(1)}<0$. But lively dynamics could go on quite for some time, 
as convergence complications in our models attest.

The second law of thermodynamics is not principally distinguished from organic vigorous attractions in one meaning: 
it requires certain (though prevalent) conditions as well. For example, Boltzmann's kinetic derivation for ideal gases 
uses the {\em Stosszahlansatz} assumption \cite[Ch.~2,\,17]{Dorfman99} 
of no correlation between the velocity distributions of different molecules:
 $\varphi(\vec{v}_1,\vec{v}_2, t) = \varphi(\vec{v}_1, t) \, \varphi(\vec{v}_2, t)$.
This {\em molecular chaos} is an adequate condition for the fabled entropy increase!
Other derivations of the second law of thermodynamics \cite[\S 3.6]{challenge2nd}, \cite[\S 5]{Uffink07}
are either asymptotic conclusions or are  based on statistics of Gibbs ensembles, 
with no dynamical supervenience coupling to the underlying dynamics. After all, the Loschmidt paradox 
of time reversal symmetry \cite[\S 4.3]{Uffink07} is actual, and fluctuations are ever possible.

\appendix
\section{Appendix}

This appendix contains a series of supplementing figures and tables. Most of them are referenced from the main text.
Besides, the following topical sections follow: 
\begin{itemize}
\item Section \ref{sec:stable5} and Figures \ref{fg:5x5e}, \ref{fg:4x4} describe all still-lives and oscillators 
on the $5\times5$ and $4\times4$ toruses, not just ground patterns. This supplements Section \ref{sec:5x5}.
\item Section \ref{sec:islands} and Tables \ref{tb:islandcomp}\,--\,\ref{tb:patterns10a} 
analyse the islands of connected live cells in the ground patterns, and survey encountered configurations of the islands.
Figures \ref{fig:gliderpairs10}, \ref{fig:pairsEASN} give examples of configurations of 
two islands on $T_{10}$. This extends Section \ref{sec:ground}.
\item Section \ref{sec:regimes} takes a closer look at negentropic decays. 
It discerns a few regimes of negentropic pattern generation in Figure \ref{fig:entropy2} and Tables \ref{tb:deficit9}, \ref{tb:deficit10}. 
This supplements Section \ref{sec:entropic}.
\item Section \ref{sec:aequal} discusses instances of nearly equal eigenvalue components, 
arising typically from a series of very similar eigenvalues equations. This extends Section \ref{sec:equal}.
\item Section \ref{sec:simplified10} and Table \ref{tb:simplify} present more general simplifications of similar eigenvector equations. 
This extends Section \ref{sec:ladirect}.
\end{itemize}

\subsection{The full set of stabilised patterns on the $5\times5$ torus}
\label{sec:stable5}

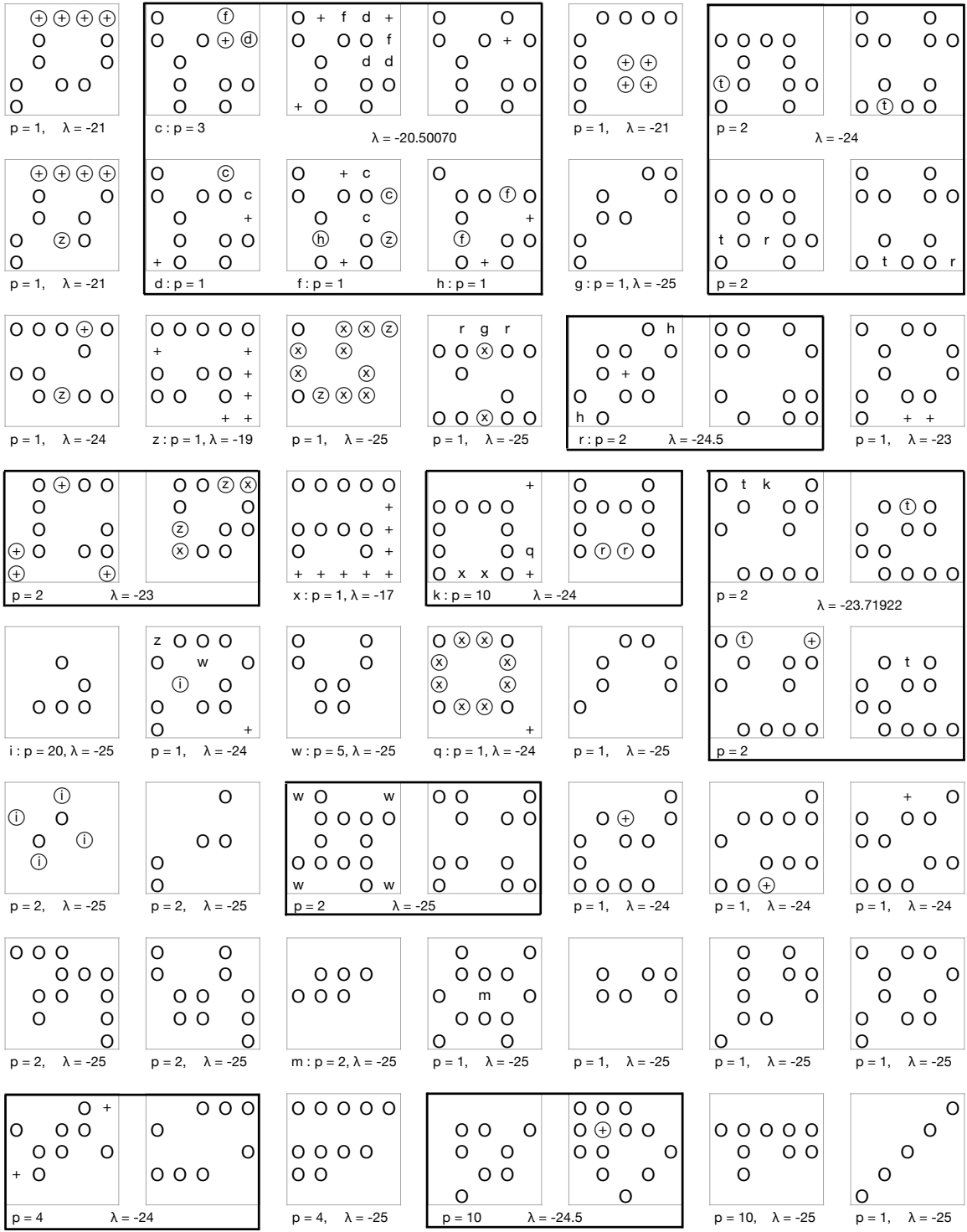
\begin{figure} 
\begin{center}
\begin{picture}(380,542)(12,14)
\put(-14,-22){\includegraphics[width=432pt]{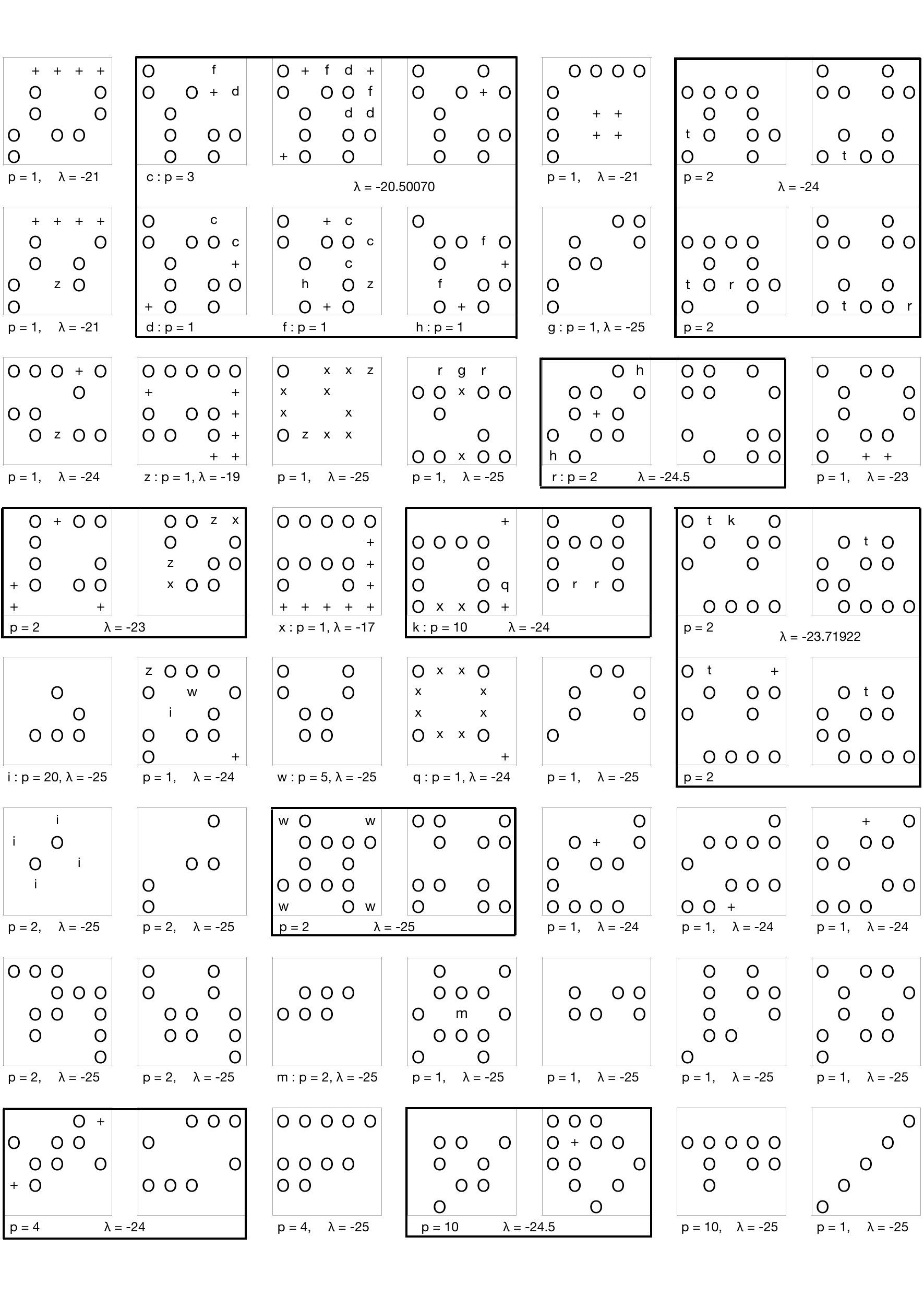}}
\multiput(2.6,555.8)(10.2,0){4}{\circle7}  
\put(86.1,546.2){\circle7}  \put(96.7,546.5){\circle7}  \put(86.1,556.8){\circle7}  
\multiput(265,526)(10.2,0){2}{\circle7}  \multiput(265,536)(10.2,0){2}{\circle7}  
\put(308,526.4){\circle7}  \put(381,516.4){\circle7} 
\multiput(2.6,486)(10.2,0){4}{\circle7}  \put(12.9,456){\circle7}  \put(86.2,486){\circle7} 
 \put(128.6,456.6){\circle7} \put(159.2,456.3){\circle7}  \put(159.2,476){\circle7}   
\put(191.8,456.5){\circle7} \put(212.2,476.5){\circle7} 
\put(23,415.8){\circle7}  \put(12.9,385.9){\circle7}  
\multiput(139,415.9)(10.2,0){3}{\circle7}  \multiput(128.8,385.9)(10.2,0){3}{\circle7} 
\multiput(118.6,405.9)(20.4,0){2}{\circle7}  \multiput(118.6,395.9)(30.6,0){2}{\circle7}  
\multiput(202,405.9)(0,-30){2}{\circle7} 
\multiput(-7.4,315.8)(0,-10){2}{\circle7}  \put(12.8,345.6){\circle7} \put(33.1,305.7){\circle7}   
\multiput(65.8,325.8)(0,-10){2}{\circle7}  \multiput(86.2,345.8)(10.1,0){2}{\circle7}  
\multiput(254.8,315.8)(10.1,0){2}{\circle7} \put(391.2,336){\circle7}  
\put(65.8,256){\circle7} \put(318,275.8){\circle7} \put(348.4,275.6){\circle7} 
\multiput(191.8,275.8)(10.1,0){2}{\circle7}  \multiput(191.8,245.8)(10.1,0){2}{\circle7}  
\multiput(181.6,265.8)(0,-10){2}{\circle7}  \multiput(212.2,265.8)(0,-10){2}{\circle7}  
\multiput(-7.4,195.8)(20.4,10){2}{\circle7}  \multiput(2.6,175.9)(20.4,10){2}{\circle7} 
\put(265,195.6){\circle7} \put(328,165.4){\circle7}  
\put(254.8,55.2){\circle7} 
\end{picture}
\end{center}
\caption{The non-empty stable or cyclic patterns on the $5\times 5$ torus absent in Figure \ref{fg:lifet5}. 
The three larger rectangles (along the top or right edge) enclose strongly connected 
orbits of four ({\sf c, d, f, h}) or two patterns. Smaller rectangles enclose phases of a periodic pattern, when the phases have different 
perturbation transitions. Only the transitions between these patterns are shown, including the inert +\,transitions. 
Transitions within the smaller strongly connected orbits (along the right edge) are indicated by {\sf t}; 
other perturbation products are labeled {\sf g, i, k, m, q, r, w, x, z}. 
The oscillation period {\sf p} of the patterns, and the dominant eigenvalue {\sf $\lambda$} of the strongly connected orbits are given.} 
\label{fg:5x5e}
\end{figure}
\begin{figure} 
\begin{center}
\begin{picture}(380,252)(12,14)
\put(-28,-29){\includegraphics[width=464pt]{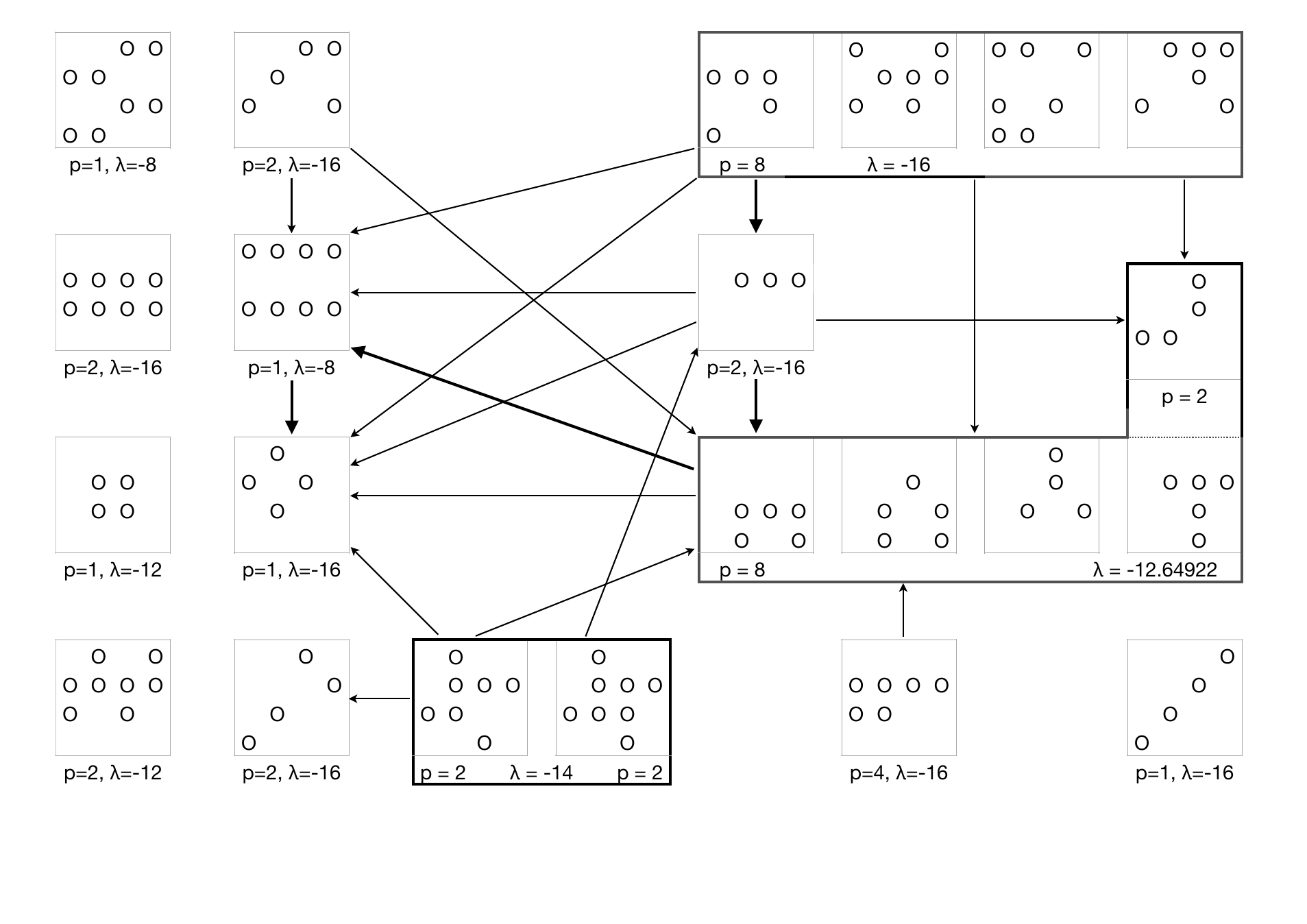}}
\end{picture}
\end{center}
\caption{The non-empty stable or cyclic patterns on the $4\times 4$ torus, and irreversible transitions between them.
Two oscillators of period 8 are presented by enclosed rows of their four different phases (on the right);
one of them forms a strongly connected orbit with a period 2 oscillator. 
The only other strongly connected orbit is formed by two oscillators of period 2 (in the last row).
The longest transition chain (between 5 patterns or strongly connected orbits) is distinguished with bold arrows.} 
\label{fg:4x4}
\end{figure}

\begin{figure}
\begin{center}
\begin{picture}(380,206)(-18,4)
\put(1,-36){\includegraphics[width=358pt]{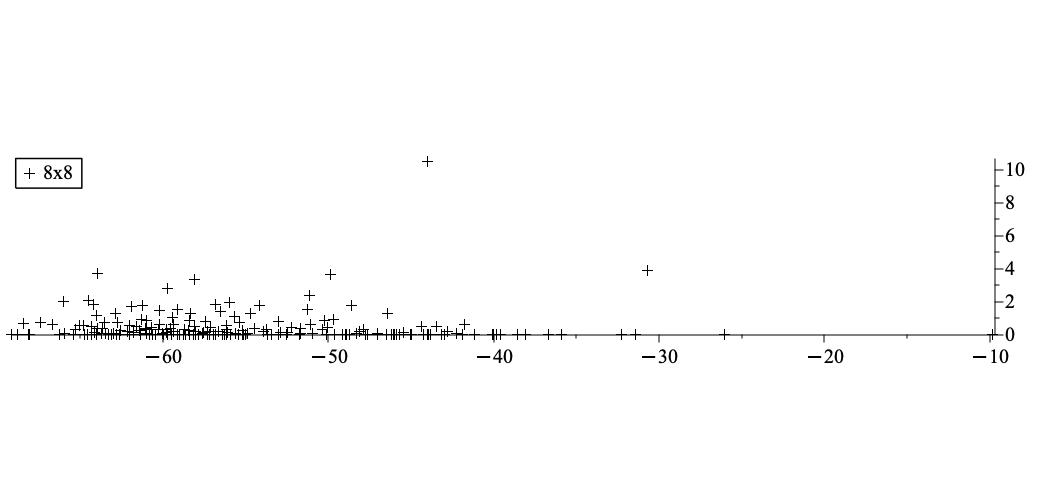}}
\put(1,79){\includegraphics[width=356pt]{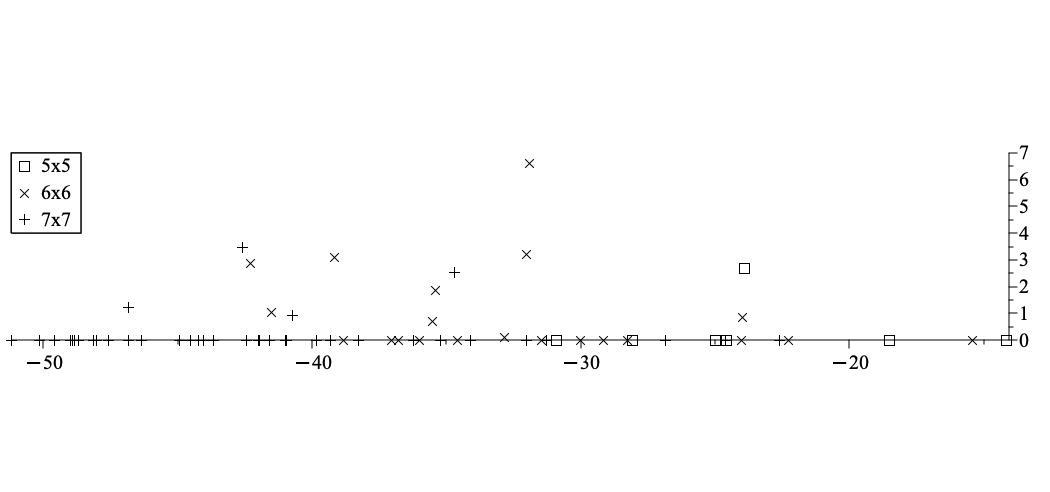}}
\put(-20,116){\refpart{a}}  \put(-20,4){\refpart{b}}
\end{picture}
\end{center}
\caption{\refpart{a} Distribution of the eigenvalues $\lambda$ with Im$\;\lambda \ge 0$ for the $n\times n$ torus
with $n\in\{5,6,7\}$. \hfill \\ 
\refpart{b} Distribution of the eigenvalues $\lambda$ with Im$\;\lambda \ge 0$ for the $8\times 8$ torus.} 
\label{fg:egvs}
\end{figure}

\begin{table} \small\hspace{-20pt}
\begin{tabular}{@{}rcrlccccccl@{}}
\hline
\# & \multicolumn{1}{l}{\;Pattern} & \hspace{-11pt}Period & The leading  & \multicolumn{7}{l@{}}{\!Transitions to:\,\dotfill} \\
 &  \multicolumn{1}{l}{\;name}  & & eigenvector & \!itself & \!\!\#0 & \!\!\#1 & \!\!\#2 & \!\!\#3 & \!\!\#4 & the rest (\#:\,rate)
\\ \hline
1 & Beehive & 1 & 103033038173. & 11 & 34 & --- & 0 & 0 & 0 & 6:\,4 \\ 
2 & Boat & 1 & 65135438716.6 & 16 & 10 & 8 & --- & 0 & 2 & 5:\,1;\, 6:\,2;\, 7:\,6;\, 9:\,4 \\ 
3 & Blinker & 2 & 42376644887.2 & 24 & 15 & 10 & 0 & --- & 0 \\ 
4 & Block & 1 & 35554325891.2 & 21 & 16 & 8 & 4 & 0 & ---  \\ 
5 & Tub & 1 & 27392934210.6 & 24 & 5 & 4 & 16 & 0 & 0 \\ 
6 & Pond & 1 & 25445873585.1 & 5 & 36 & 0 & 0 & 0 & 0 & 8:\,8 \\ 
7 & Loaf & 1 & 22202296185.9 & 6 & 35 & 0 & 0 & 0 & 0 & 8:\,4;\, 9:\,1;\, 10:\,3 \\ 
8 & Bi-block & 1 & 19061143999.9 & 11 & 28 & 6 & 0 & 4 & 0  \\ 
9 & Ship & 1 & 15375808791.6 & 8 & 23 & 8 & 6 & 0 & 0 & 7:\,4 \\ 
10 & two blocks & 1 & 4418537379.20 & 11 & 16 & 0 & 0 & 4 & 12 & 11:\,4;\, 12:\,2 \\ 
11 & beehive + block & 1 & 875770592.800 & 5 & 20 & 10 & 2 & 4 & 6  & 6:\,2 \\ 
12 & boat + block & 1 & 445833936.844 & 6 & 13 & 0 & 4 & 8 & 5  & 8:\,2; 10:\,3; 11,13,14:\,2; 15,16:\,1 \\ 
13 & loaf + block & 1 & 42149029.0867 & 4 & 20 & 4 & 4 & 0 & 6  & 5:\,1;\, 6:\,2;\, 7:\,4;\, 15,17:\,2 \\ 
14 & beehive + boat & 1 & 34471493.1712 & 0  & 25 & 6 & 6 & 1 & 5 & 7,8:\,1;\, 11:\,2;\, 18,20:\,1 \\ 
15 & ship + block & 1 & 26001020.7037 & 6 & 16 & 4 & 0 & 0 & 5 & 7:\,4;\, 10:\,2;\, 12:\,6;\, 13:\,2;\, 19:\,4 \\ 
16 & tub + block & 1 & 24244742.1141 & 8 & 1 & 0 & 0 & 8 & 8 & 11,12:\,4;\, 18:\,8 \\ 
17 & Glider & 28 & 9057500.87575 & 17 & 24 & 2 & 1.5 & 0 & 2 & 6:\,1;\, 7:\,1.5 \\ 
18 & beehive + tub & 1 & 8997227.91763 & 1 & 25 & 5 & 0 & 0 & 0 & 5:6;\;7:4;\;8,14:2;\;20:1;\;21:2;\;23:1 \\ 
19 & Toad & 2 & 7768179.60210 & 13 & 23 & 1 & 0 & 2 & 6 & 5,6,8,24:\,1 \\ 
20 & two beehives & 1 & 1647981.15313 & 0 & 33 & 12 & 0 & 0 & 0 & 22:\,4 \\ 
21 & pond + tub & 1 & 681895.096053 & 0 & 29 & 0 & 0 & 0 & 0 & 5:\,8;\, 6:\,4;\, 7:\,8 \\ 
22 & \!pyramide\,7/4/2+1 & 14 & 558374.756386 & 14.5 & 12.5 & 3.5 & 1 & 1 & 10.5  & 6:\,0.5;\, 7:\,1;\, 8:\,1.5;\, \\ 
&&&&&&&&&& \qquad 26,27:\,1;\, 28,29:\,0.5 \\ 
23 & beehive+3+1+3 & 1 & 354376.552812 & 1 & 32 & 4 & 4 & 0 & 3 & 7:\,2;\, 8:\,1;\, 25:\,2  \\ 
24 & Mango & 1 & 348663.178392 & 4 & 32 & 8 & 0 & 0 & 0 & 7:\,2;\, 8:\,1;\, 17:\,2 \\ 
25 & Block on beehive & 1 & 36554.5961507 & 7 & 22 & 4 & 0 & 0 & 0 & 6:\,8;\, 7,8:\,4 \\ 
26 & Aircraft carrier & 1 & 30364.7857483 & 8 & 26 & 7 & 0 & 0 & 4 & 30:\,4 \\ 
27 & pyramide\,7/5/3 & 28 & 27137.2444947 & 4.5 & 23.5 & \!4.25\! & 0.5 & 1.5 & \!2.75\! & 6:\,0.75;\, 7:\,1.5;\, 8:\,3.25;\,  \\ 
&&&&&&&&&& \qquad 17,22:\,1;\, 31:\,4.5 \\ 
28 & block + blinker & 2 & 15336.9525854 & 8 & 18.5 & 2 & 1 & 3 & 9 & 8,9:\,1;\, 11:\,1.5;\, 29,34:\,2 \\ 
29 & beehive + blinker & 2 & 12652.8240788 & 1.5 & 21 & 7.5 & 0 & 7 & 4 & 7:\,3; 8:\,0.5; 17:\,1; 20:\,1.5; 35:\,2 \\ 
30 & Long boat & 1 & 7519.04139119 & 7 & 21 & 4 & 0 & 0 & 4 & 9:\,2, 24:\,4;\, 32:\,3;\, 33:\,4 \\ 
31 & \!pyramide\,7/5/3+1 & 28 & 6465.04846421 & 7.5 & 20 & 6 & 0.5 & 2 & 1 & 7:\,2; 8:\,1; 19:\,0.5; 22:\,3; 27:\,5.5 \\ 
32 & Barge & 1 & 1981.59261354 & 15 & 0 & 0 & 4 & 0 & 0 & 6:\,2;\, 7:\,16;\, 24:\,4;\, 30:\,8 \\ 
33 & Long ship & 1 & 1412.33254373 & 5 & 10 & 8 & 2 & 8 & 0 & 6:\,8;\, 30:\,6;\, 36:\,2 \\ 
34 & boat + blinker & 2 & 1380.65235558 & 3.5 & 16 & 3 & 5 & 5 & 2.5 & 6:1; 9,10:0.5; 14:1.5; 17,20:0.5;  \\ 
&&&&&&&&&& \qquad 28:2; 29:3.5; 37:2; 38,39:1 \\ 
35 & pond + blinker & 2 & 958.950775673 & 0 & 26 & 2 & 0 & 4 & 0 & 6:\,3;\, 7:\,8;\, 8:\,2;\, 17:\,4 \\ 
36 & Beacon & 2 & 132.062246520 & 5 & 16 & 4 & 3 & 0 & 15  & 7:\,4;\, 13,33:\,1 \\ 
37 & ship + blinker & 2 & 106.827755014 & 0 & 27 & 4 & 1 & 0 & 3 & 7:\,1;\, 9:\,4;\, 17:\,1;\, 34:\,6;\, 39,41:\,1 \\ 
38 & tub + blinker & 2 & 71.6826326889 & 7 & 9 & 1 & 0 & 6 & 2 & 5:\,8;\, 7,14:\,1;\, 18:\,3;\, 27,28:\,1;\,  \\ 
&&&&&&&&&& \qquad 29:\,2;\, 34:\,4;\, 40:\,3 \\ 
39 & loaf + blinker & 2 & 57.7613456355 & 0 & 25 & 0 & 0 & 10 & 3  & 7:\,7;\, 8,10,17,37:\,1  \\ 
40 & boat + blinker & 2 & 9.19444584616 & 3 & 13 & 3 & 3 & 4 & 3 & 7:\,2;\, 10:\,1;\, 14:\,3;\, 17:\,2;\, 28:\,1;  \\ 
&&&&&&&&&& \qquad 29:\,4;\, 38:1;\, 39:\,4;\, 42:\,2 \\ 
41 & \!orthogonal blinkers\!\! & 2 & 5.50974014962 & 7 & 14 & 2 & 2 & 12 & 0 & 9:\,1;\, 17:\,2;\, 20:\,1;\, 29:\,6;\, 32:\,2 \\ 
42 & block + blinker & 2 & 1.0 & 8 & 16 & 0 & 0 & 8 & 3 &  11:\,6;\, 29:\,8
\\ \hline
\end{tabular}
\caption{The ground patterns, the dominant eigenvalue, and the transition data for the $7\times 7$ torus.} 
\label{tb:7x7a}
\end{table}

\begin{table} \small \vspace{-60pt} \def\arraystretch{0.96}
\centering
\begin{tabular}{@{}lrrrrrr@{}}
\hline
Torus size & $10\times10$ & $9\times9$ & $8\times8$ & $7\times 7$ & $6\times6$ & $5\times 5$ 
\\ \hline
Tub & 30.58: 1 & 33.07: 1 & 35.20: 1 & 7.58: 5 & 3.32: 5 & 1.60: 8 \\
Boat & 21.85: 2 & 22.61: 2 & 26.13: 2 & 18.02: 2 & 14.41: 3 & 17.34: 4 \\
Beehive & 15.58: 3 & 17.96: 3 & 15.09: 3 & 28.50: 1 & 25.37: 2  & 19.18: 3 \\ 
Block & 5.01: 5 & 10.17: 4 & 4.20: 5 & 9.84: 4 & 42.44: 1 & 20.65: 1  \\ 
Loaf & 4.47: 6 & 4.12: 6 & 5.38: 4 & 6.14: 7 & 5.12: 4 &  \\
Ship & 2.92: 7 & 2.80: 7 & 3.59: 6 & 1.22: 9 & 3.21: 6 & 3.19: 7 \\
Pond & 2.49: 9 & 2.61: 8 & 2.72: 8 & 7.04: 6 & 1.97: 7 & 19.45: 2 \\ 
Blinker ($p\!=\!2$) & 8.75: 4 & 5.75: 5 & 2.73: 7 & 11.72: 3 & 0.14:\,12 & 0.36: 9 \\
Glider & 28 & 27 & 14 & 17 & 0.62:\,10 & \\
Light weight spaceship & 1868 
& 1396 & 259 & & & \\
Midweight spaceship & 3452 & 1973 & 285 & & & \\
Heavy weight spaceship & 6600  & 2652 & & & & \\
Toad ($p\!=\!2$) & 160 & 107 & 55 & 19 & & \\
Beacon ($p\!=\!2$) & 35 & 272 & 48 & 36 & & \\
Clock ($p\!=\!2$) & 13450 & 3899 & & & & \\ 
Bipole ($p\!=\!2$) & 11712 & 3782 & &  & & \\
Octagon-2 ($p\!=\!5$) & 43990 & 498 & & & & \\ 
Bi-block & 2.69: 8 & 53 & 18 & 4.25: 8 & 15 & \\
Block on beehive & 15 & 153 & 100 & 25 & & \\
Barge & 29 & 115 & 36 & 32 & 28 & 14.32: 5 \\
Long barge & 157 & 78 & 60 & & 16 & \\
Long boat & 49 & 152 & 67 & 30 & 22 & \\
Long ship & 120 & 257 & 121 & 33 & 24 & \\
Very long boat & 316 & 165 & 103 & & 21 & \\
Very long ship & 800 & 305 & 155 & & 25 & \\
Aircraft carrier & 1187 & 402 & & 26 & & \\
Mango & 76 & 227 & 73 & 24 & 26 & \\
Eater & 966 & 1028 & 256 & & & \\ 
Half-bakery & 925 & 863 & 77 & & & \\ 
Bi-loaf 2 & 3925 & 3151 & 199 & & & \\ 
Boat-tie & 4429 & 411 & 175 & & & \\ 
Boat-tie-ship & 2086 & 330 & 198 & & & \\ 
Ship-tie & 1139 & 224 & 224 & & & \\ 
Paperclip & 1733 & 959 & 125 & & & \\ 
Dead spark coil & 6280 & 2892 & 190 & & & \\ 
Krake & 7005 & --- & 124 & & & \\ 
Hat & 3333 & 593 & & & & \\ 
Loop & 417 & 1131 & & & & \\ 
Snake & 6504 & 2959 & & & & \\ 
Shillelagh & 867 & 3929 & & & & \\ 
Integral sign & 9059 & 3675 & & & & \\ 
Block on table & 4583 & 2318 & & & & \\ 
Mirrored table & 6280 & 3694 & 220 & & & \\ 
Cis-mirrored bun & 97 & 1580 & 304 & & & \\ 
Cis-mirrored bookend & 349 & 2664 & 93 & & & \\ 
Trans-mirrored bookend & 367 & 3240 & 96 & & & \\ 
mirrored 2 snakes & 8336 & 4475 & 300 & & & \\ 
Cis-snake on bun & 16541 & 4661 & 303 & & & \\ 
Tub with long tail & 26747 & 5671 & & & & \\ 
Boat with long tail & 14795 & 5405 & & & & \\ 
Boat with hooked tail & 913 & 1528 & & & & \\ 
Hook with tail & 13011 & 3190 & & & & \\ 
Long hook with tail & 198135 & 2118 & & & &  \\ 
Long snake & 40272 & & & & & \\ 
Long integral & 58773 & & & & & \\ 
Mirrored cap & 23778 & & & & & \\ 
\hline
\end{tabular}
\caption{Rankings of some standard patterns on various toruses, 
with percentages for top $10$ patterns.} 
\label{tb:patterns}
\end{table}

\begin{figure} \centering
\begin{picture}(440,130)
\put(-2,-140){\includegraphics[width=440pt]{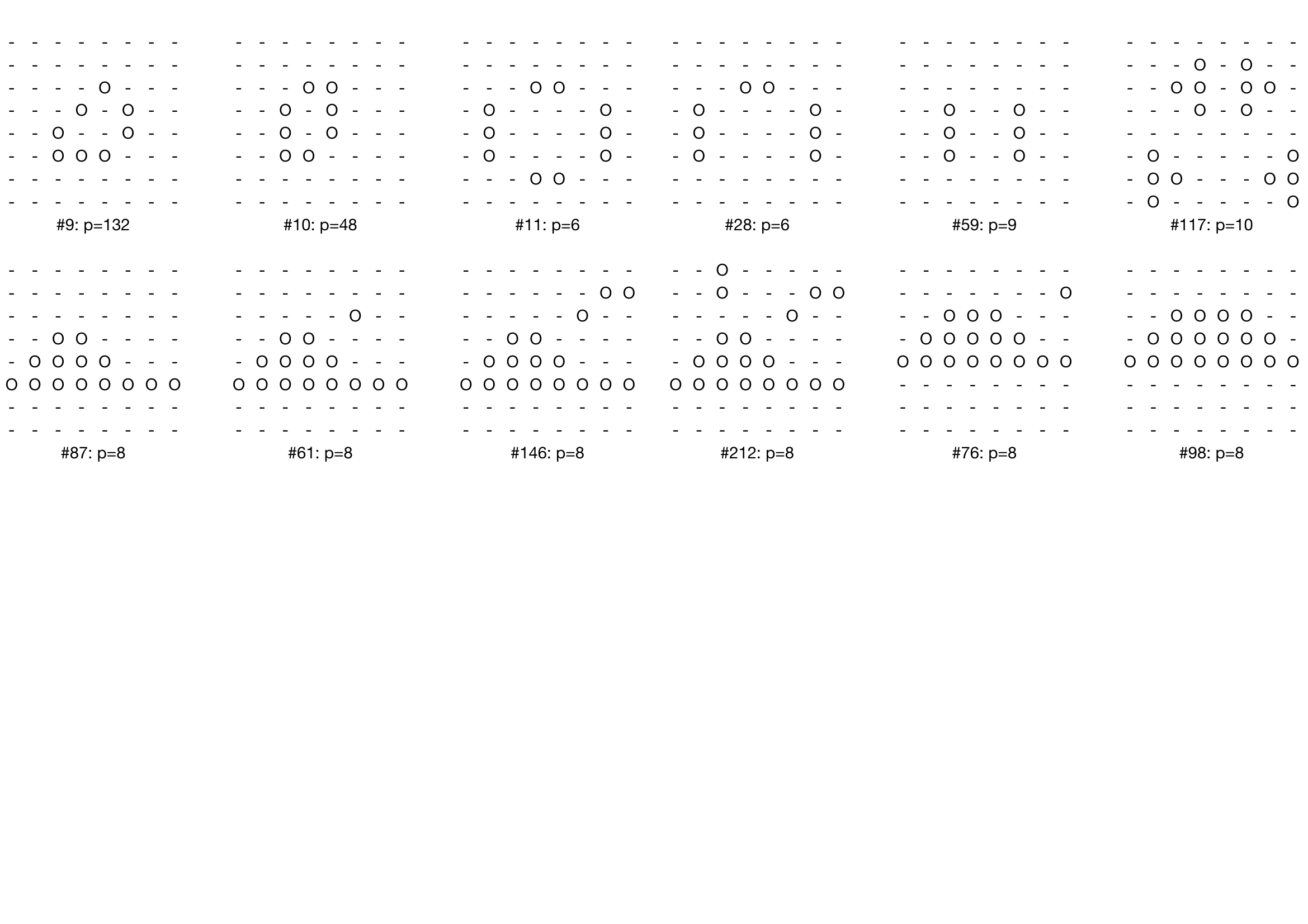}}
\end{picture}
\caption{Examples of oscillating patterns on the $8\times 8$ torus.} 
\label{fg:8x8}
\begin{picture}(440,225)
\put(-2,-114){\includegraphics[width=440pt]{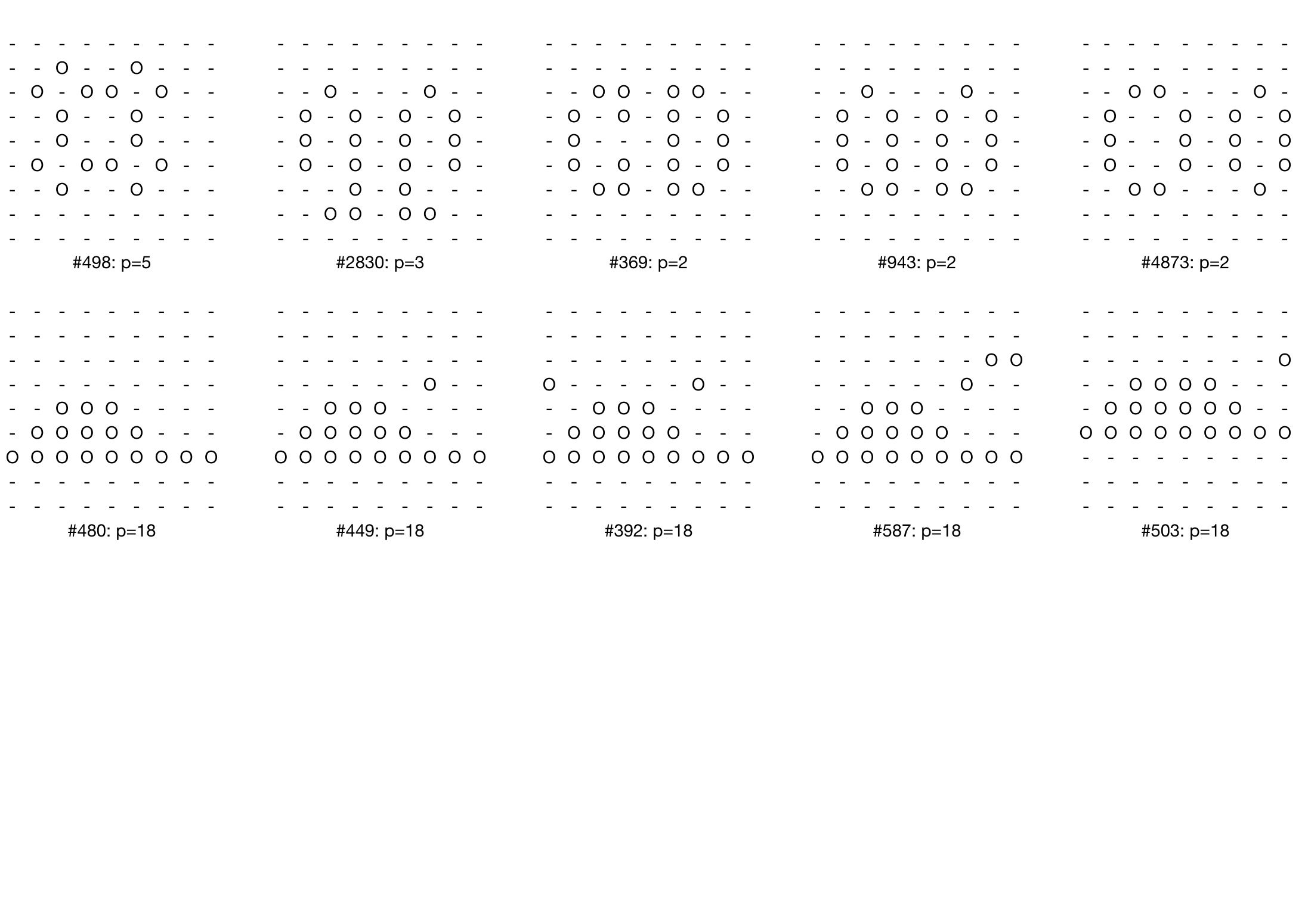}}
\end{picture}
\caption{Examples of oscillating patterns on the $9\times 9$ torus.} 
\label{fg:9x9}
\begin{picture}(440,160)
\put(-3,-180){\includegraphics[width=440pt]{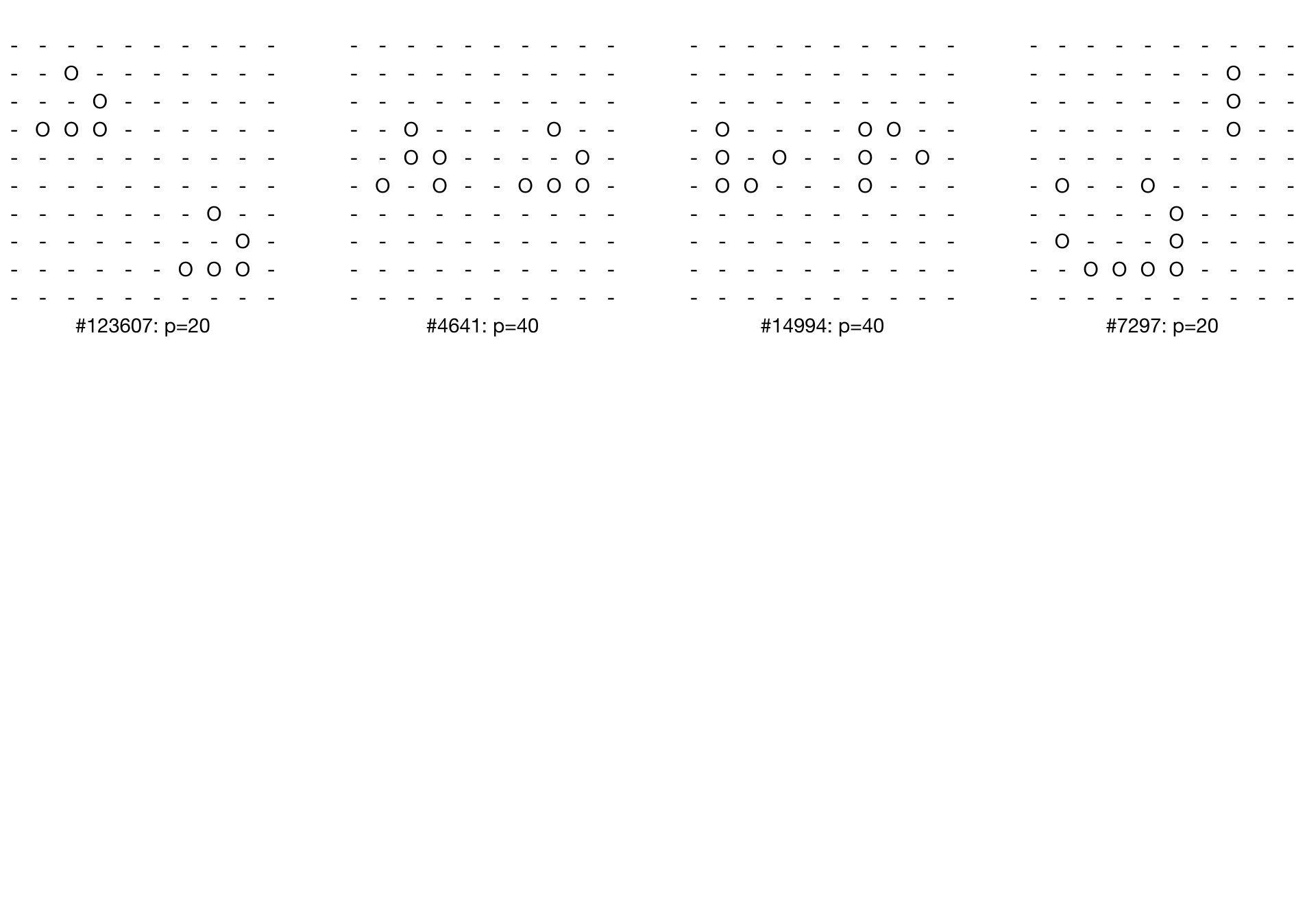}}
\end{picture}
\caption{\small{}Some pairs of Gliders, and a Light weight spaceship with Blinker, as ground patterns on the $10\times10$ torus.} 
 \label{fig:gliderpairs10}
\end{figure}

\begin{figure}
\begin{center}
\begin{picture}(440,585)
\put(5,-18){\includegraphics[width=450pt]{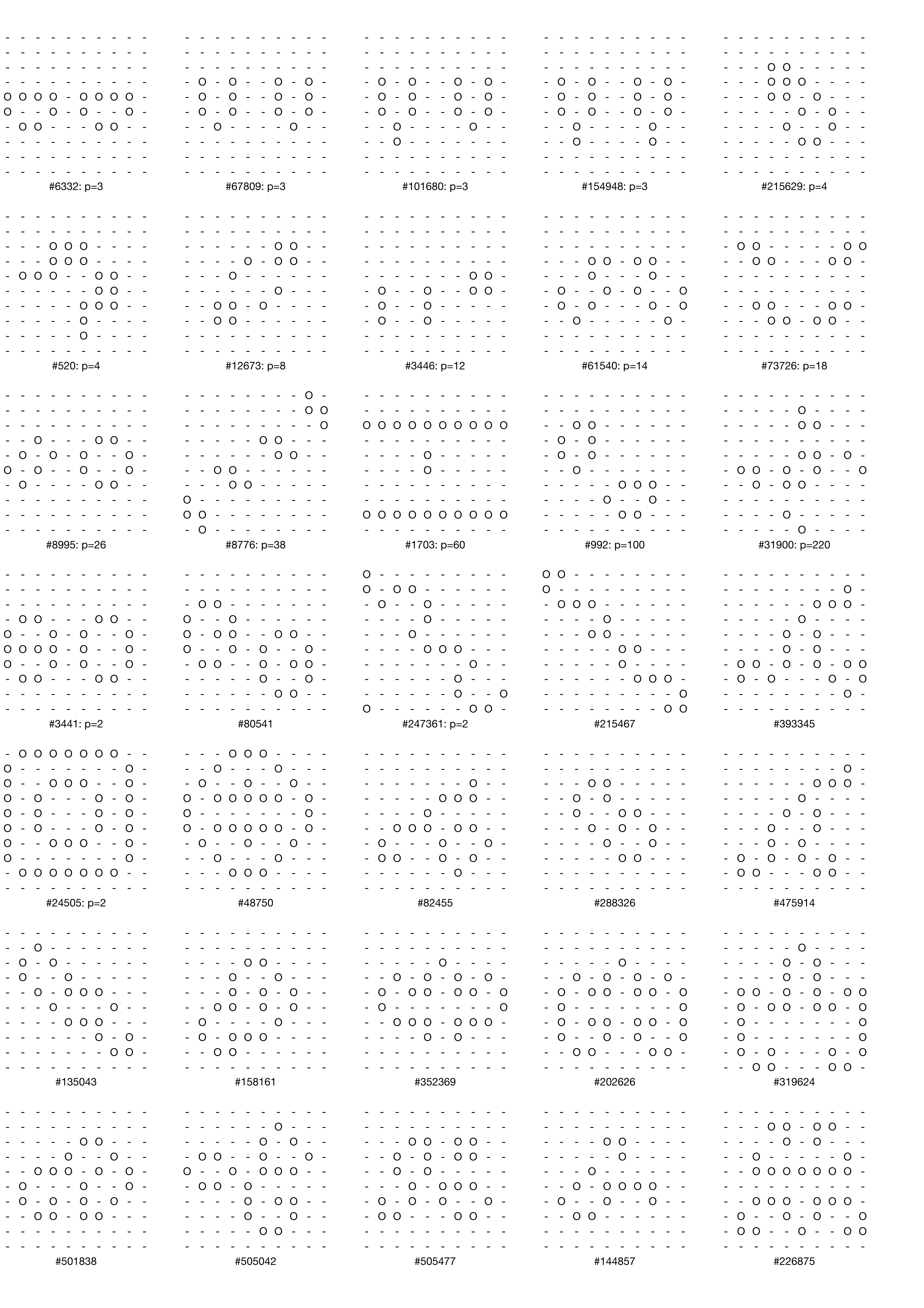}}
\end{picture}
\end{center}
\caption{Select examples of ground patterns on the $10\times 10$ torus.} 
\label{fg:10x10}
\end{figure}

All stable or cyclic patterns on the $5\times 5$ tours can be computed by a brute search.
There are 55 stable or cyclic patterns in total. 
Of them, 33 are still-lives (including the empty space). 
The patterns not present in Figure \ref{fg:lifet5} are listed in Appendix Figure \ref{fg:5x5e}. 
The two nearly square rectangles along the right edge enclose pairs of period 2 oscillators 
that may transition to each other; the corresponding cell perturbations are labeled by {\sf t}. 
Other patterns that are products of non-inert transitions are labeled alphabetically,
and the corresponding cell perturbations are labeled accordingly. 
The transitions to ground patterns or the empty space are not distinguished.
They constitute, respectively, 29.19\% and 60.19\% of all transitions from the 45 non-empty and non-ground patterns.
Inert transitions constitute 5.53\%, and the transitions between the non-ground patterns constitute just 5.10\%.
Possible circular paths of transitions are rare. 
There are only three sets of at least 2 non-ground patterns that might decay to each other indefinitely.
These three {\em strongly connected blocks} are:
\begin{itemize}
\item A set of 4 patterns labelled {\sf c, d, f, h}, enclosed in the largest rectangle at the top. 
The pattern {\sf c} is a period 3 oscillator. 
The transition matrix within this quartet is
\begin{equation}
\left( \begin{array}{cccc}
-70/3 & 2 & 3 & 0 \\
4/3 & -23 & 0 & 0 \\
1\; & 0 & -23 & 2 \\
0\; & 0 & 1 & -23 
\end{array} \right).
\end{equation}
The leading eigenvalue is approximately $-20.5007008138$. 
The only irreversible transitions to or from this quartet are {\sf r}$\,\to\,${\sf h} and {\sf f}$\,\to\,${\sf z}.
\item The two mentioned pairs of period 2 oscillators, enclosed in the rectangles along the right edge.
Their transition matrices are, respectively,
\begin{equation}
\left( \begin{array}{cc}
-25 & 1  \\ 1\; & -25
\end{array} \right), \qquad
\left( \begin{array}{cc}
\!-49/2\! & 1  \\ 1\; & -25 
\end{array} \right).
\end{equation}
Their leading eigenvalues are $-24$ and $(\sqrt{17}-99)/4 \approx-23.7192235936$, respectively. 
There are no external transitions to these pairs, therefore their patterns are not labeled.
\end{itemize}
Among the remaining $45-4-2-2=37$ non-ground patterns:
\begin{itemize}
\item The patterns {\sf x} and {\sf z} are by far the most common decay products before irreversible transition to
ground patterns or the empty space. Morevoer, they often decay inertly; hence 
they have relatively large 
$\lambda\in\{-17,-19\}$. Only the pattern \#10 in Figure \ref{fg:lifet5} has the slower self-decay rate $\lambda=-10$.
\item The patterns {\sf k}, {\sf q}, {\sf r} and the block {\sf \{c, d, f, h\}} can be intermediate products towards 
the transition to ground patterns or the empty space.
\item The patterns {\sf g},  {\sf w}, {\sf i} (Glider), {\sf m} (Toad) can be decay products as well. 
The decay to Toad is 
isolated in the schematic graph of all decays.
\item The patterns {\sf i, r, w, x, z} are decay products of, respectively, 2, 3, 2, 5, 6 non-ground patterns.
The patterns {\sf g, h, k, q, m} are decay products of one non-ground pattern outside their block.
\item The two blocks of paired period 2 oscillators, and 9 other patterns can decay to the mentioned products, 
but are not decay products themselves.
\item 19 patterns  (including \#10 from Figure \ref{fg:lifet5}) are isolated: 
they are not decay products, and decay only to ground patterns or the empty space.
\end{itemize}
The irreversible transitions between these patterns, the mentioned three blocks and the big block of ground patterns 
give a directed tree graph of the eventual decay to the empty space \#0. 
The 55 patterns could be ordered in such a way that irreversible transitions
would lower the ranking \#\,-number. 
Then the transition matrix 
of the Markov process on the 55 patterns would have an block-upper-triangular shape, 
with the non-trivial blocks of size $9,4,2,2$. 
The eigenvalues are associated to the blocks (those of size 1 as well), 
and they depend only on the transitions within the block. 
The eigenvectors would typically have non-zero entries at the patterns of the corresponding block
and of the downstream blocks, including the \#0 entry eventually.

respectively. A corresponding eigenvector may have non-zero entries not only those corresponding to the patterns in the block,
but also those corresponding to descendant blocks and patterns, downstream towards \#0.
For example, the pattern \#10 gives the largest eigenvalue $-10$ beside that of \#0. 
Its eigenvector is given in the last column of Table \ref{tb:5x5}.


Most non-ground patterns are dense in live cells.  32 have more than 10 live cells averaged according to ,
Out of 45 non-empty and non-ground patterns,
while 11 have less than 9 cells.  

\begin{remark} \rm
It is straightforward to enumerate stabilised patterns on the $4\times 4$ torus.
There are 17 stabilised patterns, including the empty space: six still lives, eight oscillators of period 2, 
an oscillator of period 4, and two oscillators of period 8. The non-empty patterns are displayed in Figure \ref{fg:4x4}.
The transitions between them stratify the patterns into a partial order. 
There are two strongly connected orbits formed by pairs of oscillators.

The stable patterns on the $3\times 3$ torus are those with 4 live cells,
while all other patterns vanish in one or two steps of ``Life".
This follows from the fact that all cells on this torus are neighbours to each other.
\end{remark}
 
\subsection{Islands}
\label{sec:islands}

\begin{figure}
\centering 
\begin{picture}(408,120)
\put(0,-80){\includegraphics[width=400pt]{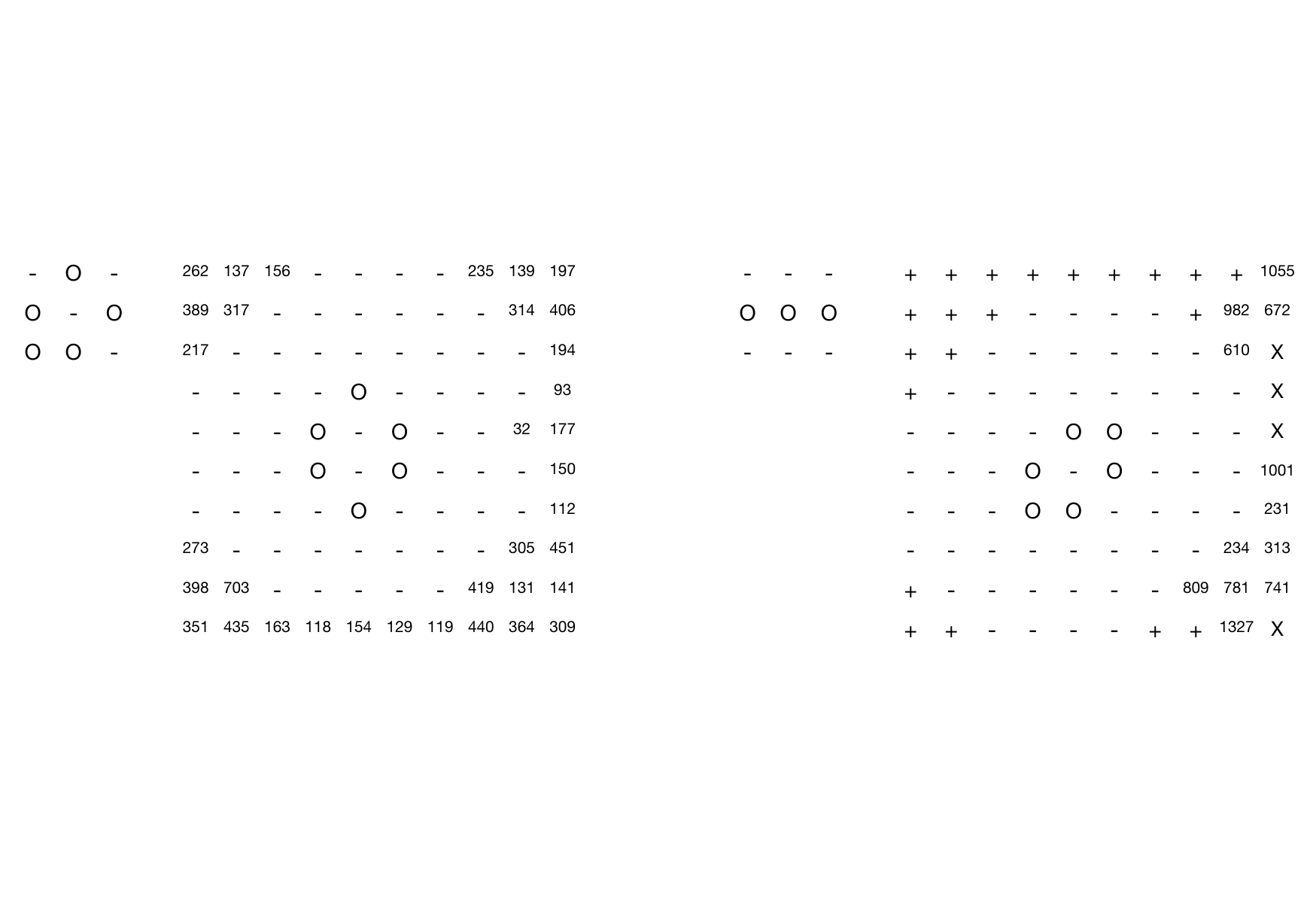}}
\put(6,10){\small{}\refpart{a}}  \put(224,10){\small{}\refpart{b}} 
\end{picture}
\caption{\small{}\refpart{a} Stable configurations of Beehive and Boat on the $10\times 10$ torus. 
The numbered cells indicate the \#-rank of a possible stable pair,
with Boat centered at that cell in the orientation displayed on the left. Boat's orientation can be changed by the symmetries of Beehive. 
\refpart{b} Possible pairs of Ship and Blinker on $T_{10}$. 
The numbered cells indicate the \#-ranks of the occurring pairs as ground patterns (with the Blinker centered at that cell), 
while the X-symbols indicate other possible pairs. The plus-signs indicate the possible pairs that are equivalent to
the marked pairs near the right edge due to the symmetries of Ship.} 
\label{fig:pairsEASN}
\vspace{25pt}
\begin{tabular}{@{}lrrrrr@{}}
\hline
Torus size ($n\times n$) & $6\times6$ & $7\times 7$ & $8\times8$ & $9\times9$ & $10\times10$
\\ \hline 
Ground patterns ($N$) & 30 & 42 & 305 & 7362 & 513875  \\
Islands (minimal set) & 25 & 19 & 50 & 111 & 633 \\
Patterns of unique single islands & 18 & 6 & 15 & 35 & 211 \\ 
All single island patterns & 21 & 14 & 28 & 79 & 444 \\
Patterns with 2 islands & 6 & 23 & 240 & 1529 & 17800 \\ 
Patterns with 3 islands & --- & --- & 6 & 2647 & 144778 \\ 
Patterns with 4 islands & --- & --- & 9 & 2949 & 311254 \\
Patterns with more islands & 1 & --- & --- & --- & 2296 \\
With varying number of islands& 2 & 5 & 22 & 158 & 37303 \\ \hline
Mean number of islands  & 1.433333 & 1.678571 & 2.003553 & 3.159954 & 3.662798  \\ 
Expected number of islands  & 1.021122 & 1.069007 & 1.047241 & 1.009124 & 1.113534 \\
Negentropy of unique single islands & 0.2839457 & 4.568118 & 2.899617 & 7.963361 & 7.395276 \\
Negentropy of single islands\;$(\times 10^{-3})$ & 9.242535 & 31.05835 & 21.35729 & 3.980562 & 37.36440 \\
Negentropy of 2 islands & 1.686294 & 1.161188 & 1.738143 & 2.039972 & 1.197233 \\
Negentropy of 3 islands & --- & --- & 4.790057 & 6.005577 & 2.303192  \\
Negentropy of 4 islands & --- & --- & 4.831765 & 11.09405 & 1.905974  \\
Negentropy of more islands & 4.788696 & --- & --- & --- & 3.878853 \\
Negentropy:\! varying number of islands & 3.348499 & 4.635856 & 1.527515 & 5.694741 & 2.849567 \\
\hline
Contain Boat & 1 & 5 & 94 & 4379 & 334642 \\ 
Contain Block & 3 & 11 & 58 & 3500 & 209323 \\ 
Contain Blinker & 2 & 11 & 38 & 2227 & 169836 \\ 
Contain Tub & 2 & 5 & 41 & 2590 & 164941 \\ 
Contain Ship & 1 & 3 & 42 & 1329 & 121366 \\ 
Contain Beehive & 1 & 7 & 35 & 1081 & 103181 \\ 
Contain Loaf & 1 & 3 & 43 & 469 & 74988 \\
Contain these 7 top islands only & 11 & 27 & 164 & 5284 & 168519 \\ 
Contain an induction coil & 2 & 1  & 18 & 154 & 27265 \\
Only with induction coils & 2 & 1 & 17 & 22 & 292 \\  
Negentropy of induction coils only & 3.157965  & 7.075707 & 4.504097 & 6.570032  & 3.156292 \\ 
Extra negentropy with induction coils\! & --- & --- & 15.64904 & 12.36673  & 6.394749 \\
\hline
\end{tabular}
\captionof{table}{Island composition of ground patterns.} 
\label{tb:islandcomp}
\end{figure} 

\begin{table} \small \vspace{-25pt} \def\arraystretch{0.9}
\begin{center}
\begin{tabular}{@{}lrrrrrrr}
\hline
Island or & & Pattern & \multicolumn{2}{l}{Negenetropy:\dotfill} & \multicolumn{3}{l}{Occurrences in sets:\dotfill} \\ 
standard pattern & \#-rank & count & singleton & \!additional & of 2 & of 3 & larger  \\ \hline
Block & 5 & 209323 & 1.29973 & 1.41033 & 2334 & 54490 & 202619 \\ 
Tub & 1 & 164941 & 0.51456 & 2.63582 & 1863 & 39522 & 153105 \\ 
Boat & 2 & 334642 & 0.66045 & 2.28742 & 5590 & 109597 & 363451 \\ 
Ship & 7 & 121366 & 1.53497 & 3.26172 & 2243 & 38647 & 97394 \\ 
Beehive & 3 & 103181 & 0.80745 & 1.68321 & 1875 & 28739 & 81131 \\ 
Loaf & 6 & 74988 & 1.34935 & 2.81118 & 2323 & 25396 & 50606 \\ 
Pond & 9 & 6634 & 1.60379 & 2.79099 & 405 & 2580 & 3767 \\ 
Blinker $(p=2)$ & 4 & 169836 & 1.05811 & 1.59349 & 1050 & 35355 & 178912 \\ 
Beacon $(p=2)$ & 35 & 10899 & 3.64575 & 2.93646 & 203 & 4304 & 6496 \\ 
Toad $(p=2)$ & 160 & 26444 & 4.88640 & 4.78730 & 473 & 7211 & 18805 \\ 
Clock $(p=2)$ & 13450 & 31122 & 11.0395 & 12.4993 & 215 & 4616 & 26291 \\ 
Bipole $(p=2)$ & 11712 & 1559 & 10.7938 & 14.5318 & 94 & 1464 & --- \\ 
Barge & 29 & 44558 & 3.43814 & 2.74734 & 502 & 10637 & 34226 \\ 
Long boat & 49 & 51452 & 3.89609 & 3.54401 & 772 & 16117 & 35746 \\ 
Long ship & 120 & 14590 & 4.61391 & 4.28111 & 350 & 6325 & 8092 \\ 
Long barge & 157 & 3531 & 4.87900 & 5.92371 & 198 & 2350 & 1012 \\ 
Very long boat & 316 & 5362 & 5.59015 & 7.07159 & 286 & 3693 & 1451 \\ 
Very long ship & 800 & 1924 & 6.68638 & 8.48708 & 141 & 1520 & 283 \\ 
Very long barge & 15397 & 684 & 11.2957 & 10.5797 & 68 & 615 & --- \\ 
Long-3 boat & 20909 & 965 & 11.9306 & 11.6019 & 77 & 887 & --- \\ 
Boat tie ship & 2086 & 1396 & 8.19730 & 14.0515 & 124 & 1272 & --- \\ 
Boat tie & 4429 & 983 & 9.33576 & 13.1495 & 101 & 882 & --- \\ 
Ship tie & 1139 & 552 & 7.27047 & 14.7830 & 64 & 488 & --- \\ 
Mango & 76 & 14493 & 4.26097 & 4.16513 & 542 & 5785 & 8283 \\ 
Eater & 966 & 44383 & 6.98555 & 6.55382 & 648 & 15459 & 28361 \\ 
Aircraft carrier & 1187 & 7865 & 7.33260 & 6.60526 & 256 & 7476 & 179 \\ 
Broken snake & 51705 & 3457 & 14.9997 & 14.4584 & 252 & 3204 & --- \\ 
Snake & 6504 & 8668 & 9.88657 & 9.03262 & 245 & 8491 & --- \\ 
Long snake & 40272 & 4179 & 13.9854 & 13.8799 & 189 & 3989 & --- \\ 
Very long snake & 31600 & 2703 & 13.0793 & 18.5325 & 137 & 2565 & --- \\ 
Snake with feather & 36741 & 1996 & 13.6124 & 12.6607 & 192 & 1803 & --- \\ 
Shillelagh & 867 & 5177 & 6.80790 & 10.0024 & 334 & 4844 & --- \\ 
Long shillelagh & 107472 & 741 & 18.0173 & 13.8703 & 248 & 493 & --- \\ 
Canoe & 1726 & 2702 & 7.90409 & 11.6063 & 171 & 2520 & 34 \\ 
Long canoe & 20753 & 959 & 11.9137 & 18.8053 & 71 & 887 & --- \\ 
Hat & 3333 & 1855 & 8.94437 & 11.0795 & 153 & 1672 & 30 \\ 
Integral sign & 9059 & 2663 & 10.3938 & 12.7382 & 139 & 2523 & --- \\ 
Hooked integral & 48603 & 1706 & 14.7452 & 14.0129 & 156 & 1549 & --- \\ 
Tub with tail & 11246 & 6588 & 10.7308 & 14.5810 & 331 & 6234 & 22 \\ 
Cis-boat with tail & 4538 & 3549 & 9.36610 & 12.8857 & 275 & 3267 & 6 \\ 
Trans-boat with tail & 12752 & 5042 & 10.9488 & 15.6350 & 274 & 4761 & 6 \\ 
Hook with tail & 13011 & 4533 & 10.9818 & 14.6119 & 302 & 4224 & 6 \\ 
Beehive with tail & 9236 & 2311 & 10.4216 & 17.0602 & 205 & 2105 & --- \\ 
Trans-barge with tail & 103528 & 1880 & 17.8628 & 26.1527 & 137 & 1742 & --- \\ 
Trans-long boat with tail & 78137 & 1361 & 16.7006 & 27.3155 & 97 & 1263 & --- \\ 
Trans-loaf with tail & 23361 & 1025 & 12.1955 & 15.1747 & 118 & 906 & --- \\ 
Cis-loaf with tail & 11702 & 773 & 10.7916 & 18.7719 & 165 & 607 & --- \\ 
Tub with nine & 29788 & 2072 & 12.8858 & 13.7348 & 149 & 1922 & --- \\ 
Trans-boat with nine & 17137 & 1424 & 11.5033 & 15.1623 & 100 & 1323 & --- \\ 
Cis-boat with nine & 32784 & 1198 & 13.2008 & 12.1767 & 125 & 1072 & --- \\ 
Beehive with nine & 27661 & 583 & 12.6625 & 13.2236 & 71 & 511 & --- \\ 
Cis-barge with nine & 24605 & 542 & 12.3375 & 11.0153 & 86 & 455 & --- \\ 
Cis-long boat with nine & 31686 & 490 & 13.0879 & 11.9562 & 83 & 406 & --- \\ 
Table & --- & 4075 & --- & 9.21664 & 45 & 1229 & 2843 \\ 
Cap & --- & 3409 & --- & 12.0811 & 43 & 1029 & 2348 \\ 
Bun & --- & 1032 & --- & 4.44720 & 49 & 775 & 702 \\ 
Bookend & --- & 579 & --- & 4.15207 & 48 & 537 & 149 \\ 
Long bookend & --- & 1673 & --- & 7.13112 & 53 & 1372 & 251 \\ 
Tub with long leg & --- & 475 & --- & 19.4614 & 20 & 455 & --- \\ 
Hook with long leg & --- & 453 & --- & 22.2503 & 18 & 435 & --- \\ 
\hline
\end{tabular}
\caption{Islands (or their common combinations such as Bipole, Aircraft carrier, Broken snake) 
on the $10\times10$ torus that occur in more than 450 ground patterns. Multiple occurrences in whole patterns
can be counted by summing up the last three columns and subtracting the third column (plus 1 if there is a singleton appearance).} 
\label{tb:islands10}
\end{center}
\end{table}

\begin{table} \small \vspace{-25pt} \def\arraystretch{0.92}
\begin{center}
\begin{tabular}{@{}lrrrrr@{}ll@{}}
\hline
Island configuration & \!\!\!Pattern\!\! & \multicolumn{2}{l}{\ First occurence\dotfill} 
& \multicolumn{1}{l}{\!\!Additional} & \multicolumn{3}{l@{}}{Occurrences in bunches:\ldots}    \\ 
& \multicolumn{1}{l}{\!\!count} & \ \#-rank &  \!\!negentropy\! & \!\!negentropy\!
& \multicolumn{2}{l}{Pairs, etc.} & \,Examples \\ \hline
2 Tubs, 2 Boats & 1727 & 3351 & 8.95035 & 8.59370 \\ 
Tub, 3 Boats & 2567 & 987 & 7.02178 & 8.48608 \\ 
Tub, 2 Boats, Beehive & 1994 & 5700 & 9.69002 & 9.11323 \\ 
Tub, 2 Boats, Blinker & 3268 & 41524 & 14.1128 & 13.1845 \\ 
Tub, 2 Boats, Block & 3097 & 2247 & 8.33617 & 7.96346 \\ 
Tub, 2 Boats, Ship & 1989 & 6226 & 9.81823 & 9.87389 \\ 
Tub, 2 Boats, Barge & 2096 & 62111 & 15.7824 & 14.8862 \\ 
Tub, 2 Boats, Long boat & 1760 & 48775 & 14.7623 & 14.1141 \\ 
Tub, Boat, Beehive, Blinker & 1720 & 7469 & 10.0994 & 9.63470 \\ 
Tub, Boat, Beehive, Block & 1731 & 5697 & 9.68962 & 8.82995 \\ 
Tub, Boat, 2 Blinkers & 2985 & 3098 & 8.85006 & 9.15503 \\ 
Tub, Boat, Blinker, Block & 2214 & 25984 & 12.4840 & 11.8367 \\ 
Tub, Boat, Blinker, Barge & 2080 & 216551 & 23.5042 & 23.2441 \\ 
Tub, Boat, Blinker, Toad & 1715 & 224361 & 24.3089 & 23.2339 & 37 & +0+4 & (\ref{eq:pshifts}) \\ 
Tub, Boat, Blinker, Clock & 2691 & 409901 & 43.1668 & 41.9599 & 28 \\ 
Tub, Boat, Block, Ship & 1953 & 12676 & 10.9375 & 9.97665 \\ 
Tub, Boat, Block, Long boat & 1731 & 66845 & 16.0784 & 15.2568 \\ 
Tub, Boat, Block, Eater & 1801 & 284762 & 29.8528 & 29.3583 \\ 
4 Boats & 1843 & 360 & 5.67516 & 6.97086 \\ 
3 Boats, Beehive & 1667 & 1893 & 8.05234 & 8.19898 \\ 
3 Boats, Blinker & 2575 & 33818 & 13.3131 & 12.5125 \\ 
3 Boats, Block & 3236 & 808 & 6.69954 & 7.59954 \\ 
3 Boats, Ship & 2114 & 2273 & 8.35627 & 9.63545 \\ 
3 Boats, Long boat & 2346 & 82555 & 16.9360 & 15.9594 \\ 
2 Boats, Beehive, Blinker & 1996 & 3665 & 9.07328 & 8.75138 \\ 
2 Boats, Beehive, Block & 2500 & 1923 & 8.08297 & 7.63867 \\ 
2 Boats, 2 Blinkers & 3441 & 1638 & 7.82931 & 8.00754 \\ 
2 Boats, Blinker, Block & 3100 & 18547 & 11.6648 & 10.9648 \\ 
2 Boats, Blinker, Ship & 1990 & 34918 & 13.4267 & 12.9836 \\ 
2 Boats, Blinker, Barge & 1994 & 214819 & 23.3354 & 22.7614 \\ 
2 Boats, Blinker, Long boat & 1656 & 216684 & 23.5185 & 22.8980 \\ 
2 Boats, Blinker, Toad & 1607 & 206312 & 22.5465 & 21.9209 & 21 & +0+2 & (\ref{eq:pshifts}) \\ 
2 Boats, Blinker, Clock & 2380 & 400074 & 41.4335 & 40.5877 & 16 \\ 
2 Boats, 2 Blocks & 2554 & 849 & 6.77109 & 6.87244 \\ 
2 Boats, Block, Loaf & 2098 & 6733 & 9.94339 & 9.29435 \\ 
2 Boats, Block, Ship & 3633 & 5600 & 9.66280 & 8.66969 \\ 
2 Boats, Block, Barge & 1777 & 99760 & 17.7106 & 16.5403 \\ 
2 Boats, Block, Eater & 2441 & 261560 & 27.9233 & 27.4371 \\ 
Boat, Beehive, 2 Blinkers & 1716 & 2210 & 8.31639 & 7.31309 \\ 
Boat, Beehive, Blinker, Block & 1731 & 3204 & 8.89530 & 8.32696 \\ 
Boat, Beehive, Block, Ship & 1648 & 7211 & 10.0424 & 9.06151 \\ 
Boat, 3 Blinkers & 1863 & 1097 & 7.22448 & 6.23718 & 9 & & \S\ref{sec:equal}\refpart{v} \\
Boat, 2 Blinkers, Block & 2223 & 1393 & 7.56922 & 7.37915 & 1 & & @263146 \\  
Boat, Blinker, Block, Ship & 1952 & 19870 & 11.8119 & 11.4593 \\ 
Boat, Blinker, Block, Long boat\! & 1647 & 197348 & 21.8498 & 21.3942 \\ 
Boat, Blinker, Block, Clock & 1664 & 389022 & 39.7843 & 38.9380 & 4 \\ 
Boat, 2 Blocks, Ship & 2161 & 2007 & 8.14380 & 7.91412 \\ 
Boat, 2 Blocks, Eater & 1824 & 240755 & 25.9833 & 27.7006 &  \\ 
Boat, Block, 2 Ships & 1622 & 2708 & 8.63553 & 9.87229 \\  \hline
Tub, 2 Blinkers, Barge & 587 & 1349 & 7.52394 & 8.55865 & 138 \\ %
Tub, 2 Blinkers, Long boat & 422 & 2542 & 8.56463 & 9.46137 & 107 \\ %
Tub, 2 Blinkers, Toad & 353 & 39556 & 13.9077 & 13.2107 & 98 & +0+9
& \S\ref{sec:equal}\refpart{v} \\ 
Tub, 2 Blinkers, Clock & 462 & 307226 & 31.6558 & 31.0602 & 118 & +0+8 & @462760+4 \\ 
Boat, 2 Blinkers, Barge & 1039 & 1407 & 7.58496 & 7.07893 & 237 & & \S\ref{sec:equal}\refpart{iii} \\
Boat, 2 Blinkers, Long boat & 801 & 2990 & 8.78079 & 8.04659 & 193 \\ %
Boat, 2 Blinkers, Toad & 597 & 21417 & 11.9843 & 11.2784 & 174 & +0+10\! \\ %
Boat, 2 Blinkers, Clock & 760 & 283166 & 29.7225 & 29.1321 & 200 & +0+8 & @458372+4 \\ 
\hline  \\ \vspace{-17pt}
\end{tabular}
\caption{Most common island configurations that feature in at least 1600 patterns, or in at least 100 bunches (see Section \ref{sec:equal}).
Examples refer to the text of Table \ref{tb:simplify}.} 
\label{tb:patterns10}
\end{center}
\end{table}

\begin{table} \small \vspace{-30pt} \def\arraystretch{0.92}
\begin{center}
\begin{tabular}{@{}lrrrrr@{}l@{\,}l@{}}
\hline
Island configuration & \!\!\!Pattern\!\! & \multicolumn{2}{l}{\ First occurrence\dotfill} 
& \multicolumn{1}{l}{\!\!Additional} & \multicolumn{3}{l@{}}{Occurrences in bunches\dotfill}   \\ 
& \multicolumn{1}{l}{\!\!count} & \ \#-rank &  \!\!negentropy & \!\!negentropy\!
& \multicolumn{2}{l}{Pairs, etc.} & \,Examples  \\ \hline
Boat, Shillelagh & 92 & 21340 & 11.97674 & 11.46404 \\
Boat, Eater & 108 & 2335 & 8.40384 & 8.65761 \\
Boat, Tub with tail & 84 & 72087 & 16.36864 & 16.36044 \\
Boat, Long shillelagh & 80 & 139038 & 19.19244 & 18.84212 \\
Blinker, Integral sign & 9 & 29092 & 12.81198 & 15.24140 & & & \#58087,\#140892 \\ 
Blinker, Bipole & 12 & 50896 & 14.93759 & 14.93360 & 2 & & \#50896,\,\#96995 \\ 
Loaf, Eater & 63 & 7897 & 10.18382 & 10.22271 \\  
\hline
Tub, Boat, Eater & 823 & 153897 & 19.73185 & 18.49674 \\ %
2 Boats, Eater & 1270 & 107938 & 18.03477 & 17.10890 \\ %
2 Boats, Aircraft carrier & 810 & 115784 & 18.34295 & 17.77375 \\ %
2 Boats, Snake & 869 & 85815 & 17.09148 & 16.60743 \\ %
Boat, Blinker, Beacon & 272 & 2184 & 8.29482 & 8.30047 & 5 & & \#240215+1 \\ 
Boat, Blinker, Toad & 1014 & 12098 & 10.85573 & 11.62049 \\ %
Boat, Blinker, Bipole & 102 & 105246 & 17.93090 & 17.45553 & 51 &\ (all) \\ 
Boat, Block, Eater & 1261 & 20313 & 11.85873 & 11.38047 \\ %
Boat, Block, Aircraft carrier & 819 & 43650 & 14.31681 & 16.14116 \\ %
Boat, Block, Snake & 886 & 57991 & 15.49070 & 15.01223 \\ %
Boat, Ship, Eater & 967 & 138808 & 19.18476 & 18.27865 \\ %
Beehive, Blinker, Bipole & 30 & 185136 & 21.10656 & 21.01280 & 15 &\ (all) & \S\ref{sec:acriterion}\refpart{v} \\ 
2 Blinkers, Barge & 193 & 18 & 2.95702 & 5.74817 & 8 & & @413133 \\ 
2 Blinkers, Long boat & 213 & 47 & 3.88251 & 6.54907 & 5 & & @418249 \\ 
2 Blinkers, Toad & 333 & 6690 & 9.93243 & 9.91230 & 20 & +1+0 & \#253539+11 \\  
2 Blinkers, Eater & 162 & 69263 & 16.21177 & 16.81380 & 40 & & @402916 \\ 
2 Blinkers, Integral sign & 9 & 215403 & 23.39006 & 25.04183 & 1 & &  \S\ref{sec:equal}\refpart{iv} \\ 
2 Blinkers, Bipole & 10 & 229500 & 24.83904 & 24.83837 & 4 & & \S\ref{sec:equal}\refpart{iv}, \S\ref{sec:acriterion}\refpart{v} \\  
2 Blinkers, Clock & 213 & 31423 & 13.05754 & 27.40143 & 12 & & \#405498+1 \\ 
Blinker, Block, Beacon & 130 & 673 & 6.39886 & 11.80852 & 7 & & @508221 \\  
Blinker, Barge, Clock & 12 & 275286 & 29.08028 & 28.48580 & 4 & & \S\ref{sec:acriterion}\refpart{ii},\;\refpart{iii}  \\ 
Blinker, Long boat, Clock & 24 & 289181 & 30.20837 & 29.26858 & 8 & & \S\ref{sec:acriterion}\refpart{ii},\;\refpart{iii} \\ 
Ship, Toad, Bipole & 6 & 509618 & 95.24631 & 94.97854 & 3 & \ (all) & \#509741 \\ 
\hline
2 Tubs, Blinker, Toad & 537 & 240481 & 25.95362 & 25.21220 & 15 & +0+2 & (\ref{eq:pshifts}) \\ 
Tub, 3 Blinkers & 992 & 975 & 7.00567 & 7.89942 & 6 & & \S\ref{sec:equal}\refpart{v} \\ 
Tub, Blinker, Block, Toad & 621 & 209374 & 22.82167 & 22.18411 & 17 & +0+1 & (\ref{eq:pshifts}) \\ 
2 Boats, Blinker, Beacon & 118 & 243013 & 26.21018 & 25.67702 & 39 & & Fig.~\ref{fg:obstructed} \\ 
Boat, Blinker, Block, Beacon & 206 & 226809 & 24.56535 & 23.98731 & 60 & & Fig.~\ref{fg:obstructed} \\ 
Boat, Blinker, Block, Toad & 1480 & 181155 & 20.8946 & 20.2548 & 20 & +0+1 & (\ref{eq:pshifts}) \\  
Boat, Blinker, Loaf, Beacon & 32 & 262148 & 27.97304 & 27.44801 & 14 & & \S\ref{sec:acriterion}\refpart{i} \\ 
Boat, Blinker, Ship, Beacon & 93 & 255563 & 27.39444 & 26.88786 & 31 & & \S\ref{sec:acriterion}\refpart{i} \\ 
Boat, Beehive, Blinker, Beacon\!\! & 68 & 242717 & 26.18596 & 25.65922 & 24 & & Fig.~\ref{fg:obstructed} \\ 
Beehive, 2 Blinkers, Barge & 267 & 965 & 6.98427 & 6.43965 & 52 & & \S\ref{sec:acriterion}\refpart{vi} \\ 
Beehive, 2 Blinkers, Toad & 162 & 16667 & 11.44896 & 13.15852 & 43 & &  \S\ref{sec:acriterion}\refpart{iv} \\ 
Beehive,\;Blinker,\,Block,\,Beacon\!\! & 62 & 226645 & 24.54863 & 24.37514 & 18 & & Fig.~\ref{fg:obstructed} \\ 
4 Blinkers & 222 & 10 & 1.91978 & 4.69409 & 34 & +1+1 & \S\ref{sec:acriterion}\refpart{iv} \\ 
3 Blinkers, Block & 546 & 852 & 6.77864 & 6.76386 & 2 & & \S\ref{sec:equal}\refpart{v} \\ 
3 Blinkers, Loaf & 510 & 838 & 6.74883 & 6.30608 & 2 & & \S\ref{sec:equal}\refpart{v} \\ 
3 Blinkers, Ship & 350 & 1645 & 7.83533 & 7.39349 & 1 & & \S\ref{sec:equal}\refpart{v} \\ 
3 Blinkers, Barge & 82 & 344 & 5.65892 & 5.31395 & 20 & & \S\ref{sec:acriterion}\refpart{vi} \\ 
3 Blinkers, Toad & 41 & 6217 & 9.81629 & 9.34152 & 6 & +0+1 & \S\ref{sec:acriterion}\refpart{iv} \\ 
3 Blinkers, Eater & 22 & 407302 & 42.67265 & 42.67265 & 10 & & \S\ref{sec:acriterion}\refpart{iv} \\ 
2 Blinkers, 2 Blocks & 505 & 486 & 5.97716 & 6.92568 & 1 & & @243498 \\ 
2 Blinkers, Block, Clock & 226 & 299712 & 31.05109 & 30.87538 & 65 & & @478623+2 \\
2 Blinkers, Long boat, Clock & 12 & 279528 & 29.42668 & 28.84335 & 6 & \ (all) & \S\ref{sec:acriterion}\refpart{ii},\refpart{iii} \\ 
\hline
Tub, Boat, Beehive, 2 Blocks & 82 & 185115 & 21.10534 & 21.40861 \\ 
Tub, Boat, Blinker, 2 Blocks & 76 & 154984 & 19.77278 & 19.77272 \\ 
2 Boats, Beehive, 2 Blocks & 124 & 192015 & 21.50258 & 27.52184 \\ 
Boat, Beehive, 3 Blocks & 124 & 150257 & 19.59906 & 26.22646 & 3 \\ 
Boat, 2 Blinkers, 2 Blocks & 44 & 256466 & 27.47375 & 26.70397 & 22 & \ (all)\\ 
Beehive, 4 Blocks & 50 & 1348 & 7.52294 & 18.46184 \\ 
\hline
\end{tabular}
\caption{Most common (overall or in bunches) configurations of 2, 3, 4 or 5 islands. } 
\label{tb:patterns10a}
\end{center}
\end{table}

Most patterns on larger toruses are configurations of several familiar islands delimited by the neighbouring relation of live cells.
For example, all 35 ways to place a Boat and a Beehive on $T_{10}$ 
resulting in a still live configuration 
occur as ground patterns on $T_{10}$; see Figure \ref{fig:pairsEASN}\refpart{a}.
In contrast, there are 16 different ways to place Ship and Blinker on $T_{10}$,
but only 12 of them occur as ground patterns; see Figure \ref{fig:pairsEASN}\refpart{b}.

Most of the islands that do not occur as individual patterns are {\em induction coils} 
\cite[{\sf Induction\_coil}]{Wiki} that typically need to be paired for stability.
Frequent induction coils are named House, Bookend, Bun, Cap, Table, etc.
Some large islands stabilize themselves by wrapping around the torus homotopically. 

The statistics about island composition of ground patterns is presented in the upper part of Table \ref{tb:islandcomp}.
The middle part of this table recapitulates mostly the same statistics in terms of mean or $v^{(1)}$-weighted expected averages 
and Boltzmanian $\log_{10}$-negentropies of probabilities. 
Negentropy of single islands is dominated by the top patterns of Table \ref{tb:patterns}.
The lower part counts the patterns that include the familiar top patterns as islands, 
or patterns with induction coils. Pond (though always a top 10 pattern) is
omitted from the counting statistics because it becomes less frequent than 
Barge, Mango, Eater, Long boat or Long ship as an island on larger toruses. 
In counting patterns with induction coils in the last rows:
\begin{itemize}
\item We ignore oscillators of period $p>2$, and we count only coil sparks and Bipole, Quadpole 
as inducing in period 2 oscillators. 
\item We consider any island (on the doubly periodic spread of the torus) 
that would not be stable if isolated as an induction coil, including Preblock
and large islands that stabilise themselves across the torus homotopy.
\end{itemize}

Categorising and counting islands is problematic in oscillating patterns, because the number and shape of islands may differ in different phases.
The second row of Table \ref{tb:islandcomp} counts a minimal set of islands sufficient to represent a phase of any oscillator.
Figures \ref{fg:8x8}, \ref{fg:9x9}, \ref{fg:10x10} exhibit oscillators in those phases for a minimal set, 
which is not hard to determine because the number of novel oscillators is not large.
Here are representative examples and issues:
\begin{itemize}
\item Spark coils \cite[{\sf Spark\_coil}]{Wiki} are period 2 oscillators that consist of two induction coils in one frame, 
but those two islands become connected in the other frame. The instances on $T_{10}$ are \#167227, \#179349.
We have the spark coils \#214, \#369, \#530, \#943 on $T_9$.
Besides, \#2206, \#2706, \#3044 on $T_9$ are combinations of spark coils with Block.
\item The oscillators \#117 on $T_8$ and \#6332, \#61540, \#73726 on $T_{10}$ are symmetric, and their symmetric halves are always disconnected.
Besides, \#209179 and \#183918 on $T_{10}$ are combinations of \#6332 with one or two Blocks, respectively.
\item The oscillators \#114 on $T_8$, \#267 on $T_9$, and \#1024, \#14342 on $T_{10}$ have phases consisting of two isolated full rows 
of live cells. A connected phase may or may not appear.
\item Moving pyramids --- such as on the second rows of Figures \ref{fg:8x8}, \ref{fg:9x9} ---
may have a common core (typically oscillating with period 2 while moving) but different sparks. 
The pyramids \#148424, \#194114 on $T_{10}$ share a common smaller core every second phase, but different larger cores in the other phases.
\item Beacon and Toad are period 2 oscillators with one or two connected components in the alternating phases.
Nearly all (except 26) 
oscillators on $T_{10}$ with varying number of islands contain Beacon or Toad, including \#150553, \#181938, \#229178 that contain three or four Beacons.
On the other hand, there are 150 patterns on $T_{10}$ containing exactly two Beacons and/or Toads.
If 64 of those cases, the Beacons and/or Toads are out of phase, and those patterns have a constant number of islands through the phases.
On $T_9$ we have two Beacons at \#890 (in phase) and \#1906 (our of phase).
\item Bipole and Quadpole are period 2 oscillators that always have 2 or 4 islands, respectively.  
Moreover, \#89870 consists of 10 isolated live cells in each frame and forms a boundless Barberpole of period 2, like \#27 on $T_{6}$. 
(Apart from these two Barberpole examples, the maximal number of islands is 6 in the patterns \#23491 and \#159860, 
which are $2\times3$ grids of Blocks.)
\end{itemize}
Table \ref{tb:islands10} lists the most frequent islands on $T_{10}$, with these provisions:
\begin{itemize}
\item Bipole and Aircraft carrier 
(and tied or siamese forms of Carrier, including Broken snake) are recognised as predominant ways of pairing Hook and Preblock islands
into stabilised configurations. Preblock appears in 12096 patterns (ignoring Beacons and oscillators of period $p>3$), 
but only in 21 of them it do not fit the Carrier motif.
Preblock is then stabilised mostly in Alef or with long forms of Table or Worm; see the last two examples in Figure \ref{fg:10x10}.
\item Incidental occurrences of most frequent patterns in oscillators (like in the first three rows of Figure \ref{fg:10x10}) 
are not counted to their statistics. On the other hand, occurrences of House (and other such induction coils) in spark coils, 
or occurrences of rarer islands such as sparks or homotopic rows of live cells, could be counted in specifically defined phases.
\end{itemize}
As occurrence negentropy of an island is usually dominated by its singleton pattern, 
Table \ref{tb:islands10} shows the negentropy of occurrences in larger sets separately in the fifth column.
Tables \ref{tb:patterns10}, \ref{tb:patterns10a} give most frequent island compositions of patterns. 
Examples in the last column are taken from Sections \ref{sec:equal}, \ref{sec:acriterion}, \ref{sec:aequal} or Table \ref{tb:simplify}.

\subsection{Negentropic transitions}
\label{sec:regimes}

Figure \ref{fig:matrices89}\refpart{b} suggests that the lower-triangular part of the transition matrix 
(representing negentropic transitions) is narrowly close to the main diagonal. 
This negentropic part could be enhanced visually by projecting the deviation from the main diagonal to a separate axis.
Figure \ref{fig:entropy2}\refpart{a} basically projects the diagonal of $M_9$ to a horizontal axis (and inverts the deviation to the positive vertical direction),
and skips some negentropic entries as it shows only the most negentropic transitions {\em from} each pattern.
The complementary Figure \ref{fig:entropy2}\refpart{b} projects the diagonal of $M_8$ to a vertical axis (and then mirrors the axes)
and shows only the most negentropic transitions {\em to} each pattern. The generating bounty from the high order oscillators \#9 and \#10 is 
then clearly visible. Most schematically, Figure \ref{fig:entropy2}\refpart{c} projects the diagonal of $M_{10}$ first to a vertical axis as well,
but the deviation counts the number of patterns from behind each rank $\#k$ (including $\#k$ possibly) with negentropic transitions to beyond $\#k$.
In terms of the sequence of linear eigenvector equations, this counts the undeterminance of equations with gradually expanding set of variables.

Figures \ref{fig:entropy2}\refpart{a},\refpart{c} allow to estimate and analyse the regimes and bottlenecks of negentropic pattern generation.
Their valleys between peaks correspond well to larger eigenvalue jumps in Figure \ref{fig:entropy}\refpart{c}.
The middle valley in Figure \ref{fig:entropy2}\refpart{a} is noticeably reflected in the separation (by the median horizontal value $\approx 899.5$)
of two denser clouds in Figure \ref{fig:entropy}\refpart{d} as well, and the horizontal values (\ref{eq:dbentropy}) correlate with the \#-ranking strongly: $r=0.9568$.
Tables \ref{tb:deficit9} and \ref{tb:deficit10} identifies the regimes from Figures \ref{fig:entropy2}\refpart{a},\refpart{c} counts patterns in terms of the number of islands.
It appears that the valleys separate predominantly the patterns with different number of islands. For example, the mentioned middle valley for $T_9$  
separates sharply the patterns with 2 or 3 islands. 
The density of oscillators with variable number of islands may vary sharply in the distinguished intervals as well.
For example, there is just one such oscillator in the interval \#3068+1642 
on $T_9$. 

\subsection{Almost equal components of eigenvectors}
\label{sec:aequal}

Here we extend Section \ref{sec:equal} by noting 
that pairs of very similar eigenvector equations may lead to near equalities of some eigenvector components. 
For example, the patterns \#3169 and \#3401 on $T_{9}$ 
nearly augment the pairs $@3170$ and $@3402$ into triples.
Here are the similar eigenvector equations:
\begin{align*}
(\lambda+66)v_{3169}-4v_{3401} & =v_{2912}+5v_{7050},\\
(\lambda+66)v_{j+3170}-4v_{j+3402} & =v_{2912}+5v_{j+7199} \qquad (j\in\{0,1\}), \\
-v_{3169}+(\lambda+65)v_{3401} & =5v_{7143},\\
-v_{j+3170}+(\lambda+65)v_{j+3402} & =5v_{j+7242} \qquad (j\in\{0,1\}).
\end{align*}
The relative difference $\Delta v_k/v_k$ for \mbox{$k\!\in\!\{3169,3401\}$} is, 
respectively, 
\mbox{$\approx\!1.78\!\cdot\!10^{-31}$} or \mbox{$4.59\!\cdot\!10^{-31\!}$.} 

The minimal eigenvalue jump on $T_{10}$ 
is between \#96995 and \#96996; 
the relative difference is $\approx 5.26\cdot10^{-83}$. 
Both eigenvector equations for these components have 60 terms:
\begin{itemize}
\item A term with \#96995 or \#96996 with the same decay rate $-\frac{143}2$.
\item 3 pairwise identical terms with \#58087, \#140892, 
or \#511226. 
\item 55 pairwise equivalent terms. Except for the isolated pair @387974, 
the involved patterns are from the cluster of size 168.
\item A term with \#509009 or \#509010 
with the decay rate $\frac12$.  
The relative difference between these two eigenvector components is $\approx 6.16\cdot10^{-5}$.
Their eigenvector equations have the same self-decay $-81$, 
share 21 equivalent terms from the same big cluster 
and 2 identical terms (with \#508338 or \#509638), 
but  \#509009 has an extra term with \#509741. 
\end{itemize}
Interestingly, there are a few long sequences 
of consecutive pairs of patterns with the same small relative jump 
in the eigenvector components. For example, consider the 81 patterns 
whose \#-ranking is given by the partial or complete sums in these expressions
similar to (\ref{eq:deltan}): { 
\begin{align} \label{eq:pshifts}
& \#408471+8009+1453+3+2+2+5786+1360+2+4+1058+22+41+28, \qquad \nonumber\\
& \#430753 
+661+35+547+27+2+2+51+38+1470+69+38+2380+226,  \nonumber\\
& \#436357 
+885+62+270+63+28+485+362+195+484+658+98+49,  \nonumber\\
& \#440110 
+600+362+46+984+2+448+70+32+466+281+59+15,  \nonumber\\
& \#444017 
+388+985+98+2+448+2+692+676+90+2+583+1095+2,  \nonumber\\
& \#449748 
+21+452+2+264+2+509+1398+645+2+523+27+1340.
\end{align}
}All 
relative differences to the (correspondingly) next eigenvector components are the same: \mbox{$\approx2.03\cdot10^{-11}$.} 
The relative difference is preserved from an initial discrepancy by parallel generating decays, apparently.
Other such sequences have 30, 24, 15 or fewer patterns. They may involve pattern bunches 
from the described clusters. For example, all pairs in the two clusters of size 30 starting from 
@478623 or @478625 follow each other immediately
in the \#-ranking with the same relative difference $\approx 1.58\cdot10^{-12}$. 
Those pairs nearly merge into quartets. 
Similarly, some quartets nearly merge into octets:
\begin{itemize}
\item Two clusters of size 15 are fully parallel with the constant relative difference $\approx 1.81\cdot10^{-8}$.
They both have two quartets that align closely: @458372+4 and @462760+4, following the notational convention in (\ref{eq:deltan}).
\item The cluster of size 49 is fully parallel to a part of the cluster of size 75, 
with the constant relative difference $\approx 2.87\cdot10^{-7}$.
This gives 
the close quartets @464367+4 and @469475+4.
\item The remaining part of the cluster of size 75 splits into two closely parallel tracks 
that deviate from each other by the constant relative difference $\approx 5.78\cdot10^{-11}$, 
and include the close quartets @501130+4 and @503001+4. 
\end{itemize}

\begin{figure}
\begin{center}
\begin{picture}(380,570)
\put(8,420){\includegraphics[width=360pt]{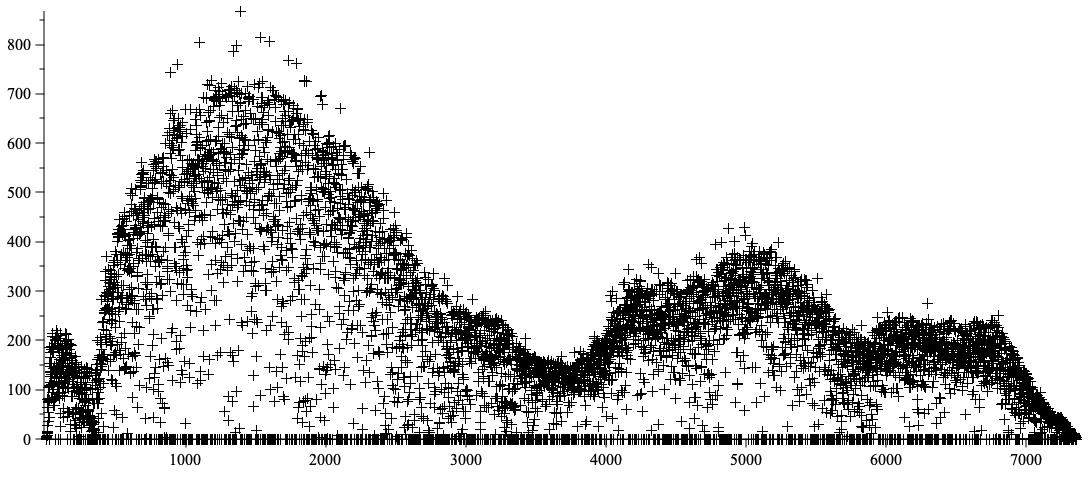}}
\put(8,212){\includegraphics[width=360pt]{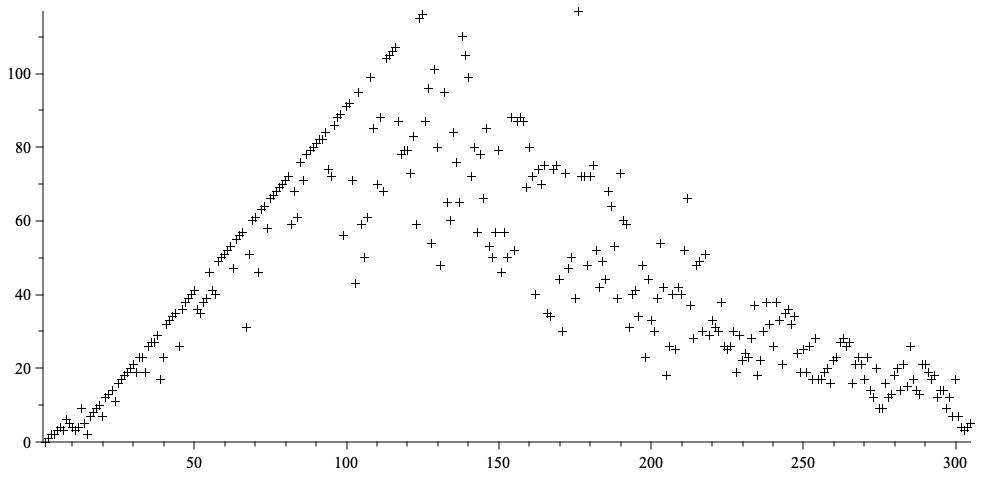}}
\put(8,4){\includegraphics[width=360pt]{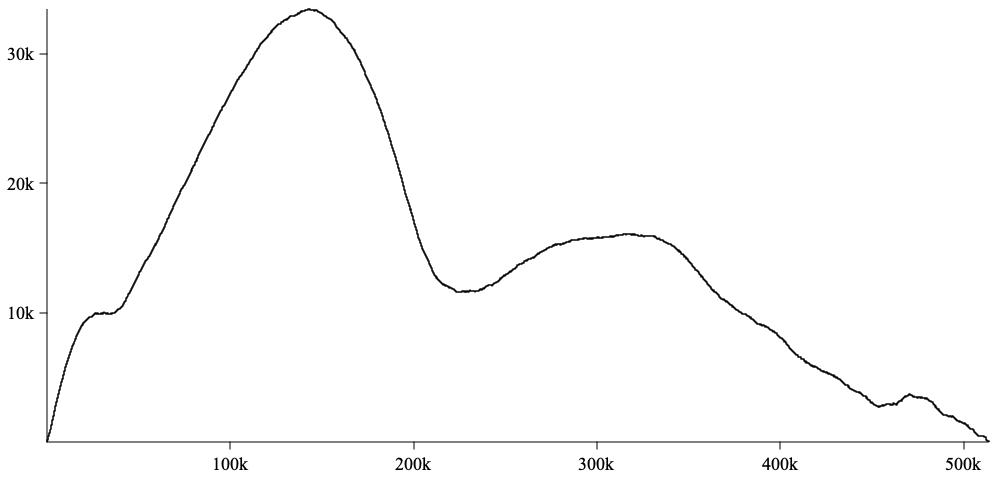}}
\put(0,418){\small{}\refpart{a}}  \put(0,211){\small{}\refpart{b}}  \put(0,6){\small{}\refpart{c}}
\end{picture}
\caption{
\refpart{a} Maximal decay drop in the \#-ranking among the transitions from each state $\#k$, for the $9\times 9$ torus.
\refpart{b} Maximal predecessor drop in the \#-ranking among the transitions to each state $\#k$, for the $8\times 8$ torus.
\refpart{c} Running deficit of the number of eigenvector equations with the maximal involved \#-rank of appearing pattern variables, 
for the $10\times 10$ torus.} 
\label{fig:entropy2}
\end{center}
\end{figure}
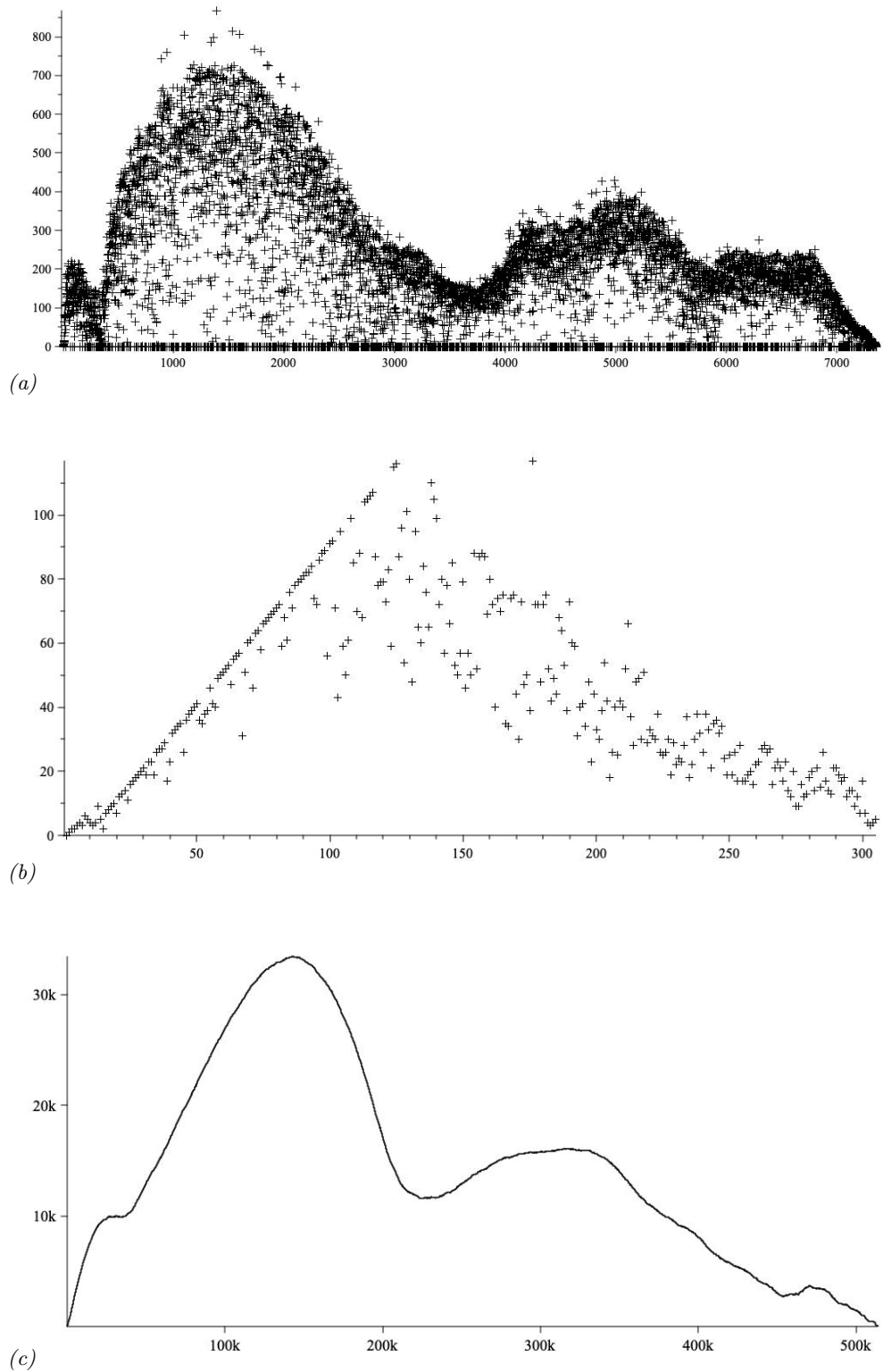

\begin{table} \small 
\begin{center}
\begin{tabular}{@{}rrr@{\qquad}rrrrr@{\;}c@{\;}l@{}}
\hline
Starting & Interval & 
\multicolumn{5}{l}{The number of islands\dotfill} & \multicolumn{3}{r@{}}{Equation deficit} \\
\#-rank & length & 1 & 2 & 3 & 4 & \!\!varies & \multicolumn{3}{r@{}}{change} \\  \hline
1 & 382 & \quad 19 & 342 & 16 & 0 & 5 & $0\nearrow$ & 129 & $\searrow 42$  \\ 
383 & 1478 & 10 & 377 & 1013 & 7 & 71 & $42\nearrow$ & 439 \\ 
1861 & 1198 & 19 & 134 & 939 & 33 & 73 & & 439 & $\searrow 134$ \\ 
3059 & 668 & 10 & 139 & 416 & 102 & 1 & $134\nearrow$ & 160 & $\searrow 82$ \\ 
3727 & 1016 & 8 & 122 & 86 & 798 & 2 & $82\nearrow$ & 215 & $\searrow 185$ \\ 
4743 & 1140 & 2 & 58 & 82 & 995 & 3 & $185\nearrow$ & 268 & $\searrow 121$ \\  
5882 & 1480 & 11 & 357 & 95 & 1014 & 3 & $121\nearrow$ & 163 & $\searrow 0$  
\\ \hline
\end{tabular}
\caption{Variation of the deficit of linear equations (vs the maximal involved \#-rank of appearing pattern variables) 
for the different regime intervals, for the $9\times9$ torus.} 
\label{tb:deficit9}
\end{center}
\end{table}

\begin{table} \small 
\begin{center}
\begin{tabular}{@{}rrr@{\qquad}rrrrrr@{\;}c@{\;}l@{}}
\hline
\multicolumn{1}{@{}l}{Starting} & Interval & 
\multicolumn{6}{l}{The number of islands\dotfill} & \multicolumn{3}{l@{}}{Equation deficit change} \\
\#-rank & length & 1 & 2 & 3 & 4 & more & \!\!varies & \\  \hline
1 & 25800 & \quad 62 & 1745 & 18650 & 4816 & 31 & 496 & $4\nearrow$ & 9856 &  \\ 
25801& 11400 & 18 & 300 & 5417 & 5302 & 3 & 360 & $9858\,\;-$ & 10027 \\ 
37201 & 105600 & 65 & 2000 & 13826 & 84665 & 10 & 5034 & $10019\nearrow$ & 33475 \\ 
142801 & 82158 & 119 & 2941 & 14763 & 60525 & 196 & 3614 & & 33475 & $\searrow 11570$ \\ 
224959 & 92420 & 68 & 4184 & 34956 & 46763 & 712 & 5737 & $11571\nearrow$ & 16095 \\ 
317379 & 136286 & 74 & 4635 & 48002 & 64549 & 851 & 18175 & & 16096 & $\searrow 2718$ \\  
453665 & 60211 & 38 & 1995 & 9164 & 44634 & 494 & 3886 & $2718\nearrow$ & 3714 & $\searrow 0$  
\\ \hline
\end{tabular}
\caption{Variation of the deficit of linear equations (vs the maximal involved \#-rank of appearing pattern variables) 
for the different regime intervals, for the $10\times10$ torus.} 
\label{tb:deficit10}
\end{center}
\end{table}

\begin{table} \small
\begin{center}
\begin{tabular}{@{}lrl@{}}
\hline
The cluster content & \!\!Size & Participating islands
\\ \hline
@3101 & 1& three Blinkers \\
@5303 & 1 & Bipole + Blinker \\
@6370 & 1 & three Blocks + Beehive \\
@6443 & 1 & two Boats + two Blinkers \\
@5132, @5457 & 2 & Block/Boat + Tub + Beehive + Blinker \\
@5228, @5531 & 2 & Block/Boat + Tub + Beehive + Blinker \\
@5778, @5955 & 2 & Block/Boat + Beehive + 2 Blinkers \\
@6464, @6655 & 2 & Block/Boat + Beehive + 2 Blinkers \\
@6801, @7028 & 2 & Tub/Boat + 3 Blinkers \\ 
@3170, @3402, @7199, @7242 & 4 & Tub/Boat + 2/3 Blinkers \\ 
@4753, @4968, @5065, @5280 & 4 & Tub/Boat + Tub/Boat + 2 Blinkers \\
@5776, @5910, @5991, @6086 & 4 & Tub/Boat + Tub/Boat + 2 Blinkers \\
@6462, @6619, @6702, @6815 & 4 & Tub/Boat + Tub/Boat + 2 Blinkers \\
@7104, @7158, @7185 & 3 & Tub/Boat + Tub/Boat + 2 Blinkers \\
@6799, @7026, @7050, @7141, @7143, @7195 & 6 &  Tub/Boat + Beehive/Blinker + 2 Blinkers \\
@6689, @6832, @6884, @7007, @7023, & 10& \{Block/Boat\;+\;Block/Boat\}/\{Boat\;+\;Tub\} \\
\quad  @7033, @7112, @7119, @7127, @7181 && /\{Ship\;+\;Block/Beehive\}\;+\;Boat\;+\;Blinker 
\\ \hline
\end{tabular}
\end{center}
\caption{Parallel clusters of patterns with equal eigenvalue components, on the $9\times9$ torus. 
The @-numbers in the the first column indicate the pairs $Eq(k):v^{(1)}_{k}=v^{(1)}_{k+1}$ of the cluster. }
\label{tb:parallel}
\end{table}

\begin{table} \small\vspace{-50pt} 
\begin{tabular}{@{}rrrrl@{}}
\hline
Size 
& Triples & Quadr. & Cases & Examples; the @-summation notation is explained with (\ref{eq:deltan})
\\ \hline 
1 
& --- & --- & 119 & 
@3440;\hspace{0.5pt}@26618;\hspace{0.5pt}@48477;\hspace{0.5pt}@63908;\hspace{0.5pt}@90087;\hspace{0.5pt}@103059;\hspace{0.5pt}@126473;\hspace{0.5pt}@141332;\ldots\\
1 
& 1 & --- & 1 & @80541 \\
2 
& --- & --- & 17 & @192792;\,@200160;\,@262904;\,@278938;\,@295524;\,@354531;\,@362607;\,\ldots \\
3 
& --- & --- & 22 & @192293;\,@224679;\,@232614;\,@256965;\,@319505;\,@343814;\,@379769;\,\ldots \\
4 
& --- & --- & 3 & @307344+2509+16960; @390407+5475+5338; @417153+93327+269 \\
5 
& --- & --- & 2 & @190889+20455+18557+101845; 
@494675+2461+2243+1580+1598 \\
6 
& --- & --- & 6 & @210789; @218380; @402328; @449864+1517; @489425; @513032+118 \\
7 
& --- & --- & 3 & @396919+10934+97+33; @413896+8771+471; @449200+1526+1495 \\
8 
& --- & --- & 1 & @212591+3132+10250+8774+4881+16950+22836+23846 \\
9 
& 1 & --- & 1 & @513286, @290742+19617+23389+179370+127+155+26+119 \\
10 
& --- & --- & 1 & @454843+3121+1179+21+4297+4055+4522+5810+4884+3061 \\
11 
 & --- & --- & 1 & @480823+4897+2271+632+1024+2+2888+5+2071+7+1708 \\
12 
& --- & --- & 2 & @225423+2979+18954+13+14299; @394748+4750+9417+1634+2+3 \\
13 
& --- & --- & 2 & @366534+2751+15215+65+9501; @403514+7934+4+2342+70+2850 \\
13 
& --- & 2 & 2 & @407707+9550, @339000+17927; @457817+4395, @453519+4228+74 \\
14 
& --- & --- & 3 & @475967+4318+3843; @482644+3411+2836; @509825+244+228+102 \\
15 
& --- & --- & 4 & @226645+164+31; @392050+10586; @462603+3916; @493897+2247 \\
15 
& --- & 2 & 3  & @440986+5523, @434758; @458372+4388, @436735; @458376+4388 \\ 
16 
& --- & --- & 1 & @442064+3988+63+1254+2+2+1599+1147+91+2+1610+2+10+1956 \\
18 
& --- & 2 & 1 & @340201+20570, @307727+8351+11333+5111+5+2+3+7653+69 \\
21 
& --- & --- & 2 & @432385+2588+2747+1189+25;\;@432387+5511+1014+28+4068+1565 \\
22 
& --- & --- & 1 & @423839+89559+77+43+22+28+28+41+36+11+39+28+3+21+19 \\
22 
& --- & 2 & 1 & @404780+10024, @332670+18758+3346+2+2+2+17519+42+117+35 \\
23 
& --- & --- & 2 & @238844+17038+3611+2+47+15926; @259293+19366+3824+3+49 \\
24 
& --- & --- & 1 & @400836+9150+31+1590+2+4+6356+2+2606+79+42+4793+61+766 \\
30 
& --- & --- & 4 & 
@360879+14073; @453685+3569+738; @478623+4147; @478625+4147 \\
31 
& --- & --- & 1 & 
@300575+20652+3873+3+60+16410+6330+121+139+171+7894 \\ 
36 
& --- & --- & 1 & 
@259618+22899+307+2+2+59+16814+1485+5476+86+47+134+49 \\
37 
& --- & --- & 3 & 
@245153+17810+3528; @388312+11048+1887+72; @447171+3856+60 \\
40 
& --- & --- & 1 & 
@236681+13614+3186+3501+6+43+8112+3911+6+51+3706+6821 \\ 
45 
& --- & --- & 1 & 
@381432+2937+8733+2143+2+2+39+59+2150+53+2+6156+3242 \\ 
48 
& --- & --- & 1 & 
@507937+284+65+215+77+5+46+157+71+79+197+61+34+4+13 \\ 
49 
& --- & 2 & 1 & 
@464371+5108, @444537+4301+742+4+4+4+4244+18+38+15+2893 \\ 
51 
& --- & --- & 1 & 
@290257+6427+12631+111+4271+74+59+3534+2+2+3826+2316+20 \\ %
55 
& --- & --- & 1 & 
@484802+3052+498+2+7+2163+774+18+11+10+1219+803+7+462 \\ %
71 
& --- & 10 & 1 & 
@416849+5727+1542+1102+2951+1953+970+26+3771+1918 \\ 
74 
& --- & --- & 1 & 
@369227+10605+2364+2281+167+2+2+34+4907+2052+248+2+2 \\ 
75 
& --- & --- & 1 & 
@259229+5888+10169+1870+2372+1728+133+1777+1812+3+1817 \\ %
75 
& --- & 6 & 1 & 
@464367+5108+31655+4+1867+4, @444535+4301+742+4+4+4 \\ 
81 
& --- & 1 & 1 & 
@513826, @134149+59277+13650+622+71691+12959+7588+99851  \\
87 
& --- & 4 & 1 & @495629+2341+87+2064, @475115+139+1859+119+2389+3255+155 \\ 
100 
& --- & --- & 1 & @395769+8963+962+48+1797+4+4+4395+1407+149+3+2+653+154 \\ 
101 
& 1 & 14 & 1 & @506098, @499550+2064+5+1481+277+9+1227+45+4+1120+186+5 \\
102 
& --- & --- & 1 & @305505+20711+4+31+3470+536+12+2+2+6+1470+485+10041 \\ 
107 
& --- & --- & 1 & @242891+607+12968+25+15+6429+211+754+126+2322+4053+3033 \\ 
113 
& --- & --- & 1 & @121526+56251+20114+58922+6929+12393+119347+7434+2087 \\ 
143 
& --- & --- & 1 & @376225+12071+69+2050+42+6827+1917+234+25+1796+1496+221 \\
168 
& --- & --- & 1 & @74633+22323+8290+119+28284+6670+6975+979+6825+12295 \\
258 
& --- & 6 & 1 & @312653+24127+261+8+18387+2907, @219963+12825+642+2526 \\
634 
& --- & --- & 1 & @351345+12959+2574+2111+216+300+174+8+36+2683+3201+2236 \\
\hline
3525:\!\! & 3 & 57 & 232 & :\;the total is summed up with multiplication by the 4th column
 \\ \hline
\end{tabular} 
\caption{Clusters of parallel patterns on the $10\times 10$ torus. 
In the last column, commas separate triples, quartets and pairs (in this order, generally) in the same cluster, 
while semicolons separate clusters. 
The number of parallel pairs is obtained by subtracting the 2nd and the 3rd columns from the first one.
} 
\label{tb:parallel10}
\end{table}

\begin{table} \small
\begin{tabular}{@{}lllclccl@{}}
\hline
The equations & \multicolumn{2}{l}{Number of terms\!} & \multicolumn{2}{l}{Simplifying terms\dotfill}  
& \!\!Obtained & \multicolumn{1}{l}{Gone} & Merging terms, \\
to subtract & \multicolumn{2}{l}{in the equations} & Identical\!\! & Equivalent\! 
& \multicolumn{1}{l}{\!\!terms} & terms & or Remarks
\\ \hline
\#50896+1 & $73\!-\!1$ & $73\!-\!1\;$ & 25 & 44; @210303 & 6 & 67 & @229500,\;@261895 \\ 
\#96995+1 & 60 & 60 & \s 3 & 55 & 4 & 56 & See \S\ref{sec:aequal}  \\ 
@413133 & $53\!-\!1$ & $53\!-\!1$ & 16 & 35; @420038 & $\quad\ 4\!-\!4$ & 53 & See \S\ref{sec:equal}\refpart{iii} \\  
\#324466+278 & 56 & 54 & 11 & 41 & 6 & 50 \\ 
\#222803+1 & $56\!-\!1$ & $56\!-\!1$ & \s 6 & 46; @312653 & 8 & 48 \\ 
@508221 & 48 & 48 & \s 1 & 46 & $\quad\ 2\!-\!2$ & 48 \\ 
\#342065+32477 &  57 & 54 & 10 & 40 & 11 \  & 46 \\ 
\#352611+55377 & 51 & 50 & \s 0 & 48 & 5 & 46 \\ 
\#127303+1 & 51 & 51 & 13 & 34 & 6 & 45 & \#440240,\;@441186 \\ 
\#339364+1 & 49 & 49 & \s 3 & 44 & 4 & 45 \\  
@243498 & 41 & 41 & \s 3 & 37 & $\quad\ 2\!-\!2$ & 41 \\ 
\#224637+1 & 44 & 41 & \s 5 & 34 & 7 & 37 \\ 
\#253539+11 & 44 & 41 & \s 2 & 37 & 7 & 37 \\ 
\#230485+45 & 43 & 41 & \s 2 & 37 &  6 & 37 \\ 
\#293983+1 & 41 & 41 & \s 7 & 32 & 4 & 37 \\ 
\#393996+5 & 40 & 39 & \s 1 & 37 & 3 & 37 \\ 
\#243844+44 & 43 & 44 & \s 2 & 37 & 9 & 35 \\ 
\#268255+518 & 43 & 42 & \s 2 & 36 & 8 & 35 & \#360394 \\ 
\#449521+1 & 39 & 37 & \s 1 & 35 & 4 & 35 \\ 
@402916 & 34 & 34 & \s 8 & 25 & $\quad\ 2\!-\!2$ & 34 \\ 
\#405498+1 & 39 & 40 & \s 6 & 29 & 8 & 32 & @486102 \\ 
\#364679+37711 & 36 & 36 & \s 0 & 34 & 4 & 32 \\ 
\#487622+3130 & $35\!-\!2$ & $34\!-\!2$ & \s 1 & 32; remark: & 3 & 32 & : @495629+2428 \\ 
\#339626+373 & 40 & 37 & \s 9 & 25 & 9 & 31 \\ 
\#211285+1 & 37 & 35 & \s 1 & 32 & 6 & 31 \\ 
\#277050+1 & 37 & 35 & \s 1 & 32 & 6 & 31 \\ 
\#240215+1 & $37\!-\!1$ & $37\!-\!1$ & \s 0 & 34; @337041 & 6 & 31 \\ 
\#213405+1 & $35\!-\!1$ & $35\!-\!1$ & \s 6 & 27; @504613 & 4 & 31 \\ 
\#489303+367 & $35\!-\!2$ & $35\!-\!2$ & \s 1 & 32; remark: & 4 & 31 & : @497970+2151 \\ 
\#355371+28531 & 37 & 36 & \s 8 & 25 & 7 & 30 \\ 
\#138740+172531 & 36 & 33 & \s 1 & 30 & 6 & 30 & @300575 \\ 
\#510125+1 & 33 & 33 & \s 2 & 29 & 3 & 30 & @511356 \\ 
$\#775+4370$ & 43 & 48 & 10 & 26 & 19 \ & 29 \\ 
\#353186+1 & 34 & 34 & \s 2 & 29 & 5 & 29 & @387809 \\ 
\#231183+152362 & 33 & 33 & \s 6 & 25 & 4 & 29 \\ 
\#240806+1 & $34\!-\!1$ & $34\!-\!1$ & \s 0 & 31; @337049 & 6 & 28 \\ 
\#278921+192 & $35\!-\!2$ & $34\!-\!2$ & \s 1 & 30; remark: & 7 & 28 & : @312653+42783 \\ 
@263146 & 28 & 28 & \s 0 & 27 & $\quad\ 2\!-\!2$ & 28 \\ 
\#403976+30634 & 42 & 35 & \s 1 & 30 & 15 \ & 27 \\ 
\#403975+30634 & 41 & 35 & \s 1 & 30 & 14 \ & 27 \\ 
\#243552+98038 & 35 & 34 & \s 2 & 28 & 8 & 27 & @360879 \\ 
\#406247+18455 & 40 & 34 & \s 1 & 29 & 14 \ & 26 \\ 
\#406248+18453 & 39 & 34 & \s 1 & 29 & 13 \ & 26 \\ 
\#436805+756 & 31 & 30 & \s 1 & 26 & 6 & 25 & @434340 \\ 
\#423451+23505 & 31 & 30 & \s 1 & 26 & 6 & 25 & @448869 \\ 
@418249 & 25 & 25 & \s 7 & 16 & $\quad\ 4\!-\!4$ & 25 & See \S\ref{sec:equal}\refpart{iii} \\  
\hline
\end{tabular}
\caption{Simplifications of equations (with at most 340 terms) by subtracting them from each other.
The first column 
identifies two equations to subtract following (\ref{eq:pshifts}) or by a bunch @-id.
Arithmetics the 2nd and 3rd columns 
indicates merging of terms under easy $Eq$-equalities (identified by the bunch @-ids in the 5th or the last columns).
The 4th and 5th columns count simplification of identical or equivalent (under bunch equalities) terms separately.
The 6th column counts the terms after simplifying the equations difference; arithmetics here 
indicates the final bunch equality of the respective eigenvalue components. 
Presence of common but non-simplifying terms causes discrepancy between the maximum of the 2nd and 3rd columns 
and the sum of the 6th and 7th columns; those terms are identified in the last column (if without :\;).
 } 
\label{tb:simplify}
\end{table}

\subsection{Row simplifications on the $10\times10$ torus}
\label{sec:simplified10}

There are many pairs of rows of the transition matrix that are equal in nearly all entries. 
For example, here are two implied eigenvector equations:
\begin{align} \label{eq:simpl2a}
(\lambda+98)\,v_{227062} & = 3v_{214437}+2v_{215722}+3v_{237370}+4v_{237374}, \\
 \label{eq:simpl2b}
(\lambda+100)\,v_{227159} & = 3v_{214437}+2v_{215722}+3v_{237370}+4v_{237374}.
\end{align}
Their difference gives a two-term relation between $v_{227062}$ and  $v_{227159}$.
Similarly, both $(\lambda+81)\,v_{250802}$ and $ (\lambda+\frac{163}2)\,v_{250835}\,$ are equal to
\begin{align}
\textstyle v_{232294}+\frac12v_{242413}+\frac12v_{246789}+\frac12v_{498714}.
\end{align}
The cases of (almost) equal eigenvalue components in Sections \ref{sec:equal} and \ref{sec:aequal} 
provide a bounty of these examples, and the number of possible substantial simplifications increases
when the transition matrix is contracted by collapsing the bunches (with equal eigenvalue components)
into single representatives. The most drastic simplifications are listed in Appendix Table \ref{tb:simplify}. 

Let us refer to the eigenvector equations by the \#-ranks of the patterns represented by the diagonal terms (with $\lambda$). 
Considering differences of equations of length at most 340, there are over 27000 possible reductions to fewer terms.
Comparably, there are over 25000 possible reductions by taking more general linear combinations of two equations.
For example, the equation \#31600 
has the terms
\begin{equation}
12v_{376473}+3v_{364258}+3v_{369850}+3v_{371319}+3v_{387481}+3v_{389673},
\end{equation}
and these terms simplify or merge after combining with the equation
\begin{equation}
(\lambda+87)v_{376473}=2v_{364258}+2v_{369850}+2v_{371319}+2v_{387481}+2v_{389673}.
\end{equation}
The length of the equation \#31600 
then reduces from 335 to 329,  thanks additionally to merging terms of the pair @450298. 
On the other hand, these $>52000$ simplifications cannot be applied all together, 
as they may interfere in targeting the same equations for simplification.
A straightforward selection lead to about 15300+9800 non-interfering simplifications by taking,
respectively, differences and more general linear combinations of two equations.
Most simplifications (about 10000, 2900, 1100, respectively) reduce the number of terms just by 1, 2 or 3.

One may aggressively use (original or derived) two-term equations to eliminate many eigenvector components.
Initially, there are 3625 equations with two terms; 
635 of them correspond to patterns in bunches.
There are 3306 two-term equations after collapsing the bunches.
Additionally, there are 180 simplifications to two-terms like in (\ref{eq:simpl2a})--(\ref{eq:simpl2b}).
Iterative elimination using these and emerging two-term equations leads to 2436 further eliminations.
It could be reasonable to subsequently eliminate using three-term equations (counting 12525 at this stage), etc.

If the transition matrix is contracted by collapsing bunches (with equal eigenvalue components)
into single representatives, its size is reduced by $3525+3+2\cdot57=3642$ in accordance
with the last row counts in Table \ref{tb:parallel10}. A distinctive medium size equation is \#11712.
It has 342 terms. Bunch contraction merges 29 pairs of terms, and additional two terms can be simplified 
by subtracting either \#126663 or \#131244. 
Other interesting equation is \#28040.
Subtracting 28 other equations from it can reduce its length 237 by 77.
The initial number 16734519 of non-zero entries in the full transition matrix
is decreased by the described simplifications as follows:
\begin{itemize}
\item The bunch collapse diminishes by 76377.
\item Ignoring the downstream pattern \#14342 diminishes by 101.
\item The chosen simplifications by 
linear combinations  of two equations (of length at most $340$) diminishes by 39748.
\item Subsequent iterative elimination by using two-term equations diminishes by 90574.
\end{itemize}
The number of non-zero terms is thereby decreased just by 1.24\%. 

The longest equations in (\ref{eq:toplength}) loose many terms thanks to merely collapsing the bunches.
The five densest rows loose then over 3500 non-zero entries each.
Several reductions of over 30000 terms are possible by taking simple linear combinations of the top rows.
For example, the simplified difference of the equations \#3 and \#4 has 466711 terms; compare with (\ref{eq:toplength}).
Further laborious simplification with shorter equations might be even more significant. 
The following observation looks more practical: the seven variables $v_1,\ldots,v_6$ and $v_9$ appear only in the top 11 equations.
As one dense equation can be ignored due to the overdetermination of the eigenvector system,
elimination of these variables leaves just three very dense equations (for iterative computation of Section \ref{sec:laiterate}, say).

The considered simplifications offer some potential for computing the leading eigenvector directly once the eigenvalue is known. 
But this simplification scheme has these substantial downsides:
\begin{itemize}
\item Except for the diagonal entries, the initial matrix for the eigenvector system has rational numbers as entries, mostly integers.
The described simplifications proliferate either the symbolic or high precision real $\lambda_1$ throughout the modified matrix.
\item The 
density contrast between the upper-triangular and lower-triangular parts in the matrix 
may be weakened due to these operations. This would probably diminish the convergence rate of concluding iterations.
\end{itemize}
It could be interesting to investigate how much the observed equalities of eigenvector components or equation simplifications
carry over to the next larger torus, as straightforward pairwise comparison of equations takes increasingly more time.


\small 
\bibliographystyle{alpha}

\end{document}